\documentclass{article}

\usepackage{microtype}
\usepackage{graphicx}
\usepackage{subcaption}
\usepackage{booktabs} % for professional tables
\usepackage{subfiles}

\usepackage{tabularx}
\newcolumntype{R}{>{\raggedleft\arraybackslash}X}
\newcolumntype{C}{>{\centering\arraybackslash}X}

\usepackage{hyperref}

\usepackage[accepted]{icml2026}

\usepackage{amsmath}
\usepackage{amssymb}
\usepackage{mathtools}
\usepackage{amsthm}

\usepackage[capitalize,noabbrev]{cleveref}

\usepackage{thmtools}
\usepackage{thm-restate}

\theoremstyle{plain}
\newtheorem{theorem}{Theorem}[section]
\newtheorem{proposition}[theorem]{Proposition}
\newtheorem{lemma}[theorem]{Lemma}

\theoremstyle{definition}

\theoremstyle{remark}
\newtheorem{remark}[theorem]{Remark}

\usepackage{commands}
\newcommand{\emphasize}[1]{\textit{#1}}

\usepackage{xcolor}
\usepackage[textsize=tiny]{todonotes}

\icmltitlerunning{Gromov-Wasserstein at Scale, Beyond Squared Norms}

\begin{document}

% \nocite{*} %TMP

\twocolumn[

\icmltitle{Gromov-Wasserstein at Scale, Beyond Squared Norms}



  \icmlsetsymbol{equal}{*}

  \begin{icmlauthorlist}
    \icmlauthor{Guillaume Houry}{heka}
    \icmlauthor{Jean Feydy}{heka}
    \icmlauthor{François-Xavier Vialard}{ligm}
  \end{icmlauthorlist}

  \icmlaffiliation{heka}{Inria, Université Paris Cité, Inserm, HeKA team, F-75015 Paris, France}
  \icmlaffiliation{ligm}{LIGM, Université Gustave Eiffel, Marne-la-Vallée, France}

  \icmlcorrespondingauthor{Guillaume Houry}{guillaume.houry@inria.fr}

  \icmlkeywords{Point Set Registration, Optimal Transport, Gromov-Wasserstein}

  \vskip 0.3in
]



\printAffiliationsAndNotice{}  % no special notice (required even if empty)

\begin{abstract}
A fundamental challenge in data science is to match disparate point sets with each other. While optimal transport efficiently minimizes point \emph{displacements} under a bijectivity constraint, it is inherently sensitive to rotations. Conversely, minimizing \emph{distortions} via the Gromov-Wasserstein (GW) framework addresses this limitation but introduces a non-convex, computationally demanding optimization problem. In this work, we identify a broad class of distortion penalties that reduce to a simple alignment problem within a lifted feature space. Leveraging this insight, we introduce an iterative GW solver with a linear memory footprint and quadratic (rather than cubic) time complexity. Our method is differentiable, comes with strong theoretical guarantees, and scales to hundreds of thousands of points in minutes. This efficiency unlocks a wide range of geometric applications and enables the exploration of the GW energy landscape, whose local minima encode the symmetries of the matching problem.

\end{abstract}

\section{Introduction}

\paragraph{Point Set Registration.}
From 3D point clouds to imaging datasets,
many objects are best described as collections of
samples in a vector space. 
Establishing correspondences between these sets is an important task that facilitates critical downstream operations such as cross-domain information transfer, domain adaptation, and the formulation of geometric loss functions for shape registration and generative modeling \cite{flamary_domain_adaptation, genevay2018learning}.

Taking as input a source and a target distribution, defined up to permutations of the samples, we look for a coupling that matches corresponding points and structures with each other.
When both distributions are embedded in the same vector space, endowed with a domain-specific metric, a first idea is to match points from the source distribution with their closest neighbors in the target set.
Accordingly, correspondences between 3D point clouds may be computed using the standard Euclidean metric on plain $xyz$ coordinates
or custom features that leverage local shape contexts \cite{ICP, FPFH}.

\begin{figure*}[t]
\centering
\includegraphics[width=\linewidth]{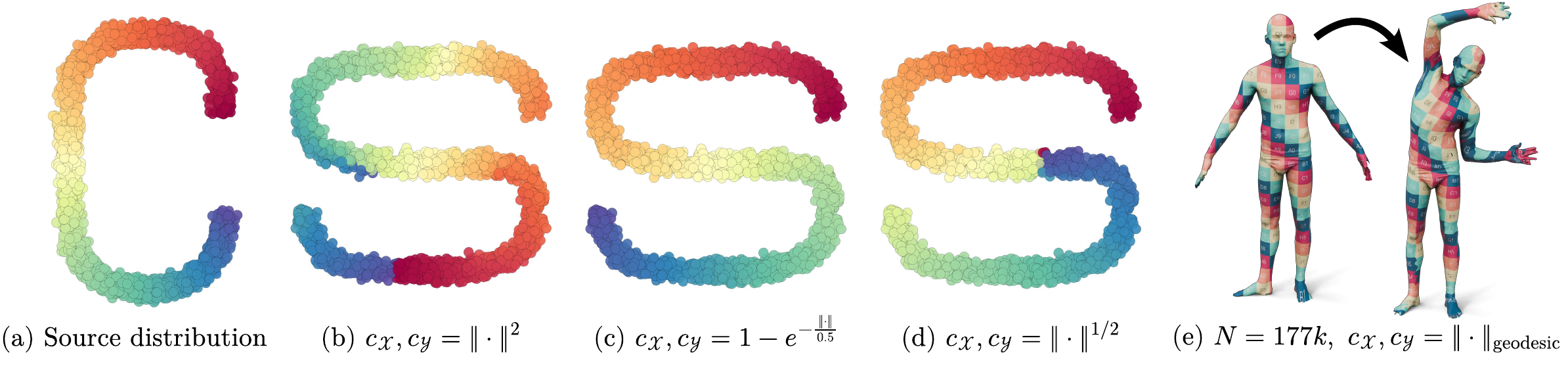}
\caption{(a) We optimize the GW objective of Eq.~\eqref{eq:GW_loss} to match a source distribution of points in the unit square (``C'') with a target (``S''). Colors let us visualize the destination of every source point.
(b)~The preservation of squared distances between points
has been studied extensively, but prioritizes the alignment
of principal axes over smoothness.
We propose a scalable GW solver for a broad class of penalties
that may promote the preservation of topology~(c) 
or find a balance between local and global structure~(d).
(e)~This opens the door to applications on high-resolution data, such as those silhouettes sampled with 177k points each.
}
\label{fig:introduction}
\end{figure*}

\paragraph{Optimal Transport (OT).}
To avoid degenerate couplings that match all source points with a single target,
a common approach is to look for a \emph{bijective} assignment.
Under this constraint, minimizing the average distance between corresponding points leads to a linear optimization problem
known as the Earth Mover, Monge-Kantorovitch, Linear Assignment or Optimal Transport (OT) problem in different communities.
As detailed in \cite{peyre2019computational}, measure theory generalizes this framework to clouds that contain different numbers of points, weighted samples or continuous distributions. Bijectivity constraints can then be enforced exactly, or relaxed to account for outliers and variations of the point sampling density \cite{unbalanced_ot}.

While exact solvers for the OT problem struggle to scale beyond a few thousand points in difficult cases, iterative methods have been designed to handle large distributions at a small cost in the accuracy of the coupling.
We may cite the auction algorithm which remains ideally suited to sparse problems in operations research \cite{auctions} or semi-discrete solvers which have become invaluable for physics simulations \cite{levy_semidiscrete, merigot_semidiscrete}.
In machine learning research, solvers based on entropic regularization have emerged as the standard approach for computing OT couplings between large distributions sampled in vector spaces.
These methods are based on variations of the Sinkhorn algorithm,
which streams well on massively parallel hardware and provides smooth gradients
\cite{sinkhorn1967diagonal,kosowsky1994invisible,cuturi2013sinkhorn}.
Modern implementations leverage annealing and multiscale heuristics \cite{gerber2017multiscale,schmitzer2019stabilized,feydy2020analyse},
scaling to millions of points to
unlock a wide range of applications \cite{schiebinger2019optimal, shen2021accurate, qu2023power}. %, buze2024anisotropic}.

\paragraph{Gromov-Wasserstein (GW).}
In spite of these achievements, OT still suffers from two fundamental limitations:
it requires a common embedding or at least a metric between the source and target samples,
which is problematic for multimodal applications;
and it is not robust to large non-uniform shifts such as global rotations.

To address these issues, an appealing strategy is to look for alignments that minimize the \emph{distortion} induced by the coupling instead of the \emph{displacement} of the samples.
When framed in the language of measure theory, this perspective leads to the Gromov-Wasserstein (GW) optimization problem: a versatile framework capable of matching diverse data types, from graphs to continuous probability distributions \cite{memoli2011gromov}.
By focusing on intrinsic structure preservation instead of extrinsic feature engineering, GW has rapidly emerged as a convenient tool for graph machine learning \cite{chowdhury2019gromov,xu2019scalable,brogat2022learning}, image processing \cite{thual2022aligning,takeda2025unsupervised,beier2025tangential,salmona2023gromov,zhang2021deepacg} and geometric data analysis \cite{chowdhury2021generalized,clark2025generalized}.

%\paragraph{Entropic Regularization (EGW).}
Despite this growing popularity, GW remains hard to solve.
It is equivalent to the quadratic assignment problem,
which implies that finding a global optimum is NP-hard in the general case
while standard local minimization methods are prohibitively expensive in time and memory \cite{titouan2019optimal}.
As in the linear OT case, this context motivates the development of solvers based on the Sinkhorn iterations for an Entropy-regularized Gromov-Wasserstein (EGW) problem \cite{peyre2016gromov,solomon2016entropic,xu2019gromov}.
This technique is central to all state-of-the-art GW solvers \cite{kerdoncuff2021sampled,scetbon2022linear,li2023efficient}
but comes at a cost:
the regularization creates geometric biases that degrade both the quality of the optimal matching and the convergence speed of EGW solvers.

%\paragraph{Distortion Penalties.}
Meanwhile, the properties of an optimal GW coupling heavily depend on
the objective used to penalize distortions.
%making CNT methods compatible with euclidean approximations of geodesic distances commonly used in the computer graphics community 
%\cite{nadler2005diffusion,panozzo2013weighted,lipman2010biharmonic}.
%
Recent works have largely focused on the preservation of squared Euclidean distances between pairs of points, thanks to algebraic identities that simplify the optimization landscape \cite{vayer2020contribution,delon2022gromov,wang2023neural,zhang2024gromov,dumont2025existence};
this was recently exploited by \cite{rioux2024entropic} to derive a robust dual solver for EGW in this setting.
Squared-norm EGW can also be seen as a concave-penalized OT problem, an algorithmic framework that has been thoroughly investigated in \cite{sebbouh2024structured}.
However, as illustrated in Fig.~\ref{fig:introduction}b, squared norm penalties are biased towards long-range interactions and are thus sensitive to outliers, which limits the applicability of this line of work to real-world problems.

\paragraph{Contributions.} As a middle ground between general distortion penalties (that lead to complex optimization landscapes) and the preservation of squared distances between points (which is too brittle for most applications) we study the preservation of pairwise quantities which are \emph{Conditionally of Negative Type} (CNT) \cite{maron2018probably,sejourne2022generalized}.
As illustrated in Fig.~\ref{fig:introduction}, this class of GW problems is large enough to encode a broad range of behaviors, with attention paid to the preservation of long- or short-range structures depending on the target application.

We present three main contributions. First, we demonstrate that the optimal assignment in this setting decomposes into a linear map between feature spaces, followed by a standard OT coupling under the squared Euclidean distance. This structural insight reveals that the non-convex GW landscape is parameterized by a transformation matrix with dimensions determined solely by the preserved quantities. Second, we prove that the regularization biases typical of EGW solvers are either absent or easily mitigated within the CNT framework. Finally, we leverage these findings to develop a versatile, robust, and differentiable EGW solver that consistently outperforms state-of-the-art methods.
Our code is publicly available at \url{github.com/guillaumeHoury/egw-solvers}.

\section{Background and Notations}

\paragraph{Optimal Transport Plans.} 

We now introduce our main definitions
and refer to \cite{peyre2019computational}
for additional background or intuitions.
Let the source $\alpha \in \prob{\X}$ and the target $\beta \in \prob{\Y}$ be probability measures with compact support over two topological spaces $\X$ and $\Y$.
Let $\coupling{\alpha}{\beta} \subset \prob{\X \times \Y}$ be the set of \emphasize{couplings} (or \emphasize{transport plans}) between $\alpha$ and $\beta$, i.e. the set of probability measures $\pi$ over $\X \times \Y$ that satisfy the marginal constraints:
\begin{equation*}
     \int \diff\pi(\cdot,y) = \alpha \quad \text{ and } \quad \int\diff\pi(x,\cdot) = \beta~.
\end{equation*} 

Given a measurable cost function $\func{c}{\X \times \Y}{\R}$, the \emphasize{optimal transport} (OT) problem between $\alpha$ and $\beta$ reads:
\begin{equation*}
   \OT(\alpha, \beta) := \min_{\pi \in \coupling{\alpha}{\beta}} \int c(x,y)\, \diff \pi(x,y)~.
\end{equation*}
When the measures $\alpha$ and $\beta$ are supported on at most $\NN$ points, this linear program can be solved in $\mathcal{O}(\NN^3)$ time. 

\paragraph{The Sinkhorn Algorithm.} 
A common approach to accelerate optimal transport solvers is to add an entropic penalty to the objective.
For any positive temperature $\varepsilon > 0$, the \emphasize{entropic optimal transport} (EOT) problem reads:
\begin{multline*}
\OT_{\varepsilon}(\alpha, \beta) := \min_{\pi \in \coupling{\alpha}{\beta}} \int c \cdot \diff \pi + \varepsilon \KL{\pi}{\alpha \otimes \beta} \\
\text{where }\quad\KL{\pi}{\alpha \otimes \beta} = \int \log \left(\frac{\diff \pi}{\diff \alpha \diff \beta} \right) \diff \pi~.
\end{multline*}
The \emphasize{Sinkhorn algorithm} reduces this problem to a maximization over the set of pairs of real-valued, continuous functions $(f,g) \in \Cont{\X} \times \Cont{\Y}$.
%It alternatively updates:
%\begin{equation*}
%    f_{s+1} = \Sink{\varepsilon,\beta}(g_s) \quad \text{ and } \quad g_{s+1} = \Sink{\varepsilon,\alpha}(f_{s+1})~,
%\end{equation*}
%where $\Sink{\varepsilon,\beta}$ and $\Sink{\varepsilon,\alpha}$ are the Sinkhorn operators.
We provide details in \cref{appendix:sinkhorn},
and note that the optimal EOT coupling $\pi^*$ can then be reconstructed from the optimal pair $(f^*, g^*)$ with:
\begin{equation}
    \frac{\diff\pi^*}{\diff \alpha \diff \beta}(x, y) = \exp \left( \frac{f^*(x) + g^*(y) - c(x,y)}{\varepsilon} \right)~.
    \label{eq:sinkhorn_pi}
\end{equation}

\paragraph{Gromov-Wasserstein (GW).}
As discussed in the introduction, penalizing distortions instead of displacements is desirable in many applications.
Given two \emphasize{base costs} $\func{c_\X}{\X \times \X}{\R}$ and $\func{c_\Y}{\Y \times \Y}{\R}$, the \emphasize{Gromov-Wasserstein} (GW) problem thus reads: 
\begin{gather}%\begin{multline*}
\GW(\alpha, \beta) := \min_{\pi \in \coupling{\alpha}{\beta}} \Loss(\pi)~, \quad \text{where} \label{eq:GW_loss} \\
\Loss(\pi) :=\!\int \left( c_\X(x,x')-c_\Y(y,y') \right)^2 \diff\pi(x,y) \diff\pi(x',y')~. \nonumber
\end{gather}%\end{multline*}

Under mild assumptions, $\GW$ quantifies how far distributions are from being isometric to each other:
\begin{restatable}{proposition}{gwisometry}
\label{prop:gwisometry}
Let $c_\X$ and $c_\Y$ be two symmetric functions with non-negative values such that: 
\begin{equation*}
    c_\X(x,x')=0 \Leftrightarrow x = x' \quad \!\text{and}\! \quad c_\Y(y,y')=0 \Leftrightarrow y = y'~.
\end{equation*}
Then, $GW(\alpha, \beta) = 0$ if and only if $\alpha$ and $\beta$ are isometric, i.e. 
there exists an application $\func{I}{\X}{\Y}$ that pushes
$\alpha$ onto $\beta$ such that for all $x$ and $x'$ in the support of $\alpha$:
\begin{equation*}
    c_\X(x,x') = c_\Y(I(x),I(x'))~.
\end{equation*}
\end{restatable}

\paragraph{Entropic Regularization.} As in the OT case, the \emphasize{Entropic Gromov-Wasserstein (EGW)} problem corresponds to:
\begin{equation}
\label{eq:entropic_gromov_wasserstein}
\GW_{\varepsilon}(\alpha, \beta) := \min_{\pi \in \coupling{\alpha}{\beta}} \Loss(\pi) + \varepsilon \KL{\pi}{\alpha \otimes \beta}.
\end{equation}
This regularization enables the use of scalable methods based on the Sinkhorn algorithm, but also introduces numerical biases and instabilities that we discuss in \cref{subsection:appendix:solver_cvg}.

\begin{figure*}[t]
    \centering
    \includegraphics[height=4cm]{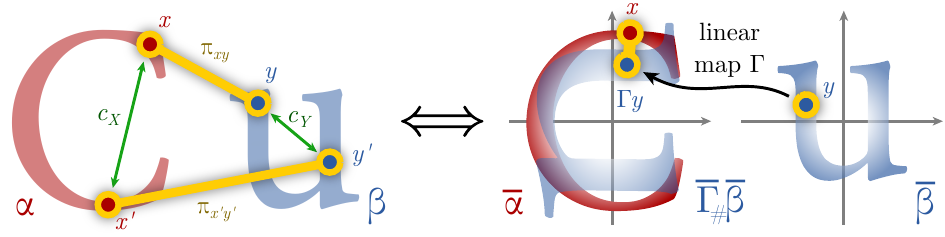}
    \caption{(left) In order to match a source distribution $\alpha$ (red ``C'') with a target distribution $\beta$ (blue ``u''), the GW objective
    of Eq.~\eqref{eq:GW_loss} penalizes distortions
    between corresponding pairs of points.
    (right) \cref{corrolary:egw_icp_pi} shows that up to known embeddings in Hilbert spaces and the optimization of a linear alignment $\Gamma$,
    we can reduce this problem to the computation of an OT 
    plan for the squared norm cost.
    }
    \label{fig:main_theorem}
\end{figure*}

\paragraph{Embeddable Costs.}

A non-negative, continuous cost $\func{c}{\X \times \X}{\R^+}$ is 
\emphasize{conditionally of negative type} (CNT) if it is symmetric, satisfies $c(x, x) = 0$ for all $x \in \X$\footnote{This condition can be relaxed by considering the cost $\tilde c(x,y) = c(x,y) - \frac{1}{2}(c(x,x) + c(y,y))$.} and 
is such that for any finite collection of points $x_1,\dots, x_n \in \X$ and coefficients $\lambda_1,\dots,\lambda_n \in \R$:
\begin{equation*}
  \sum_{i=1}^n \lambda_i = 0 ~~\Longrightarrow~~ \sum_{i=1}^n \lambda_i \lambda_j \cdot c(x_i, x_j) \leq 0~.
\end{equation*}
Moreover, we say that a cost $c$ is \emph{definite} CNT if $c(x,x')=0$
if and only if $x = x'$. 
We refer to \citep[Appendix C]{bekka2007kazhdan} for an introduction to this class of pairwise functions that includes tree distances, hyperbolic geodesic distances, spherical distances
and all $p$-powered Euclidean distances $c(x,x') = \norm{x-x'}^p$
for $0 < p \leq 2$.

CNT costs are relevant in machine learning because they are to squared Euclidean norms what reproducing kernels are to scalar products \cite{berlinet2011reproducing}.
This is made clear in the following characterizations:
\begin{theorem}[\citet{schoenberg1938metric}]
\label{theorem:cnt_embedding}
Let $c$ be a symmetric function on $\X \times \X$ satisfying $c(x,x) = 0$. The following statements are equivalent:
($i$) the cost $c$ is CNT,  ($ii$) there exists a real Hilbert space $\Hilb$ and a continuous mapping $\func{\varphi}{\X}{\Hilb}$ such that
for all  $x$ and $x'$ in $\X$,
\begin{equation}
    \label{eq:schoenberg_embedding}
     \quad c(x, x') ~=~ \norm{\varphi(x) - \varphi(x')}_{\Hilb}^2~, \text{~and}
\end{equation}
($iii$) for every $\lambda > 0$, $e^{-\lambda c(\cdot,\cdot)}$ is a positive definite kernel.
\end{theorem}
Note that the embedding $\varphi$ is not uniquely defined, but is necessarily injective if
$c$ is a definite CNT cost.
\begin{proposition}[\citet{schoenberg1938metric}] \label{prop:cnt_positive_kernel}
Let $x_0\!\in \!\X$. A cost $c$ is CNT if and only if it induces a positive definite kernel:
\begin{equation*}
    k(x,x') ~=~ \big( c(x, x_0) + c(x_0, x') - c(x, x') \big) \,/\, 2~.
\end{equation*}
\end{proposition}

\paragraph{Hilbert-Schmidt Operators.} A linear map between Hilbert spaces $\func{\Gamma}{\Hilb_\Y}{\Hilb_\X}$ is a \emph{Hilbert-Schmidt (HS) operator} if it is the limit of a sum of rank-$1$ operators: 
\begin{equation*}
    u v^{\perp}: w \longmapsto \scal{v}{w} u \quad
    \text{for } u \in \Hilb_\X \text{ and } v \in \Hilb_\Y.
\end{equation*} 
These operators form a Hilbert space $\Hspace = \HS(\Hilb_\Y, \Hilb_\X)$ that generalizes the notion of matrices to infinite dimensions: in particular, when $\Hilb_\X$ and $\Hilb_\Y$ are Euclidean spaces, the HS norm $\norm{\cdot}_{\HS}$ is identical to the Frobenius norm.

%\paragraph{Hilbert-Schmidt Operators.} Let $\Hilb_\X$ and $\Hilb_\Y$ be two Hilbert spaces.
%A linear map $\func{\Gamma}{\Hilb_\X}{\Hilb_\Y}$ is a \emph{Hilbert-Schmidt (HS) operator} if 
%\begin{equation*}
%   \norm{\Gamma}_{\HS}^2 := \sum_i \norm{\Gamma e_i}_{\Hilb_Y}^2 < + \infty 
%\end{equation*} for any orthonormal basis $(e_i)$ of $\Hilb_\X$.
%These operators form a Hilbert space, $\HS(\Hilb_\X, \Hilb_\Y)$, whose scalar product induces the norm $\norm{\cdot}_{\HS}$.
%It is generated by rank-$1$ operators: 
%\begin{equation*}
%    \text{For } x \in \Hilb_\X \text{ and } y \in \Hilb_\Y, \quad y \otimes x^*: z \longmapsto \scal{x}{z} y.
%\end{equation*} 
%Intuitively, HS operators generalize matrices to infinite-dimension: in the finite case, $\HS(\R^\DD, \R^\EE)$ coincides with the set of matrices $\R^{\DD \times \EE}$ and the HS norm is identical to the Frobenius norm.
\label{section:background}

\section{Gromov-Wasserstein with CNT Costs}
\label{section:gw_cntcosts}

\paragraph{GW features.} We introduce a non-linear embedding that best encodes the geometry of the GW problem.

\begin{restatable}[GW-embeddings]{definition}{gwembeddings}
\label{def:gw_embeddings}
Let $\alpha$ be a probability measure on $\X$,
$\func{c_\X}{\X\times \X}{\R^+}$ be a definite CNT cost
and $\func{\varphi}{\X}{\Hilb_\X}$ be an injective embedding
provided by \cref{theorem:cnt_embedding} to represent $c_\X$.
Since Eq.~\eqref{eq:schoenberg_embedding} is invariant by translations,
we can assume that $\int \varphi(x) \,\diff \alpha(x)~=~ 0_{\Hilb_\X}$.
Then, we say that the injective map:
\begin{equation*}
    \Phi : x\in \X \mapsto 
    \big( \, \varphi(x), \tfrac{1}{2}\|\varphi(x)\|^2 \, \big) \in \Hilb_\X \times \R
\end{equation*}
is a GW-embedding of the distribution $\alpha$ with respect to $c_\X$.
Similarly, we extend any Hilbert-Schmidt operator $\Gamma \in \Hspace$ to our product feature space with:
\begin{equation*}
   \overline{\Gamma}:
    (y, s) \in \Hilb_\Y \times \R 
    \mapsto 
    (\Gamma y, s) \in \Hilb_\X \times \R ~.
\end{equation*}
\end{restatable}

\paragraph{Example.}
When $c_\X$ and $c_\Y$ are squared Euclidean norms
on the finite-dimensional spaces $\X=\R^\mathrm{D}$ and $\Y=\R^\mathrm{E}$,
the CNT embedding spaces coincide with the base spaces
as $\Hilb_\X = \R^\mathrm{D}$ and $\Hilb_\Y = \R^\mathrm{E}$.
The GW-embeddings of two distributions
$\alpha \in \prob{\X}$
and
$\beta \in \prob{\Y}$ read:
\begin{align*}
    \Phi:x\in \R^\mathrm{D} &\mapsto \big(
    x - x_\alpha, \tfrac{1}{2}\|x - x_\alpha\|^2
    \big) \in \R^\mathrm{D+1}~~\text{and} \\
    \Psi:y\in \R^\mathrm{E} &\mapsto \big(
    y - y_\beta, \,\tfrac{1}{2}\|y - y_\beta\|^2\,
    \big) \in \R^\mathrm{E+1}~,
\end{align*}
where
$x_\alpha = \int z \,\diff \alpha(z)$
and
$y_\beta = \int z \,\diff \beta(z)$
denote the centers of mass of both distributions.
We can identify any linear map
$\func{\Gamma}{\Hilb_\Y}{\Hilb_\X}$
with a $\mathrm{D}\times\mathrm{E}$
matrix and write:
\begin{equation*}
    \overline{\Gamma}~=~
    \left(
    \begin{array}{c|c}
            \Gamma & 0 \\
            \hline
            0 & 1
    \end{array}
    \right)~\text{as a $(\mathrm{D}+1)\times(\mathrm{E}+1)$
matrix.}
\end{equation*}

\paragraph{Dual Formulation.}
We can now state our main theorem:

\begin{restatable}{theorem}{entropicgwcntlinearized}
\label{prop:entropic_gw_cnt_linearized}
Let $\alpha \in \prob{\X}$
and
$\beta \in \prob{\Y}$ 
be two probability distributions with compact supports on topological spaces $\X$ and $\Y$,
endowed with definite CNT costs $c_\X$ and $c_\Y$.
Let 
$\Phi(x)=(\varphi(x), \tfrac{1}{2}\|\varphi(x)\|^2)$
and
$\Psi(y)=(\psi(y), \tfrac{1}{2}\|\psi(y)\|^2)$
denote their respective GW-embeddings, as in \cref{def:gw_embeddings}.
Then, for any temperature $\varepsilon \geq 0$,
the EGW problem of Eq.~\eqref{eq:entropic_gromov_wasserstein} is equivalent to:
\begin{equation*}
     GW_{\varepsilon}(\alpha, \beta) =  C(\alpha, \beta) +  8 
     \min_{\Gamma \in \Hspace} 
     \min_{\pi \in \coupling{\alpha}{\beta}}
     \mathcal{F}(\Gamma, \pi)~,
\end{equation*}
where $C(\alpha, \beta)$ is an additive constant and:
\begin{align*}
\mathcal{F}(\Gamma, \pi)
~:=&~
\norm{\Gamma}_{\HS}^2
~+~
(\varepsilon/8)\cdot \KL{\pi}{\alpha \otimes \beta} \\
-&~2
\int \scal{\overline{\Gamma}}{\Phi(x) \Psi(y)^{\top}}_{\HS} \diff \pi(x,y)~.
\end{align*}
\end{restatable}

Our proof builds upon similar dual formulations proposed recently for squared Euclidean costs in finite dimensions \cite{rioux2024entropic,zhang2024gromov,sebbouh2024structured,pal2025wasserstein}.
The objective function 
$\mathcal{F}(\Gamma, \pi)$
%$\func{\mathcal{F}(\Gamma, \pi)}{\Hspace \times \coupling{\alpha}{\beta}}{\R}$ 
is not \emph{jointly} 
convex, but has the following properties:
\begin{restatable}{theorem}{egwicpgamma}
\label{corrolary:egw_icp_gamma}
$\mathcal{F}$ is convex with respect to $\Gamma$ and:
\begin{align*}
\Gamma^\star(\pi) := \argmin_{\Gamma \in \Hspace} 
     \mathcal{F}(\Gamma, \pi)
     = \int \varphi(x) \psi(y)^{\top} \,\diff \pi(x,y)~.
\end{align*}
\end{restatable}

\begin{figure*}[t]
    \centering
    \includegraphics[width=\linewidth]{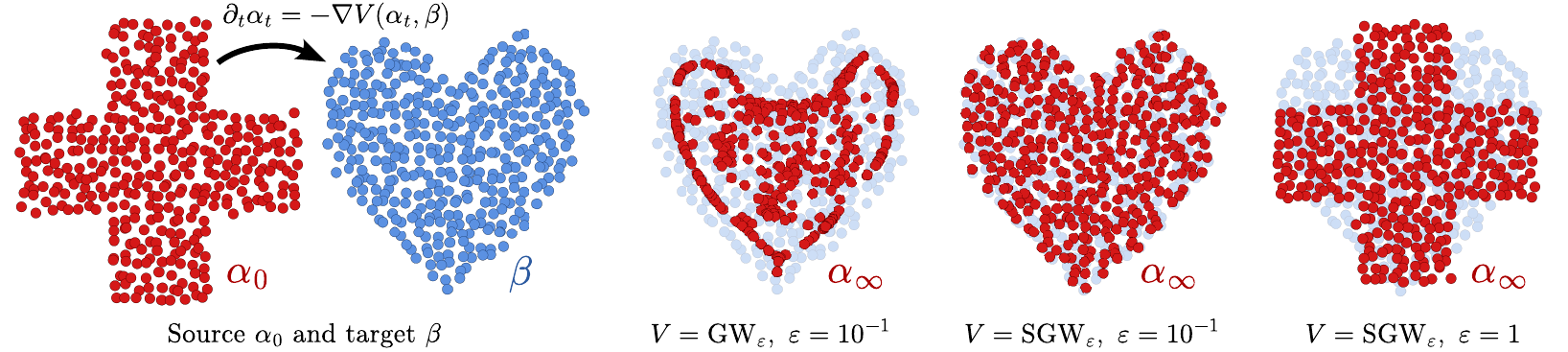}
    \caption{Wasserstein gradient flows $\alpha_{t+\delta t} = \alpha_{t} - \delta t \nabla V(\alpha_t, \beta)$ transporting a source measure $\alpha_0$ (red cross) towards a target $\beta$ (blue heart).
    Entropic bias causes the limit measure to collapse in small clusters
    with $\GW_\varepsilon$.
    The debiased
    Sinkhorn GW divergence $\SGW_\varepsilon$ fixes this issue, as the flow converges to the target.
    At large temperature $\varepsilon$, convergence fails and the flow remains cross-shaped.}
    \label{fig:gradientflow_entropicbias}
\end{figure*}

\begin{restatable}{theorem}{egwicppi}
\label{corrolary:egw_icp_pi}
Let
$\overline{\alpha} = \Phi_\#\alpha$
and 
$\overline{\beta} = \Psi_\#\beta$
denote the pushforward images of the distributions
$\alpha$ and $\beta$ by their respective GW-embeddings.
Since $\Phi$ and $\Psi$ are injective, we can identify the supports
of $\alpha$ and $\beta$ with those of
$\overline{\alpha}$ and $\overline{\beta}$
and remark that
$\mathcal{F}$ is convex with respect to $\pi$, with:
\begin{align*}
\pi^\star(\Gamma) := \argmin_{\pi \in \coupling{\alpha}{\beta}}
     \mathcal{F}(\Gamma, \pi)
     =
     \arg_{\pi \in \coupling{\alpha}{\beta}}
     \OT_{\varepsilon/8}(\overline{\alpha}, \overline{\beta})
\end{align*}
for the \emph{bilinear} cost
$c_\Gamma(x,y) := -2 \langle x, \overline{\Gamma} y \rangle_{\Hilb_\X\times \R}$.

Equivalently, if we identify the supports of $\overline{\beta}$ and $\overline{\alpha}$ with their pushforward images under the actions
of the linear map $\overline{\Gamma}$ and of its transpose 
$\overline{\Gamma}^{\top}$,
we can write that:
\begin{align*}
\pi^\star(\Gamma) ~&=~
     \arg_{\pi \in \coupling{\alpha}{\beta}}
     \OT_{\varepsilon/8}
     \big(\, \overline{\alpha}, \, \overline{\Gamma}_\# \overline{\beta} \,\big)
     \\
     &=~
     \arg_{\pi \in \coupling{\alpha}{\beta}}
     \OT_{\varepsilon/8} \big(\, \overline{\Gamma}_\#^{\top} \overline{\alpha}, \,\overline{\beta} \,\big)
\end{align*}
for the \emph{squared norms costs} in $\Hilb_\X \times \R$
and $\Hilb_\Y \times \R$.
\end{restatable}

As illustrated in Fig.~\ref{fig:main_theorem}, this remarkable result allows us to offload most of the complexity of the GW problem to the well-understood theory of entropy-regularized OT with squared Euclidean costs.
In the remainder of the paper, we 
show that starting from an arbitrary linear map $\Gamma_0$ in $ \Hspace$
and iteratively applying the alternate minimization updates:
\begin{equation}
    \label{eq:alternate_minimization_scheme}
    \pi_{t+1}\gets \pi^\star(\Gamma_t)~,~~~~~
    \Gamma_{t+1} \gets \Gamma^\star(\pi_{t+1})
\end{equation}
is both easy to implement at scale and enjoys strong theoretical guarantees.

\section{Theoretical Guarantees of CNT-EGW}
\label{section:sinkhorn_biases}

%Although entropic regularization offers significant computational benefits, it also modifies the properties of the GW problem.
%We identify three main side effects: the \emphasize{entropic bias} alters the nature of EGW itself, while the \emphasize{optimization bias} arises from the structure of EGW solvers and the \emphasize{convergence drift} is induced by numerical approximations in the Sinkhorn algorithm.
%Exploiting the dual formula of \cref{prop:entropic_gw_cnt_linearized}, we show that these biases are either absent or correctable when the costs are CNT.
%These findings highlight the relevance of CNT costs for EGW, while providing new techniques to speed-up solver computations without sacrificing their stability.

%\cref{corrolary:egw_icp_pi} reformulates the original EGW problem as a minimization over a collection of quadratic OT costs, which enjoy strong theoretical properties and have been widely studied in the literature.
We now leverage \cref{corrolary:egw_icp_pi} to prove two key results for EGW with CNT costs that have important practical implications: EGW can be debiased, enabling its use for variational methods; and EGW solvers are provably convergent, making them reliable in all contexts.

\subsection{Entropic Bias and Symmetric Debiasing}
\label{subsection:entropic_bias}

\paragraph{Entropic Bias.}
Unlike GW, the EGW problem does not satisfy the geometric properties of \cref{prop:gwisometry}: if $\varepsilon > 0$, we no longer have $\GW_{\varepsilon}(\alpha, \alpha) = 0$ and $\alpha$ does not necessarily minimize $\GW_{\varepsilon}(\alpha, \beta)$ over $\beta \in \prob{X}$.
This \emphasize{entropic bias} is well known in the EOT setting and 
is often addressed by introducing the \emphasize{Sinkhorn divergence}:
\begin{equation*}
    \SOT_{\varepsilon}(\alpha, \beta) = \OT_{\varepsilon}(\alpha, \beta) - \tfrac{1}{2} ( \OT_{\varepsilon}(\alpha, \alpha) + \OT_{\varepsilon}(\beta, \beta))~.
\end{equation*}
We propose a new formulation of the result of \cite{feydy2019interpolating}
that highlights the importance of the class of CNT costs for entropy-regularized OT:
\begin{restatable}{theorem}{cntotproperties}
\label{prop:cnt_ot_properties}
Let $c$ be a continuous cost  (not necessarily  vanishing on the diagonal) on a bounded domain. Then, 
   $c$ is CNT if and only if $\SOT_{\varepsilon} \geq 0$ for any $\varepsilon > 0$.
\end{restatable}

By analogy, we define the \emphasize{Sinkhorn GW} divergence as:
\begin{equation*}
    \SGW_{\varepsilon}(\alpha, \beta)\!=\!\GW_{\varepsilon}(\alpha, \beta) - \tfrac{1}{2} (\GW_{\varepsilon}(\alpha, \alpha) + \GW_{\varepsilon}(\beta, \beta)).
\end{equation*}

The use of SGW was previously proposed in \cite{sejourne2022generalized,rioux2024entropic}; however, whether this approach effectively corrects the entropic bias remained an open question.
In the CNT case, we provide a positive answer:
\begin{restatable}{theorem}{cntsgwproperties}
\label{prop:cnt_sgw_properties}
    If $c_\X$ and $c_\Y$ are CNT costs, then  $\SGW_{\varepsilon}(\alpha, \beta) \geq 0$ for any distributions $\alpha$, $\beta$ and temperature $\varepsilon > 0$. 
    If $\alpha$ and $\beta$ are isometric, then $\SGW_{\varepsilon}(\alpha, \beta) = 0$.
\end{restatable}
Yet, the nullity of $\SGW_{\varepsilon}(\alpha, \beta)$ does not imply that $\alpha$ and $\beta$ are isometric.
When the regularization parameter $\varepsilon$ is too large, $\SGW_{\varepsilon}$ may fail to separate non-isometric measures: we discuss examples in \cref{prop:sgw_counterex}.
In finite dimension, we provide a criterion that controls this effect based on the eigenvalues of the covariance matrices  $\Sigma_{\alpha} = \int x x^\top \diff\alpha(x)$ and $\Sigma_{\beta} = \int y y^\top \diff\beta(y)$ of centered distributions.

%\begin{restatable}{proposition}{sgwdefinitenessquad}
%\label{prop:definiteness_cnt_quad}
%    Let $d \in \N$, $\varepsilon > 0$ and $R > 0$. 
%    Let $\delta = K \varepsilon d \cdot (1  +\log(\frac{dR}{\varepsilon}))$, with $K$ a constant. 
%    Let $c_X,c_Y = \norm{\cdot}^2$.
%    For all $\alpha,\beta \in \mathcal{M}_{1,\sqrt{\delta},R}^+(\R^d)$, $\alpha$ and $\beta$ are isometric if and only if  $\SGW_{\varepsilon}(\alpha, \beta) = 0$. 
%\end{restatable}

\begin{restatable}{proposition}{sgwdefinitenessquad}
\label{prop:definiteness_cnt_quad}
    Let $\varepsilon > 0$, $\alpha \in \prob{\R^\DD}$, $\beta \in \prob{\R^{\EE}}$ with supports of diameter $R$, $R'$, and $c_\X, c_\Y$ the squared norm of $\R^\DD$ and $\R^{\EE}$.
    Let $\lambda_\alpha$, $\lambda_\beta$ be the smallest eigenvalues of $\Sigma_\alpha$ and $\Sigma_\beta$.
    There are constants $C(\DD,R)$ and $C(\EE,R')$ such that, if:
    \begin{equation*}
        \varepsilon  \max \left(1, \log(1/\varepsilon)\right) \leq \max \left(C(\DD,R) \!\cdot\! \lambda_\alpha^2,~C(\EE,R') \!\cdot\! \lambda_\beta^2 \right) \!,
    \end{equation*}
    then $\alpha$ and $\beta$ are isometric if and only if $\SGW_{\varepsilon}(\alpha, \beta) = 0$.
\end{restatable}

In words: the measures must be spread across all ambient dimensions for EGW to accurately capture their differences.
This result cannot be generalized to infinite dimensions, since the covariance operators $\Sigma_\alpha$ and $\Sigma_\beta$ are compact, with $0$ as the infimum of their singular values.
As illustrated in \cref{fig:gradientflow_entropicbias},
keeping $\varepsilon$ as small as possible is thus essential to let
$\SGW_\varepsilon$ separate any two non-isometric distributions.

\subsection{Convergence of EGW Solvers}
\label{subsection:optim_bias}

Most existing solvers for the EGW problem
of Eq.~\eqref{eq:entropic_gromov_wasserstein}
rely on descent schemes in the space of couplings $\coupling{\alpha}{\beta}$.
Illustratively, \textsc{Entropic-GW} -- the influential algorithm of \cite{peyre2016gromov} -- relies on the updates:
%To benefit from the Sinkhorn algorithm speedup, EGW must be formulated as a series of EOT problem. 
%The method proposed by \cite{peyre2016gromov} which inspired most of the EGW solvers hence relies on the updates: 
\begin{equation}
\label{eq:entropic_cgd_sequence}
\pi_{t+1}\!=\!\argmin_{\pi \in \coupling{\alpha}{\beta}}\!\left( \int\! \grad_{\pi_t}{\Loss} \cdot \diff\pi + \varepsilon \KL{\pi}{\alpha \otimes \beta}\!\right)\!.
\end{equation}
The authors proved the equivalence of Eq.~\eqref{eq:entropic_cgd_sequence} with a projected gradient descent (PGD) for the KL geometry with a step size $\tau = 1/\varepsilon$.
Although PGD is guaranteed to converge when $\tau$ is sufficiently small, this may not be compatible with the small values of $\varepsilon$ that are required to approximate
the true GW problem. 
In practice, \textsc{Entropic-GW} may thus oscillate between distinct transport plans without reaching convergence: we show a typical failure case in \cref{fig:gradientbias}. 

%On the other hand, the alternate minimization of \cref{eq:alternate_minimization_scheme} ensures that EGW loss decreases at each step, preventing optimization instabilities. Strikingly, both strategies are actually equivalent: we immediately deduce the stability of \textsc{Entropic-GW} in the CNT case.
Fortunately, this cannot happen with CNT costs
since \textsc{Entropic-GW} is then equivalent to our scheme:
\begin{restatable}[Primal descent]{proposition}{equivalencepgd}
\label{prop:equivalence_pgd}
Couplings $(\pi_t)$ obtained by alternate minimization from Eq.~\eqref{eq:alternate_minimization_scheme} also satisfy Eq.~\eqref{eq:entropic_cgd_sequence},
with a monotonic decrease of the EGW loss at every step.
\end{restatable}

Similarly, we draw an equivalence between alternate minimization and gradient descent on the dual space $\Hspace$:
\begin{restatable}[Dual descent]{proposition}{equivalencedualgrad}
\label{prop:equivalence_dualgrad}
Operators $\Gamma_t \in \Hspace$ obtained by alternate minimization from Eq.~\eqref{eq:alternate_minimization_scheme} also satisfy:
\begin{equation}
\label{eq:gd_sequence}
    \begin{gathered}
     \Gamma_{t+1} = \Gamma_t - (1/2) \cdot \grad{\Dualfun_{\varepsilon}}(\Gamma_t) \\  \text{where} \quad \Dualfun_{\varepsilon}: \Gamma \mapsto 
     \textstyle\min_{\pi \in \coupling{\alpha}{\beta}}
     \mathcal{F}(\Gamma, \pi)
\end{gathered}
\end{equation}
derives from the objective $\mathcal{F}$ defined in \cref{prop:entropic_gw_cnt_linearized}.
\end{restatable}

\cite{rioux2024entropic}
studied a similar dual gradient descent scheme in finite dimensions,
with step sizes that must stay small ($\tau = \mathcal{O}(\varepsilon)$)
to account for the reduced smoothness of $\Dualfun_\varepsilon$
at low temperatures $\varepsilon$.
Beyond generalizing this framework to CNT costs, \cref{prop:equivalence_dualgrad} makes the gradient step size independent of $\varepsilon$.
Our alternate minimization point of view allows us to adapt the majorization-minimization (MM) framework to our problem \cite{lange2016mm}, providing quantitative convergence results even with a fixed $\tau = 1/2$:

\begin{restatable}{theorem}{altminconvergencereg}
\label{prop:altmin_convergence_reg}
Let $(\pi_t, \Gamma_t)$ be the alternate sequence of Eq.~\eqref{eq:alternate_minimization_scheme}.
Then, $\norm{\Gamma_t - \Gamma_{t+1}}_{\HS} \rightarrow 0$ and every  subsequential limit of $(\Gamma_t)$ is a critical point of $\Dualfun_{\varepsilon}$, with:
\begin{equation}
\label{eq:cvg_estimate_gamma}
    \textstyle\sum_{k=0}^{t-1} \norm{\Gamma_k - \Gamma_{k+1}}_{\HS}^2 ~~\leq~~ \Dualfun_{\varepsilon}(\Gamma_0) - \Dualfun_{\varepsilon}(\Gamma_{t})~.
\end{equation}
\end{restatable}

\paragraph{Optimization in the Space of Linear Operators.} Although we established the equivalence of alternate minimization, primal descent and dual descent, these optimal schemes offer distinct perspectives.
By parameterizing the problem via the operator $\Gamma$, whose dimensions reflect the geometries of the costs, we effectively recast the combinatorial matching problem as a linear registration task in a high-dimensional feature space.

This reparameterization has profound practical implications. 
Because estimating a global linear transform is inherently more robust than resolving a precise point-wise matching, the solver becomes tolerant to approximation errors. Consequently, we can employ cheap, coarse estimates of the Optimal Transport solutions needed in \cref{corrolary:egw_icp_pi} without compromising the final convergence.
We discuss this point further in \cref{appendix:additional_expes}.
Moreover, since the dual variable $\Gamma$ encodes the global geometric alignment of the measures $\alpha$ and $\beta$ without being tied to their supports, it provides a robust framework to visualize the optimization landscape (as illustrated in \cref{section:experiments}) and naturally supports the transfer of correspondences across different sampling resolutions.

\section{Implementation}

\paragraph{Notations.}
We now consider two discrete measures:
\begin{equation*}
   \alpha = \textstyle\sum_{i=1}^\NN a_i \dirac{x_i} ~~\text{ and }~~ \beta = \textstyle\sum_{j=1}^\MM b_j \dirac{y_j}~,
\end{equation*}
where the samples
$x_1, \dots, x_\NN \in \X$ and
$y_1, \dots, y_\MM \in \Y$
belong to the base spaces $\X$ and $\Y$
while the two collections of probability weights
$a_1, \dots, a_\NN \geqslant 0$ and 
$b_1, \dots, b_\MM \geqslant 0$ sum up to $1$.
We identify admissible couplings $\pi\in \coupling{\alpha}{\beta}$ with matrices $(\pi_{ij})~\in~\R^{\NN\times \MM}$ that satisfy the marginal constraints $\sum_{j} \pi_{ij} = a_i$ and $\sum_{i} \pi_{ij} = b_j$ for all $i$ and $j$.

\paragraph{CNT-GW Solver.} 
Our method relies on the embeddings introduced in \cref{section:gw_cntcosts}.
Assuming that we have access to \emph{centered} collections
of embedding vectors $X_1, \dots, X_\NN \in \R^\DD$
and
$Y_1, \dots, Y_\MM \in \R^\EE$
such that for all $i$ and $j$:
\begin{align}
    \textstyle \sum_{i=1}^\NN a_i X_i ~&=~ 0~, &
    c_\X(x_i, x_j)~&=~ \|X_i - X_j\|^2~, \label{eq:embedding_X}\\
    \textstyle \sum_{j=1}^\MM b_j Y_j ~&=~ 0~, &
    c_\Y(y_i, y_j)~&=~ \|Y_i - Y_j\|^2~, \label{eq:embedding_Y}
\end{align}
\cref{alg:dual_kernelpca} provides an efficient implementation
of the alternate minimization scheme of Eq.~\eqref{eq:alternate_minimization_scheme}.
In this pseudo-code, that we detail in \cref{subsection:sinkhorn_algorithms}, $\textsc{SinkhornOT}^{\|\cdot\|^2}_{\varepsilon/8}$
refers to a solver of the EOT problem
for the squared Euclidean cost $c(x, y) = \|x-y\|^2$ in $\R^{\DD+1}$.
Following \cref{corrolary:egw_icp_pi}
and Eq.~\eqref{eq:sinkhorn_pi}, we use two dual vectors
$(f_i) \in \R^\NN$ and $(g_j) \in \R^\MM$ to encode the current optimal coupling $\pi_t \in\R^{\NN\times \MM}$ as:
\begin{equation*}
    \pi_{ij}
    ~=~
    a_i b_j \cdot \exp \tfrac{8}{\varepsilon}[
    f_i + g_j - \|\BAR{X}_i - \BAR{Z}_j\|^2
    ]
\end{equation*}
without having to store a large $\NN\times \MM$ array in memory.
In \cref{appendix:subsection:egw_grad}, we explain how to compute efficiently the derivatives of the $\GW_\varepsilon$ and $\SGW_\varepsilon$
loss functions with respect to the weights and positions of the samples.
Combined with the guarantees of \cref{section:sinkhorn_biases},
this original contribution opens the door to the use
of scalable loss functions based on GW theory.

\begin{algorithm}[t]
\caption{\textsc{CNT-GW} \\
\textbf{Input:} 
probability weights $a_1, \dots, a_\NN$
and $b_1, \dots, b_\MM$, \\
\phantom{~~~~}embeddings $X_1, \dots, X_\NN \in \R^\DD$
and
$Y_1, \dots, Y_\MM \in \R^\EE$. \\
\textbf{Output:} optimal dual potentials
$f \in \R^\NN$ and $g \in \R^\MM$, \\
\phantom{~~~~}optimal linear map $\Gamma \in\R^{\DD\times\EE}$.
}
\begin{algorithmic}[1]
\label{alg:dual_kernelpca}
%\Input{blabla}
% \STATE {\bfseries Input:} $a,b$, $c_X,c_Y$,  $\varepsilon > 0$, $S > 0$, $D > 0$,$\text{tol} > 0$.
\STATE $\Gamma_0 \gets 0_{\DD\times\EE} \text{~or another initialization in~}\R^{\DD\times\EE} $
\vspace*{.1cm}
\STATE $\BAR{X}_i \gets [X_i, \tfrac{1}{2}\|X_i\|^2] \in \R^{\DD+1}$
\vspace*{.1cm}
\WHILE{$\norm{\Gamma_t - \Gamma_{t+1}} > \text{tol}$}
    \vspace*{.1cm}
    \STATE $\BAR{Z}_j \gets [\Gamma_t Y_j, \tfrac{1}{2}\|Y_j\|^2] \in \R^{\DD+1}$

    \STATE $f, g \gets \textsc{SinkhornOT}^{\|\cdot\|^2}_{\varepsilon/8}(a_i, \BAR{X}_i ; b_j, \BAR{Z}_j)$

    \STATE $\widetilde{X}_j \gets \sum_{i=1}^\NN a_i X_i \cdot \exp \tfrac{8}{\varepsilon}[
    f_i + g_j - \|\BAR{X}_i - \BAR{Z}_j\|^2
    ]$
\vspace*{.1cm}
    \STATE $\Gamma_{t+1} \gets 
           \sum_{j=1}^\MM b_j \widetilde{X}_j Y_j^\top  \in \R^{\DD\times \EE} $
    \vspace*{.1cm}
    \STATE $t \gets t+1$
\vspace*{.1cm}
\ENDWHILE
\end{algorithmic}
\end{algorithm}

\paragraph{CNT Embeddings.}
While embeddings $X_i$ and $Y_j$ are easy to compute
when $c_\X$ and $c_\Y$ correspond to squared Euclidean costs,
simple expressions for the injective embedding maps of \cref{prop:entropic_gw_cnt_linearized}
may not always be available.
In this context, we propose to apply Kernel PCA \cite{scholkopf1997kernel}
on the kernels $k_\X$ and $k_\Y$ induced by
$c_\X$ and $c_\Y$ via \cref{prop:cnt_positive_kernel}.
Selecting a small number of PCA components allows us to produce
low-dimensional embeddings $X_i$ and $Y_j$ that approximately
satisfy Eqs.~(\ref{eq:embedding_X}-\ref{eq:embedding_Y}).
In the experiments below, we combine this method with other dimension reduction
techniques to handle cases where $c_\X$ and $c_\Y$
are (non-squared) Euclidean norms, kernel-based formulas or even geodesic distances
on 3D surfaces \cite{coifman2005geometric,lipman2010biharmonic,panozzo2013weighted}.
In \cref{appendix:subsection:detailed_kernel_impl}, we also present an \emph{exact} Algorithm, 
$\textsc{Kernel-GW}$,
that works directly with the cost functions $c_X$ and $c_Y$
but has to represent the (possibly infinite-dimensional) operator
$\Gamma$ as a $\NN \times \MM$ array via the kernel trick.

\textbf{Complexity.}
We can reuse
efficient computational routines that were initially designed for
the popular ``Wasserstein-2'' problem, i.e. OT with a squared Euclidean cost.
In \cref{section:detailed_implementation}, we detail how to use the \texttt{KeOps} library \cite{feydy2020fast,charlier2021kernel}
to implement the iterations of \textsc{CNT-GW} with
a $\mathcal{O}(\NN+\MM)$ memory footprint
and $\mathcal{O}(\NN\MM)$ time complexity.
Likewise, we compute kernel PCA embeddings
with a $\mathcal{O}(\NN\DD + \MM\EE)$ memory footprint
and $\mathcal{O}(\NN^2\DD + \MM^2\EE)$ time complexity.
All in all, finding a local minimum of the EGW problem for CNT costs
with \cref{alg:dual_kernelpca} thus requires
$\mathcal{O}(\NN\DD + \MM\EE)$ storage space
and
$\mathcal{O}(\NN^2\DD + \MM^2\EE + \NN \MM \ninner \nouter)$
time, where
$\ninner$ and $\nouter$
stand for the number of iterates in the \textsc{SinkhornOT}
and \textsc{CNT-GW} solvers, respectively.
These two integers depend on the geometry of the input distributions
and usually stay in the $20$-$200$ and $10$-$50$ ranges, respectively.
In \cref{subsection:multiscale_gw}, we also introduce a multiscale \textsc{MsGW} solver that implements the first outer iterations of \textsc{CNT-GW}
on coarse sub-samplings of the input distributions.

To conclude, we stress that our Algorithms remain local optimization methods.
While our fast solvers can be used on random initializations
$\Gamma_0$ to quickly explore the basins of the optimization landscape,
computing the global optimum of the EGW problem in full generality remains difficult.

\label{section:implementation}

\section{Experiments}
\label{section:experiments}

%\begin{table*}[t!]
%  \caption{Evaluation of convergence times and GW loss $GW_{\varepsilon}$ on shape data. We use $c_\X,c_\Y = \norm{\cdot}_{\text{geodesic}}$ for \textsc{Kids} and \textsc{Faust}, $c_\X,c_\Y = \norm{\cdot}_{\R^3}$ otherwise.}
%  \label{table:benchmark}
%  \begin{center}
%    \begin{small}
%      \begin{sc}
%        \begin{tabular}{lrrrrrrrrrrr}
%          \toprule
%          Shapes & $\NN, \MM$ & \multicolumn{2}{c}{EGW} & \multicolumn{2}{c}{KGW} & \multicolumn{2}{c}{QLrGW} & \multicolumn{2}{c}{CntGW} & \multicolumn{2}{c}{MsGW}  \\
%           & & Time &  $GW_{\varepsilon}$ & Time &  $GW_{\varepsilon}$ & Time &  $GW_{\varepsilon}$ & Time &  $GW_{\varepsilon}$ \\
%          \midrule
%          Horses & $4k$ & $81s$ & - & $16s$ & - & $15s$ & - & $3s$ & - &$\mathbf{2s}$ & - \\
%          Hands & $10k$& $426s$ & - & $256s$ & - & $45s$ & - & $21s$ & - & $\mathbf{5s}$ & - \\
%          Dogs & $36k$& \multicolumn{2}{c}{Mem} & \multicolumn{2}{c}{Mem} & $205s$ & - & $145s$ & - & $\mathbf{26s}$ & - \\
%          Kids &$60k$ & \multicolumn{2}{c}{Mem} & \multicolumn{2}{c}{Mem} & $2,025s$ & $\mathit{0.025}$ & $279s$ & $\mathit{0.020}$ & $\mathbf{83s}$ & $\mathit{0.020}$ \\
%          Faust &$177k$ & \multicolumn{2}{c}{Mem} & \multicolumn{2}{c}{Mem} & \multicolumn{2}{c}{-} & \multicolumn{2}{c}{-} & $\mathbf{355s}$ & $\mathit{0.014}$ \\
%          \bottomrule
%        \end{tabular}
%      \end{sc}
%    \end{small}
%  \end{center}
%  \vskip -0.1in
%\end{table*}

\begin{table*}[t!]
  \caption{Evaluation of different solvers on shape data (solving time and GW objective attained after convergence).
  We use $c_\X,c_\Y = \norm{\cdot}_{\text{geodesic}}$ for \textsc{Kids}, \textsc{Faust}, \textsc{Hips} and \textsc{Vessels}; $c_\X,c_\Y = \norm{\cdot}_{\R^3}$ otherwise. For \textsc{Faust}, only \textsc{MsGW} was evaluated as it is the only solver converging in less than one hour. True objectives $GW_{\varepsilon}$ are only reported on \textsc{Horses} and \textsc{Hands}; for larger datasets, these could not be computed without memory overflow, and approximate values are given instead (in italic).}
  \label{table:benchmark}
  \begin{center}
    \begin{small}
        \begin{tabularx}{\textwidth}{lrRRRRRRRRRR}
          \toprule
           \sc Shapes & $\NN, \MM$ & \multicolumn{2}{c}{\sc \quad EGW} & \multicolumn{2}{c}{\sc \quad KGW} & \multicolumn{2}{c}{\sc \quad QLrGW} & \multicolumn{2}{c}{\sc \quad CntGW} & \multicolumn{2}{c}{\sc \quad MsGW}  \\
           & & \sc Time &  $GW_{\varepsilon}$ & \sc Time &  $GW_{\varepsilon}$ & \sc Time &  $GW_{\varepsilon}$ & \sc Time &  $GW_{\varepsilon}$ & \sc Time &  $GW_{\varepsilon}$ \\
          \midrule
          \sc Horses & $4k$ & $81s$ & $1.4\text{e-2}$ & $16s$ & $1.3\text{e-2}$ & $15s$ & $1.4\text{e-2}$ & $3s$ & $1.3\text{e-2}$ &$\mathbf{2s}$ & $1.3\text{e-2}$ \\
          \sc Hands & $10k$& $426s$ & $1.3\text{e-2}$ & $256s$ & $1.3\text{e-2}$ & $45s$ & $1.3\text{e-2}$ & $21s$ & $1.3\text{e-2}$ & $\mathbf{5s}$ & $1.3\text{e-2}$ \\
          \sc Femurs &$25k$ & \sc Mem. & \sc Mem. & \sc Mem. & \sc Mem. & $54s$ & $\mathit{6.3}\textit{e-3}$ & $21s$ & $\mathit{6.5}\textit{e-3}$ & $\mathbf{10s}$ & $\mathit{6.6}\textit{e-3}$ \\
          \sc Dogs & $36k$& \sc Mem. & \sc Mem. & \sc Mem. & \sc Mem. & $205s$ & $\mathit{1.2}\textit{e-2}$ & $145s$ & $\mathit{1.2}\textit{e-2}$ & $\mathbf{26s}$ & $\mathit{1.2}\textit{e-2}$ \\
          \sc Kids &$60k$ & \sc Mem. & \sc Mem. & \sc Mem. & \sc Mem. & $2,025s$ & $\mathit{2.5}\textit{e-2}$ & $279s$ & $\mathit{2.0}\textit{e-2}$ & $\mathbf{83s}$ & $\mathit{2.0}\textit{e-2}$ \\
          \sc Hips &$60k$ & \sc Mem. & \sc Mem. & \sc Mem. & \sc Mem. & $479s$ & $\mathit{1.0}\textit{e-2}$ & $342s$ & $\mathit{9.4}\textit{e-3}$ & $\mathbf{89s}$ & $\mathit{9.4}\textit{e-3}$ \\
          \sc Vessels &$100k$ & \sc Mem. & \sc Mem. & \sc Mem. & \sc Mem. & $2,444s$ & $\mathit{7.7}\textit{e-3}$ & $891s$ & $\mathit{7.5}\textit{e-3}$ & $\mathbf{196s}$ & $\mathit{7.4}\textit{e-3}$ \\
          \sc Faust &$177k$ & \sc Mem. & \sc Mem. & \sc Mem. & \sc Mem. & $>1h$ &- & $>1h$ & - & $\mathbf{355s}$ & $\mathit{1.4}\textit{e-2}$ \\
          \bottomrule
        \end{tabularx}
    \end{small}
  \end{center}
  \vskip -0.1in
\end{table*}

\begin{figure*}[t]
\centering
\includegraphics[width=\linewidth]{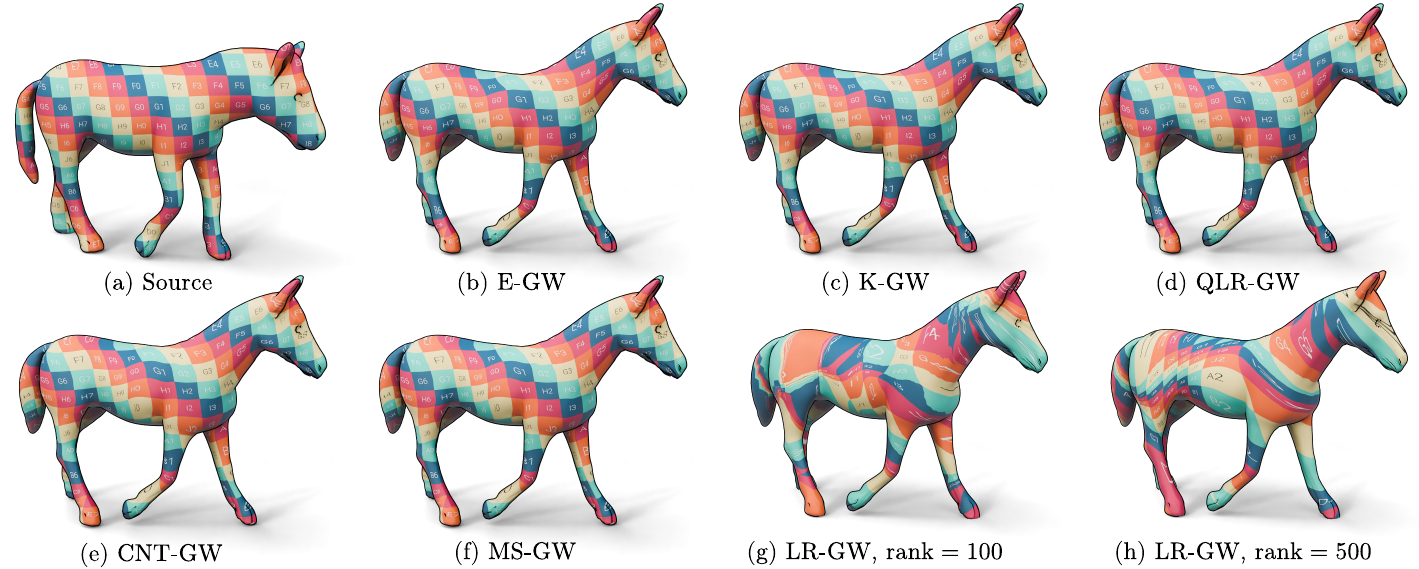} 
    \caption{GW-based texture transfer between horse shapes, using transport plans computed with different solvers (as detailed in \cref{appendix:experiment_details}).
    We also illustrate the limitations of heuristic solvers by plotting transport plans obtained with \textsc{LowRank-GW}.}
    \label{fig:horse_texture_comparisons}
\end{figure*}

\paragraph{Setup.}
We now present some benchmarks and illustrative experiments,
performed on a Dell laptop powered by an Intel i7-12700H CPU and
a relatively modest Nvidia RTX 3050 GPU with 4Gb of VRAM.
When not specified otherwise, we use embeddings
in dimension $\DD=\EE=20$ (that explain $80\%$ to $90\%$ of the variance of the kernel matrices), normalize the input distributions to have a radius of $1$ and
use a temperature $\varepsilon = 10^{-3}$.
For the sake of simplicity, we run the \textsc{SinkhornOT} solver
with a fixed number of $100$ iterations and stop all GW solvers
when the objective decreases by less than $10^{-5}$ between two consecutive steps. For our methods, this
usually corresponds to $10$ to $50$ ``outer'' steps depending
on the geometry of the problem.
We provide full details for our experiments in \cref{appendix:experiment_details} and leave a thorough exploration of these parameters to future works.

%\paragraph{Differentiation of Gromov-Wasserstein losses.} Our new EGW formulation also allows us to efficiently compute its gradient with respect to the input measures $\alpha$ and $\beta$.
%Indeed, EOT gradients can be estimated simply from Sinkhorn dual potentials \cite{feydy2019interpolating}, a technique that we adapt to the EGW case in \cref{appendix:subsection:egw_grad} and that is made tractable thanks to the kernel trick.
%This open the way for dynamic deformations based on EGW potentials, and make of Gromov-Wasserstein a suitable loss function for geometric machine learning applications.  

%\paragraph{Sinkhorn annealing.} A la fin, pour obtenir un plan de transport sharp.

\paragraph{Benchmarks.} 
As detailed in \cref{appendix:subsection:review_computational_GW},
the \textsc{Entropic-GW} (\textsc{EGW}) \cite{peyre2016gromov} and
\textsc{Quadratic-LowRank-GW} (\textsc{QLrGW}) \citep[Algorithm 2]{scetbon2022linear}
solvers provide strong baselines for the estimation of optimal GW couplings at scale.
In \cref{table:benchmark}, we compare our \cref{alg:dual_kernelpca} (\textsc{CntGW}), \cref{alg:kernel_gw} (\textsc{KGW})
and \cref{alg:multiscale_dual_kernelpca} (\textsc{MsGW}) to these state-of-the-art solvers
on pairs of shapes from the CAPOD \cite{papadakis2014canonically}, SMAL \cite{Zuffi_CVPR_2017}, KIDS \cite{rodola2014dense}, FAUST \cite{Bogo_CVPR_2014} and MedShapeNet \cite{li2025medshapenet} datasets. 
We supplement these results by extensive convergence plots, ablation studies
and visualizations in \cref{appendix:additional_expes}.
We observe a significant speed-up for our new solvers \textsc{KGW} and \textsc{CntGW} over their corresponding baselines  \textsc{EGW} and \textsc{QLrGW}, while achieving similar GW objectives: \cref{fig:horse_texture_comparisons} confirms that the output of these methods are all qualitatively equivalent.
Most importantly, our multiscale $\textsc{MsGW}$ solver consistently outperforms competitors by one to two orders of magnitude.
Although heuristic accelerations exist, like \textsc{LowRank-GW} \citep[Section 5]{scetbon2022linear}, they generate important artifacts that strongly degrade the transport plan quality (\cref{fig:horse_texture_comparisons}g and h).
As a result, our multiscale algorithm is the only GW solver capable of scaling to more than $\NN=100k$ samples in minutes while providing transport maps of sufficient quality for shape matching applications.

\paragraph{Optimization Landscape.} In \cref{fig:introduction}e, we use \textsc{MsGW} with geodesic costs to transfer a texture from one FAUST surface to another.
We guide the registration by specifying $5$ pairs of landmarks
$(x_1, y_1)$, \dots, $(x_5, y_5)$ 
that should be matched with each other
at the top of the skull and at the tips of the arms and legs. They induce a covariance matrix
$\Gamma_0 = \sum_{k=1}^5 X_k Y_k^\top$ in feature space that we use to initialize the alternate minimization: this ensures that we fall in the ``correct'', user-defined local minimum of the GW objective.

%$\Gamma$ was initialized using $5$ landmarks (on the head and the extremities of the legs and arms) and a last annealed Sinkhorn was run after after convergence to obtain a sharp transport plan.
%Using Multiscale-GW, the computation took $11$ min in total ($70 s$ for Kernel PCA, $15 s$ for coarse optimization, $365 s$ for fine-scale optimization and $195 s$ for annealed Sinkhorn).
%Therefore, we are able to compute accurate correspondences between large scale shapes in only a few minutes.

Going further, we explore the optimization landscape of the GW problem
as we match two poses from the KIDS dataset, with cost functions
that correspond to geodesic distances embedded in spaces
of dimension $\DD=\EE=20$ \cite{panozzo2013weighted}.
We sample the coefficients of $500$ initializations $\Gamma_0 \in \R^{\DD\times\EE}$ uniformly at random (with a relevant scale factor),
and run the ``coarse'' phase of \textsc{MsGW} with this diverse collection of seeds.
We identify $25$ distinct local minima, that we refine in the second step of \textsc{MsGW} to obtain high-resolution transport plans $\pi$.
As illustrated in \cref{fig:kids_landscape},
the local minima of the GW problem correspond to body part permutations,
with the 8 best candidates attracting $82\%$ of our random seeds $\Gamma_0$.
In the lower right corner, the global optimum corresponds to the widest basin of attraction and represents the correct alignment of limbs.

We perform a Principal Component Analysis on the 25
linear maps $\Gamma\in\R^{\DD\times\EE}$ that encode these local minima, weighted by the number of seeds that they attracted.
This allows us to represent them in the most relevant
2-dimensional plane of the full operator space,
with contour lines of the objective function $\Dualfun_\varepsilon$
of Eq.~\eqref{eq:gd_sequence} in the background.
This figure reveals the structure of the matching problem,
as a left-right symmetry dominates the visualization.

\begin{figure}[t]
\centering
\includegraphics[width=\linewidth]{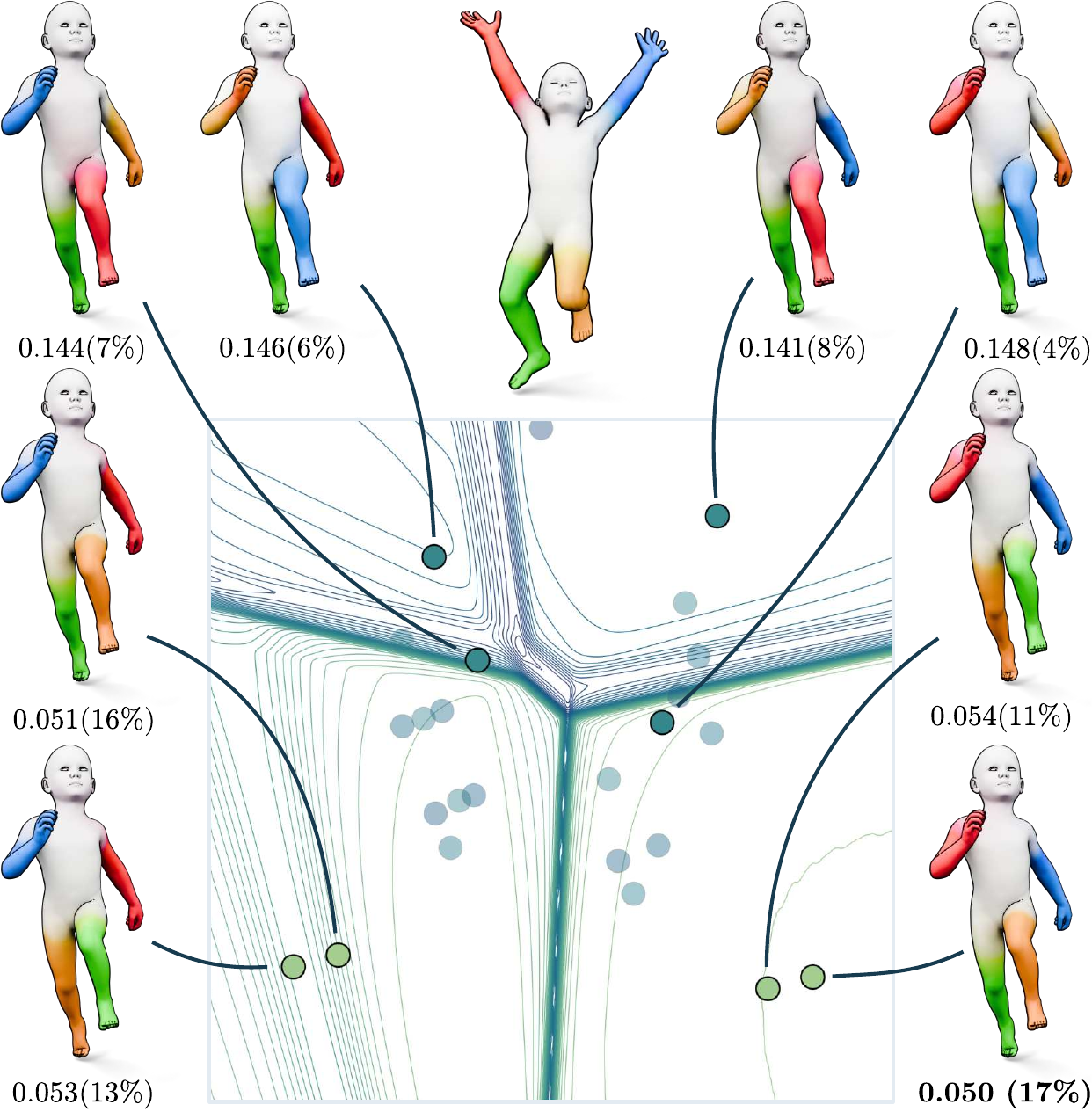} 
    \caption{Optimization landscape of the EGW problem, visualized in a dual plane.
    We match a source point cloud (with arms raised) to a target pose (running).
    We highlight the $8$ best minima, the corresponding EGW losses and the proportions of random seeds $\Gamma_0$ that fell in their attraction basins.
    }
    \label{fig:kids_landscape}
\end{figure}

\paragraph{GW Barycenters.} 
%We finally showcase an application of our framework to anatomic shape processing.
%\cref{fig:vertebres} shows that GW provides anatomically meaningful comparison between vertebras with clearly non-isometric shapes; it also highlights the benefits of the CNT framework, allowing to use geodesic distances as base costs -- whereas squared norm GW would only enable squared geodesic distances that provide non satisfying results.
Since our solvers provide gradients for the $\GW_\varepsilon$
objective and the debiased $\SGW_\varepsilon$ divergence,
we can easily implement gradient flows and other variational schemes
that involve the GW metric.
To illustrate this, we compute GW barycenters between
anatomical shapes from the BodyParts3D\footnote{(c) The Database Center for Life Science licensed under CC Attribution-Share Alike 2.1, Japan.} dataset
in \cref{fig:egw_barycenters}.
Specifically, we sample $\NN=3,000$ points uniformly on the source and target surfaces to create two distributions of points $\alpha$ and $\beta$ in $\R^3$.
Then, we perform gradient descent on a point cloud $z_1$, \dots, $z_\NN$ 
that parameterizes a measure $\gamma_z = \tfrac{1}{\NN}\sum_i \delta_{z_i}$ in order to minimize the weighted objective:
\begin{equation*}
    (1 - \lambda) \cdot \SGW_\varepsilon(\gamma_z,\alpha) + \lambda \cdot \SGW_\varepsilon(\gamma_z, \beta)~,
\end{equation*}
where $\lambda\in [0,1]$ is an interpolation slider.
We turn the resulting point clouds into surface meshes
using Poisson surface reconstruction~\cite{kazhdan2006poisson}.

\begin{figure}[t]
\centering
\includegraphics[width=\linewidth]{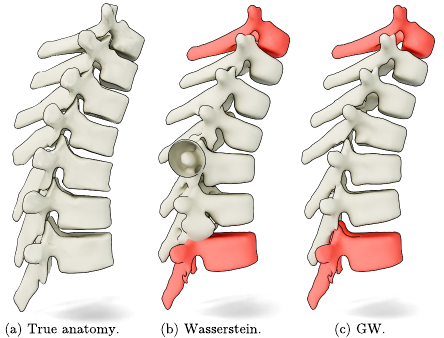} 
\caption{(a) Interpolating between the first (top) and seventh (bottom) thoracic vertebrae.
We compute barycenters on normalized data, and align them
with true anatomical positions for visualization purposes.
(b)~Wasserstein barycenters are now affordable, but
create many topological artifacts \cite{agueh2011barycenters}.
(c)~We compute GW barycenters for the Euclidean cost in $\R^3$.
The GW metric puts more emphasis on topology preservation
and could become a versatile baseline for 3D shape analysis.
}
\label{fig:egw_barycenters}
\end{figure}

\section{Conclusion}

While these experiments validate our theoretical findings, further development is required to turn these illustrations into competitive state-of-the-art methods for shape registration or domain adaptation.
Specifically, future works should adapt our approach to the unbalanced \cite{sejourne2021unbalanced} or fused-GW settings \cite{vayer2020fused,thual2022aligning}. Furthermore, the impact of cost functions, embedding dimensions or annealing schedules for the temperature $\varepsilon$ \cite{chizat2024annealed} remains to be fully characterized.

Despite these necessary extensions, we believe our work offers significant conceptual value to the field. 
By establishing a surprising connection between a broad class of Quadratic Assignment Problems (GW with CNT costs) and the linear registration of point clouds, we provide a new geometric perspective on structured data alignment.
This perspective closely relates to Procrustes–Wasserstein (PW) approaches \cite{grave2019unsupervised,alvarez2019towards}, which combine rigid alignment with optimal transport–based matching. 
A deeper comparison between the GW and PW frameworks would certainly provide valuable insights into both problems, from both geometric and algorithmic viewpoints.
This bridge also invites further investigation into the functional maps framework, where analogous problems have been explored in spectral or data-driven feature spaces \cite{ovsjanikov2012functional,ren2020maptree,pai2021fast}, and may pave the way for the development of provably approximate global optimizers \cite{jubran2021provably}.

% \newpage

\section*{Acknowledgements}

This work was supported by the French ``Agence Nationale de la Recherche'' via the ``PR[AI]RIE-PSAI'' project (ANR-23-IACL-0008).
We would also like to thank Robin Magnet for relevant comments
and his invaluable help with the rendering of the 3D figures.

\section*{Impact Statement}

This paper presents work whose goal is to advance the field of machine learning. There are many potential societal consequences of our work, none of which we feel must be specifically highlighted here.

\bibliography{paper}
\bibliographystyle{icml2026}

%%%%%%%%%%%%%%%%%%%%%%%%%%%%%%%%%%%%%%%%%%%%%%%%%%%%%%%%%%%%%%%%%%%%%%%%%%%%%%%
%%%%%%%%%%%%%%%%%%%%%%%%%%%%%%%%%%%%%%%%%%%%%%%%%%%%%%%%%%%%%%%%%%%%%%%%%%%%%%%
% APPENDIX
%%%%%%%%%%%%%%%%%%%%%%%%%%%%%%%%%%%%%%%%%%%%%%%%%%%%%%%%%%%%%%%%%%%%%%%%%%%%%%%
%%%%%%%%%%%%%%%%%%%%%%%%%%%%%%%%%%%%%%%%%%%%%%%%%%%%%%%%%%%%%%%%%%%%%%%%%%%%%%%

\newpage
\appendix
\onecolumn

\section{Additional Background}
\label{appendix:sinkhorn}
\subsection{The Sinkhorn Algorithm}

By convex duality, the EOT problem is equivalent to the maximization of the following objective:
\begin{equation*}
\OT_{\varepsilon}(\alpha, \beta) = \max_{f\in \Cont{X}} \max_{g \in \Cont{Y}} \int f(x) \diff\alpha(x) + \int g(y) \diff \beta(y)
    - \varepsilon \int \left( \exp \left( \frac{f(x) + g(y) - c(x,y)}{\varepsilon} \right) - 1 \right) \diff\alpha(x) \diff\beta(y).
\end{equation*}
The Sinkhorn algorithm solves this dual problem by block coordinate ascent.
Starting from an initial potential $g_0 \in \Cont{Y}$, it computes a sequence of functions $(f_s, g_s) \in \Cont{X} \times \Cont{Y}$ through alternating updates:
\begin{equation*}
f_{s+1} = \Sink{\varepsilon,\beta}(g_s) \quad \text{ and } \quad g_{s+1} = \Sink{\varepsilon,\alpha}(f_{s+1}),
\end{equation*}
where the Sinkhorn operators $\Sink{\varepsilon,\beta}$ and $\Sink{\alpha}$ are defined as:
\begin{equation}
\label{eq:sinkhorn_operators}
\begin{split}
    \Sink{\varepsilon,\beta}(g): x \in X \mapsto & - \varepsilon \cdot \log \int \exp \left( \frac{g(y) - c(x,y)}{\varepsilon} \right) \diff \beta(y), \\
    \Sink{\varepsilon,\alpha}(f) : y \in Y \mapsto & - \varepsilon \cdot \log \int \exp \left( \frac{f(x) - c(x,y)}{\varepsilon} \right) \diff \alpha(x).
\end{split}
\end{equation}
The sequence $(f_s,g_s)$ converges to the optimal dual potentials $(f^*, g^*)$, recovering the optimal EOT plan $\pi^*$ via:
\begin{equation*}
    \forall~x,y \in X\times Y, \quad \frac{\diff\pi^*}{\diff \alpha\diff \beta}(x, y) = \exp \left( \frac{f^*(x) + g^*(y) - c(x,y)}{\varepsilon} \right)~.
\end{equation*}
%In the finite setting, the Sinkhorn updates require at most $O(N^2)$ operations and can be implemented with a linear memory footprint using the \texttt{KeOps} library \cite{feydy2020fast}.
%Going further, when $c(x,y) = \|x-y\|^2$ is the squared Euclidean distance, annealing and multiscale heuristics let modern solvers reach sub-quadratic run times: see for instance the \texttt{GeomLoss} library \cite{feydy2019interpolating}.
However, the intermediate $\pi_s$ obtained by using $(f_s,g_s)$ instead of $(f^*,g^*)$ in the previous formula are not valid couplings as they do not exactly satisfy the marginal constraints.
For discrete measures of size $\NN$, the number of iterations required to reduce $\Err(\pi_s)$ to a threshold $\delta$ scales as $s = O(\delta^{-2} \cdot ( \log(\NN) + \norm{c}_{\infty}/\varepsilon))$ \cite{altschuler2017near}, where $\norm{c}_{\infty} = \max(c) - \min(c)$ and:
\begin{equation*}
    \Err(\pi_s) = \norm{\int \diff\pi_s(\cdot,y) - \alpha}_1 +  \norm{\int \diff\pi_s(x,\cdot) - \beta}_1.
\end{equation*}
Therefore, while the convergence rate is not strongly influenced by the input size $\NN$, it depends on the entropic regularization: keeping $\varepsilon$ sufficiently large is crucial to ensure fast computations.

\subsection{Review of the Computational Gromov-Wasserstein Literature}
\label{appendix:subsection:review_computational_GW}

This section provides an overview of existing GW solvers, including both unregularized and entropic approaches.
We also discuss algorithms designed for GW extensions such as Fused \cite{vayer2020fused} and Unbalanced \cite{sejourne2021unbalanced} Gromov-Wasserstein, as these frameworks offer a broader perspective on current computational challenges.
To maintain clarity, we adopt the notation from \cref{section:background} for general measures and \cref{section:implementation} for finite discretizations.

\subsubsection{Unregularized Gromov-Wasserstein Solvers}
\label{appendix:subsection:literature_unregularized_solvers}

\paragraph{Global Optimization Methods.} GW is closely related to the quadratic assignement problem (QAP), which is known to be NP-hard \cite{burkard1998quadratic}.
Consequently, finding a global minimum is computationally expensive.
Recent works have proposed relaxations to make this global search tractable: \cite{mula2024moment} formulate GW as a generalized moment problem that can be truncated at the desired degree of precision, while \cite{chen2024semidefinite} propose a semi-definite relaxation solvable in polynomial time.
Alternatively, \cite{ryner2023globally} exploit the structure of squared Euclidean costs to design a cutting plane algorithm that efficiently explores the set of admissible solutions.
Despite these advances, global methods remain slow and are limited to inputs containing at most a few hundred points.

\paragraph{Conditional Gradient Descent (Franke-Wolfe).} By contrast, local minima can be obtained in polynomial time using first-order methods.
The standard approach is the Conditional Gradient Descent (or Franke-Wolfe) algorithm \cite{titouan2019optimal}, whose updates are given by:
\begin{equation}
\label{eq:franke_wolfe}
\pi_{t+1} = (1 - \gamma) \cdot \pi_t + \gamma \cdot \tilde{\pi}_t \quad \text{ where } \quad \tilde{\pi}_t = \argmin_{\pi \in \coupling{\alpha}{\beta}} \int \grad_{\pi_t}{\Loss} \cdot \diff\pi~,
\end{equation}
where $\gamma \in [0, 1]$ is either fixed or selected by line search.
Each step is equivalent to a linear optimal transport problem with cost $c = \grad_{\pi_t}{\Loss}$, effectively transforming GW into a sequence of OT subproblems solvable in cubic time.
For finite measures, these subproblems take the form:
\begin{equation*}
    \tilde{\pi}^{t+1} = \argmin_{\pi \in \coupling{\alpha}{\beta}} \sum C^{t}_{ij} \pi_{ij}  \quad \text{ where } \quad  C^{t}_{ij} = \sum_{kl} (c_\X(x_i, x_k) - c_\Y(y_j, y_l))^2 \pi^{t}_{kl}~.
\end{equation*}
When the base costs $c_X$ and $c_Y$ are CNT, the loss $\Loss$ is concave: it allows the step size $\gamma$ to be set to $1$, ensuring convergence without needing to tune this parameter. \cite{maron2018probably}.

\paragraph{Cheaper Gromov-Wasserstein Variants.} 
Many practical applications only require a quantification of similarity between $\alpha$ and $\beta$ rather than an exact matching.
In such cases, cheaper alternatives to the full GW problem may suffice.
This category includes \textsc{Sliced GW} \cite{titouan2019sliced}, \textsc{Quantized-GW} \cite{chowdhury2021quantized}, \textsc{Minibatch-GW} \cite{fatras2021minibatch} and \textsc{Linear-GW} \cite{beier2022linear}, as well as GW-inspired distances for Gaussian mixture models \cite{salmona2023gromov}.
While these heuristics can be computed in quadratic or even linear time, they rely on heavy approximations and are less suitable for tasks requiring precise point-wise correspondences such as shape matching.

\subsubsection{Entropic Gromov-Wasserstein Solvers}
\label{appendix:subsection:literature_egw_solvers}

\paragraph{Exact Entropic Solvers.} Entropic regularization, introduced by \cite{solomon2016entropic,peyre2016gromov}, allows GW to be solved using the efficient Sinkhorn algorithm.
The standard solver, \textsc{Entropic-GW}, is a Projected Gradient Descent (PGD) in the KL-divergence geometry.
This method, widely adopted in the literature, uses updates of the form: 
\begin{equation*}
\pi_{t+1} = \argmin_{\pi \in \coupling{\alpha}{\beta}} \left( \int  \grad_{\pi_t}{\Loss} \cdot \diff\pi + \varepsilon \KL{\pi}{\alpha \otimes \beta} \right).
\end{equation*}
These updates are similar to the Franke-Wolfe steps of Eq.~\eqref{eq:franke_wolfe} and can be viewed as a GW adaptation of the \emphasize{soft assign} method proposed in \cite{gold2002graduated}. 
Focusing on Unbalanced GW, \cite{sejourne2021unbalanced} also propose a bilinear relaxation of the EGW loss which can be solved by alternate minimization:
\begin{equation*}
\pi_{t+1} = \argmin_{\pi \in \coupling{\alpha}{\beta}} \left[ \int \left( \norm{x-x'}^2-\norm{y-y'}^2 \right)^2 \diff\pi(x,y) \diff\pi_t(x',y') + \varepsilon \KL{\pi}{\alpha \otimes \beta} \right].
\end{equation*}
For CNT base costs, this relaxation is tight and its minimum coincides with the true EGW solution.
In the finite case, this technique is equivalent to \textsc{Entropic-GW} and consists in solving a sequence of EOT problems:
\begin{equation}
\label{eq:pgd_finitecase}
    \pi^{t+1} = \argmin_{\pi \in \coupling{\alpha}{\beta}} \left( \sum C^{t}_{ij} \pi_{ij}  + \varepsilon \KL{\pi}{\alpha \otimes \beta} \right) \quad \text{ where } \quad  C^{t}_{ij} = \sum_{kl} (c_\X(x_i, x_k) - c_\Y(y_j, y_l))^2 \pi^{t}_{kl}~.
\end{equation}
Using the Sinkhorn algorithm, these steps can be computed in $\mathcal{O}(\NN\MM)$ time.
The overall complexity, however, is dominated by the computation of the cost matrix $C^t$ that takes $\mathcal{O}(\NN^2 \MM^2)$ time.
In practice, the same solution can be obtained with $C^{t}_{ij} = - 2 \sum_{k} c_\X(x_i, x_k) \sum_l c_\Y(y_j, y_l) \pi^{t}_{kl}$ instead \cite{peyre2016gromov}, so that the overall running time becomes cubic.
When costs are squared norms, \cite{scetbon2022linear} further reduce this computation to $\mathcal{O}(\NN\MM\DD)$; similarly, the dual solver \textsc{Dual-GW} of \cite{rioux2024entropic} achieves quadratic complexity and is guaranteed to converge.

\paragraph{Accelerated Approximations.} Apart from the squared norm case, solving EGW remains expensive due to the computation of the matrix $C^t$. 
The algorithm complexity can still be improved thanks to several approximation strategies introduced in the last few years.
\textsc{Sampled-GW} \cite{kerdoncuff2021sampled} computes $C^{t}$ using a random subset of indices $k$ and $l$ in Eq.~\eqref{eq:pgd_finitecase}, while \textsc{Sparse-GW} \cite{li2023efficient} applies a sparsity mask to fasten both cost matrix computations and Sinkhorn iterations.
Alternatively, \textsc{Quadratic-LowRank-GW} \cite{scetbon2022linear} approximates the input matrices $(C^\X_{ij}) = (c_\X(x_i, x_j))$ and $(C^\Y_{ij}) = (c_\Y(y_i,y_j))$ via low-rank factorization to accelerate the estimation of $C^t$. 
These three algorithms reach a quadratic running time, improving significantly the scalability of the original EGW solvers.
With \textsc{LowRank-GW}, \cite{scetbon2022linear} combine the low-rank cost reduction of \textsc{Quadratic-LowRank-GW} with a low rank constraint on the transport plan, further reducing computation time to $\mathcal{O}(\NN+\MM)$ but generating quantization artifacts in the transport plan.
Finally, \textsc{ULOT} (Unsupervised Learning of Optimal Transport) \cite{mazelet2025unsupervised} trains neural networks to predict optimal couplings, removing the optimization burden at inference time -- although its training phase keeps a cubic complexity.

\paragraph{Proximal Solvers.} A last strategy consists in adding a proximal term to the GW loss rather than an entropic regularization \cite{xu2019scalable,xu2019gromov,kerdoncuff2021sampled}. 
The updates of \textsc{Proximal-GW} are given by:
\begin{equation*}
\pi_{t+1} = \argmin_{\pi \in \coupling{\alpha}{\beta}} \left( \int  \grad_{\pi_t}{\Loss} \cdot \diff\pi + \varepsilon \KL{\pi}{\pi_t} \right).
\end{equation*}
This approach ensures the convergence of $(\pi_t)$ to the minimum of the \emphasize{unregularized} GW problem while remaining Sinkhorn-compatible.
Indeed, in the finite case, it can be solved using:
\begin{equation*}
    \pi^{t+1} = \argmin_{\pi \in \coupling{\alpha}{\beta}} \left( \sum C^{t}_{ij} \pi_{ij}  + \varepsilon \KL{\pi}{\alpha \otimes \beta} \right) \quad \text{ where } \quad  C^{t}_{ij} = - 2 \sum_{kl} c_\X(x_i,x_k) c_\Y(y_j, y_l) \pi^{t}_{kl} - \varepsilon \log (\pi^{t}_{ij})~.
\end{equation*}
This alternative is attractive when exact, unregularized GW solution are required.
However, it has other limitations.
The proven convergence rates only apply to finite measures and degrade as the number of points increases \cite{xie2020fast}.
The sequential dependency of the proximal term (due to the $\log(\pi^t)$ term) complicates its integration with the advanced scalability techniques described in \cref{section:implementation}, such as lazy tensor manipulation with KeOps or multiscaling, which are essential for large-scale applications.
Finally, in many machine learning contexts, the entropic term acts as a regularizer against noisy input, making EGW preferable to the unregularized GW formulation.

\subsubsection{Choice of Baselines}

Despite the wide variety of Gromov-Wasserstein solvers available, only entropic-regularized methods (\cref{appendix:subsection:literature_egw_solvers}) scale to inputs containing thousands of points, and only accelerated approximations can handle tens of thousands. 
Consequently, we focus our benchmarks on two main competitors: \textsc{Entropic-GW}, the reference solver for exact EGW problems, and \textsc{Quadratic-LowRank-GW}, the current state-of-the-art for approximate solving.
We discarded heuristic methods and linear-time approximations as they introduce artifacts that are unsuited to our shape matching applications (\cref{fig:horse_texture_comparisons}g and h). 
Finally, we excluded other potential baselines for the following reasons: \begin{itemize} 
\item \textsc{Sampled-GW} relies on stochastic sampling: due to this randomness, the solver does not converges effectively which prevents reliable comparisons of convergence speed. 
\item \textsc{Proximal-GW} addresses a distinct problem formulation, and does not scale to $\NN>10^4$ points. 
\item Methods such as \textsc{Sparse-GW} or \textsc{ULOT} rely on different computational backends or address different settings, making fair comparisons difficult to perform. 
\end{itemize}

Note that although we did not include them in our quantitative benchmarks, we still discuss the convergence properties of \textsc{Sampled-GW} and \textsc{Proximal-GW} in \cref{subsection:appendix:solver_cvg}. 
Finally, we stress that the quadratic-time methods excluded from our benchmarks were originally validated on inputs with at most a few thousand nodes. 
Therefore, while we do not provide a direct performance comparison against all pre-existing GW solvers, our experiments demonstrate scalability to problems whose sizes are $10$ to $100$ times larger than those addressed in prior works.

\newpage

\section{Detailed Implementation}
\label{section:detailed_implementation}
We detail here the pseudo-code of the main routines used in our method.

\subsection{Sinkhorn Algorithms}
\label{subsection:sinkhorn_algorithms}

\paragraph{Standard Sinkhorn.} \cref{alg:sinkhorn} details the standard version of the Sinkhorn algorithm.
It implement the discrete counterpart of the Sinkhorn operators defined in \cref{eq:sinkhorn_operators}, adressing the discrete EOT problem:
\begin{equation*}
    \min_{\pi \in \R^{\NN\times \MM}} \left[ \sum_{ij} C_{ij} \pi_{ij} + \varepsilon \sum_{ij} \pi_{ij} \log \left(\frac{\pi_{ij}}{a_i b_j} \right)\right] \quad \text{ subject to } \quad \sum_j \pi_{ij} = a_i \quad \text{ and } \quad \sum_i \pi_{ij} = b_j,
\end{equation*}
where $C \in \R^{\NN \times \MM}$ is a cost matrix and $a \in \R^\NN, b \in \R^\MM$ are two positive vectors summing to $1$.
From the final dual potentials $(f^\ninner, g^\ninner)$, we recover the approximate optimal plan using \cref{eq:sinkhorn_pi}:
\begin{equation*}
    \forall~i,j, \quad \pi_{ij} = a_i b_j \exp((f^\ninner_i + g^\ninner_j - C_{ij})/\varepsilon).
\end{equation*}

\begin{algorithm}
\caption{\textsc{Sinkhorn}.}
\label{alg:sinkhorn}
\begin{algorithmic}
\STATE {\bfseries Parameters:} $\varepsilon > 0$, $\ninner > 0$.
\STATE {\bfseries Inputs:} $a \in \R^\NN,b \in \R^\MM$, $C \in \R^{\NN\times \MM}$.
\STATE Initialize $f^0 \in \R^\NN, g^0 \in \R^\MM$
\FOR{$n = 0$ to $\ninner-1$}
    \STATE $f^{n+1} \leftarrow - \varepsilon \cdot [\log \sum_j \exp (\log(b_j) + (g^n_j - C_{ij}) / \varepsilon ) ]_i$
    \STATE $g^{n+1} \leftarrow - \varepsilon \cdot [\log \sum_i \exp (\log(a_i)+ (f^{n+1}_i - C_{ij}) / \varepsilon )]_j$
\ENDFOR
\STATE {\bfseries Return $f^\ninner, g^\ninner$}
\end{algorithmic}
\end{algorithm}

\paragraph{Symmetrized Sinkhorn.} \cref{alg:sinkhorn_symm} presents the symmetrized variant of Sinkhorn introduced in \citep[Algorithm 3.4]{feydy2020analyse}.
This variant ensures that the resulting pair $(f^\ninner, g^\ninner)$ remains the same if the roles of $\alpha$ and $\beta$ are exchanged with each other.
As discussed in \cref{subsection:appendix:solver_cvg}, we also notice a drastic convergence speed-up when using symmetrized Sinkhorn in EGW solvers.

\begin{algorithm}
\caption{\textsc{Sinkhorn-symm}.}
\label{alg:sinkhorn_symm}
\begin{algorithmic}
\STATE {\bfseries Parameters:} $\varepsilon > 0$, $\ninner > 0$.
\STATE {\bfseries Inputs:} $a \in \R^\NN,b \in \R^\MM$, $C \in \R^{\NN\times \MM}$.
\STATE Initialize $f^0 \in \R^\NN, g^0 \in \R^\MM$
\FOR{$n = 0$ to $\ninner-1$}
    \STATE $f^{n+1} \leftarrow \frac{1}{2} (f^n -  \varepsilon \cdot [\log \sum_j \exp (\log(b_j) + (g^n_j - C_{ij}) / \varepsilon )]_i)$
    \STATE $g^{n+1} \leftarrow \frac{1}{2} (g^n - \varepsilon \cdot [\log \sum_i \exp (\log(a_i) + (f^n_i - C_{ij}) / \varepsilon)]_j)$
\ENDFOR
\STATE {\bfseries Return $f^\ninner, g^\ninner$}
\end{algorithmic}
\end{algorithm}

\paragraph{Warm-Start Initialization.} The number of iterations needed to reach convergence is highly sensitive to the initialization of the potentials $(f^0, g^0)$.
Since EGW solvers involve a sequence of successive Sinhorn calls, we employ a warm-start strategy: the optimal potentials $(f, g)$ from one optimization step are used to initialize the Sinkhorn loop of the next step.
This greatly reduces the number of iterations required for convergence.

\paragraph{Sinkhorn Annealing.} In order to get sharp transport plans, a common strategy consists is to decrease the regularization parameter $\varepsilon$ at every Sinkhorn iteration.
As proven in \cite{chizat2024annealed}, using a squared root decay $\varepsilon_n = \varepsilon_0 / \sqrt{n}$ guarantees the convergence of the reconstructed transport plans towards the unregularized OT solution.
We use this annealing scheme to apply the transfer texture of \cref{fig:introduction}, only in the last step of the EGW solver.

\subsection{Kernel matrices computations.} 

\label{subsection:appendix_kernel}

\paragraph{Centered Kernels From Costs.}\cref{alg:kernel} computes the kernel matrix $K$ associated with a CNT cost matrix $(C_{ij}) = (c(x_i,x_j))$ using the equation of \cref{prop:cnt_positive_kernel}:
\begin{equation}
\label{eq:kernel_from_cost}
    k(x,x') ~=~ \big( c(x, x_0) + c(x_0, x') - c(x, x') \big) \,/\, 2~.
\end{equation}
The procedure includes a centering step (line 3) which ensures that the embeddings are centered in the kernel space. 
This centering also ensures that the output is independent of the choice of base point $x_0$; we can therefore take $1$ as base index in line 2 of the algorithm without inducing any bias.

\begin{algorithm}
\caption{\textsc{Kernel}}
\label{alg:kernel}
\begin{algorithmic}
\STATE {\bfseries Inputs:} $C \in \R^{\NN \times \NN}, a \in \R^\NN$.
\STATE $K \leftarrow [(C_{i1} + C_{1j} - C_{ij}) / 2]_{ij}$
\STATE $K \leftarrow [K_{ij} - \sum_i a_i K_{ij} - \sum_j a_j K_{ij} + \sum_{ij} a_i a_j K_{ij}]_{ij}$
\STATE {\bfseries Return $K$}
\end{algorithmic}
\end{algorithm}

\paragraph{Kernel PCA.} The Kernel PCA of a centered kernel matrix $K \in \R^{\NN \times \NN}$ is computed from its eigendecomposition.
Let $\Lambda_{1:\DD} \in \R^{\DD \times \DD}$ denote the diagonal matrix containing the $\DD$ largest eigenvalues of $K$ and $V_{1:\DD} \in \R^{\NN \times \DD}$ the corresponding eigenvectors.
The Kernel PCA of $K$ in dimension $\DD$ is then given by $X_i = V_{i,1:\DD} \sqrt{\Lambda_{1:\DD}} \in \R^{\DD}$ for all $i \in \{1,\dots,\NN \}$.

When the cost $c$ is given by a simple mathematical formula, e.g. an analytical function of the Euclidean norm, the kernel of \cref{eq:kernel_from_cost} admits a simple formula as well and the kernel matrix $K$ can be implemented as a \emph{symbolic tensor} with the \texttt{KeOps} library.
Combined with the function \texttt{eigsh} of the \texttt{scipy} library, the truncated PCA of $K$ can thus be obtained without explicitly storing the full matrix.
This allows us to perform Kernel PCAs on $\NN > 10^4$ points without memory overflows.

\subsection{Exact Kernel Embeddings with \textsc{Kernel-GW}}
\label{appendix:subsection:detailed_kernel_impl}

Rather than relying on approximate embeddings, the alternate minimization of \cref{eq:alternate_minimization_scheme} can be implemented exactly using the \emph{kernel trick}.
If $k_\X$ and $k_\Y$ are the kernels obtained from $c_\X$ and $c_\Y$ by \cref{prop:cnt_positive_kernel}, we have: 
\begin{equation*}
    \scal{\varphi(x) \psi(y)^{\top}}{\varphi(x') \psi(y')^{\top}}_{\HS} = k_\X(x,x') k_\Y(y,y')~.
\end{equation*}
Denoting by $\Gamma = \sum_{ij} \varphi(x_i) \psi(y_j)^{\top} \pi_{ij}$ the cross-correlation of a given transport plan $\pi \in \R^{\NN \times \MM}$ , the cost $c_\Gamma$ becomes:
\begin{equation}
\label{eq:c_gamma_kerneltrick}
    c_\Gamma(x,y) = -4 \cdot k_\X(x,x)k_\Y(y,y) - 16 \cdot \sum_{k,l} k_\X(x,x_k)k_\Y(y,y_l)\pi_{kl}~.
\end{equation}
Note that for the sake of simplicity, we include the factor $8$
that was factored out in our statement of \cref{prop:entropic_gw_cnt_linearized}.
\cref{alg:kernel_gw} thus implements the optimization scheme without handling $\Gamma$ explicitly (where $C^X, C^Y$ are the cost matrices of $(x_1,\dots,x_\NN)$ and $(y_1,\dots,y_{\MM})$).
\textsc{Kernel-GW} provides few computational benefits compared to \textsc{Entropic-GW}, as both methods store a large matrices of size $\NN \times \MM$.
Yet, it will help us to isolate the specific effects of kernelization on solver convergence in \cref{subsection:appendix:solver_cvg}, discarding the approximation errors induced by \textsc{CNT-GW} and \textsc{LowRank-GW}.

\begin{algorithm}[ht]
\caption{\textsc{Kernel-GW}.}
\label{alg:kernel_gw}
\begin{algorithmic}
\STATE {\bfseries Parameters:} $\varepsilon > 0$, $\nouter>0$, $\ninner > 0$.
\STATE {\bfseries Inputs:} $a \in \R^{\NN}$, $b \in \R^{\MM}$, $C^{X} \in \R^{\NN \times \NN}$, $C^{Y} \in \R^{\MM \times \MM}$.
\STATE $K^X, K^Y \leftarrow \textsc{Kernel}(C^X, a),~\textsc{Kernel}(C^Y, b)$
\STATE Initialize $\pi^0$
\FOR{$t = 0$ to $\nouter-1$}
    \STATE $ C \leftarrow [-4 \cdot K^X_{ii} K^Y_{jj}
           - 16 \cdot \sum_{kl} K^X_{ik} K^Y_{jl} \pi^t_{kl}~]_{ij}$
    \STATE $f^{t+1}, g^{t+1} \leftarrow \textsc{Sinkhorn-symm}_{\varepsilon,\ninner}(a, b, C)$
    \STATE $\pi^{t+1} = [a_i b_j \exp((f^{t+1}_i + g^{t+1}_j - C_{ij})/\varepsilon)]_{ij}$
\ENDFOR
\STATE {\bfseries Return} $\pi^\nouter, f^\nouter, g^\nouter$
\end{algorithmic}
\end{algorithm}

\subsection{Adaptive Sinkhorn Iterations.} 

In \cref{alg:kernel_gw_adaptive}, we present the adaptive variant of \cref{alg:kernel_gw} where the fixed iteration count $\ninner$ is replaced by an adaptive schedule.
While an initial value $n_{\text{inner}}^0$ is still required, the algorithm automatically increases precision when needed: therefore, this hyperparameter can be set to a small arbitrary value.
Since that $\textsc{EGW-Loss}(\pi^t)$ is guaranteed to decrease at each step when $\pi^t$ corresponds to the true Sinkhorn solution, any violation of this decrease indicates a lack of Sinkhorn convergence: therefore, the adaptive method is guaranteed to converge to a EGW local minimum for $t \rightarrow + \infty$ independently of the choice $n_{\text{inner}}^0$ (although this initial value can influence which local minimum will be attained).

Adaptive scheduling can also be applied to \textsc{CNT-GW}.
In this case, we use the squared Euclidean norm of the embedding space as the base cost for the computations of \textsc{EGW-Loss}, since it can be computed more efficiently.

\begin{algorithm} 
\caption{\textsc{Kernel-GW} with adaptive Sinkhorn iterations.}
\label{alg:kernel_gw_adaptive}
\begin{algorithmic}
\STATE {\bfseries Parameters:} $\varepsilon > 0$, $T > 0$, $n_{\text{inner}}^0 > 0$.
\STATE {\bfseries Inputs:} $a \in \R^\NN,b \in \R^\MM$, $C^X \in \R^{\NN\times \NN}$, $C^X \in \R^{\MM\times \MM}$.
\STATE $K^X, K^Y \leftarrow \textsc{Kernel}(C^X, a),~\textsc{Kernel}(C^Y, b)$
\STATE Initialize $\pi^0$
\FOR{$t = 0$ to $\nouter-1$}
    \STATE $ C \leftarrow [-4 \cdot K^X_{ii} K^Y_{jj}
           - 16 \cdot \sum_{kl} K^X_{ik} K^Y_{jl} \pi^t_{kl}~]_{ij}$
    \STATE $f^{t+1}, g^{t+1} \leftarrow \textsc{Sinkhorn-symm}_{\varepsilon,n_{\text{inner}}^t}(a, b, C)$
    \STATE $\pi^{t+1} = [a_i b_j \exp((f^{t+1}_i + g^{t+1}_j - C_{ij})/\varepsilon)]_{ij}$
    \IF{$\textsc{EGW-Loss}(\pi^{t+1}) > \textsc{EGW-Loss}(\pi^{t})$}
    \STATE $n_{\text{inner}}^{t+1} \leftarrow 2 \cdot n_{\text{inner}}^t$
    \ELSE 
    \STATE $n_{\text{inner}}^{t+1} \leftarrow n_{\text{inner}}^t$
    \ENDIF
\ENDFOR
\STATE {\bfseries Return} $\pi^\nouter, f^\nouter, g^\nouter$
\end{algorithmic}
\end{algorithm}

\subsection{Implementation of \textsc{Multiscale-GW}}
\label{subsection:multiscale_gw}

\cref{alg:multiscale_dual_kernelpca} directly adapts the multiscale technique of \cite{feydy2020analyse} to \textsc{CNT-GW}.
The $\textsc{Coarsen}$ function is implemented using K-Means clustering; the coarse points $X^c \in \R^{\NN^c \times \DD}$ are the K-Means centroids and the weights $a^c$ are the aggregated weights of the clusters:
\begin{equation*}
    a^c_i = \sum_{j \in \mathcal{C}_i} a_j \quad \text{where} \quad \mathcal{C}_i = \left\{j = 1,\dots,\NN \; \text{ s.t. } \; i = \argmin_{k=1,\dots,\NN^c}{\norm{X_j - X^c_k}^2} \right\}.
\end{equation*}
The Sinkhorn potentials $(f, g)$ are upsampled from coarse to fine using the block-wise interpolation described in \cref{alg:multiscale_interpolate}, and are used together with the coarse dual solution $\Gamma^c$ to initialize the large-scale solver. 

\begin{algorithm}
\caption{\textsc{Multiscale-GW}.}
\begin{algorithmic}
\label{alg:multiscale_dual_kernelpca}
\STATE {\bfseries Parameters:} $\varepsilon > 0$, $D > 0$, $\rho \in [0, 1]$.
\STATE {\bfseries Inputs:} $a \in \R^\NN,b \in \R^\MM$, $X \in \R^{\NN \times \DD}$, $Y \in  Y^{\MM \times \EE}$.
\STATE $(X^{c}, a^{c}), (Y^{c}, b^{c}) \leftarrow \textsc{Coarsen}(X, a, k=\lceil\rho \NN \rceil),~~ \textsc{Coarsen}(Y, b, k=\lceil\rho \MM \rceil)$
\STATE $\Gamma^c, f^c, g^c \leftarrow \textsc{CNT-GW}(a^{c}, b^{c}, X^{c}, Y^{c})$
\STATE $f, g \leftarrow \textsc{Interpolate}(g^c, X, Y^{c}, b^c),~~ \textsc{Interpolate}(f^c, Y, X^{c}, a^c)$
\STATE $\Gamma, f, g \leftarrow \textsc{CNT-GW}(a, b, X, Y ~\vert~ \Gamma^0=\Gamma^c, (f^0,g^0)=(f,g))$
\STATE {\bfseries Return} $\Gamma, f, g$
\end{algorithmic}
\end{algorithm}

\begin{algorithm}
\caption{\textsc{Interpolate}.}
\begin{algorithmic}
\label{alg:multiscale_interpolate}
\STATE {\bfseries Inputs:} $g^c \in \R^{\MM^c}, X \in \R^{\NN \times \DD}, Y^{c} \in R^{\MM^c \times \DD}, b^c \in R^{\MM^c}$.
\STATE $C \leftarrow 
    [-4 \norm{X_i}^2\norm{Y^c_j}^2 - 16 \sum_{rs} X_{ir}\, Y^c_{js}\, \Gamma_{rs}]_{ij}$
\STATE $f \leftarrow - \varepsilon \cdot [\log \sum_j \exp (\log(b^c_j) + (g^c_j - C_{ij}) / \varepsilon ) ]_i$
\STATE {\bfseries Return} $f$
\end{algorithmic}
\end{algorithm}

\subsection{EGW Gradient Computations.}
\label{appendix:subsection:egw_grad}

We finally detail the computation of EGW gradients with respect to the input locations $x_i$ and $y_j$.
Specifically, we seek:
\begin{equation*}
    G^X_i = \grad_{x_i}{\GW_{\varepsilon}(\alpha, \beta)}~~\text{ and }~~G^Y_j = \grad_{y_j}{\GW_{\varepsilon}(\alpha, \beta)} \quad \text{ where }~~\alpha = \sum_i a_i \delta_{x_i}~~\text{ and }~~\beta = \sum_i b_j \delta_{y_j}.
\end{equation*}
We could compute them by differentiating through the whole EGW solver using PyTorch's autograd: however, this method would be costly in time and memory and would not scale on large point clouds.
Therefore, we need more efficient methods to perform these computations.
Assuming that the gradient is well defined, we rather apply the envelope theorem to write:
\begin{equation*}
    \grad_{x_i}{\GW_{\varepsilon}(\alpha, \beta)} = \grad_{x_i}C(\push{\varphi}{\alpha},\push{\psi}{\beta}) + \grad_{x_i}{\OT_{\varepsilon}^{\Gamma^{*}}}(\push{\varphi}{\alpha},\push{\psi}{\beta}).
\end{equation*}
The main complexity lies in the presence of the kernel embeddings $\varphi$ and $\psi$ that we need to differentiate through.
We provide two techniques to perform these computations:
the first one, based on the \textsc{Kernel-GW} algorithm, outputs exact EGW gradient while the second one provides approximate gradients based using \textsc{CNT-GW}.

\paragraph{Exact Gradient Computation} The kernel trick provides an exact method to compute EGW gradients.
First, the constant can be differentiated automatically with PyTorch using the explicit formula:
\begin{equation*}
    C(\push{\varphi}{\alpha},\push{\psi}{\beta}) = \sum_{ik} c_\X(x_i,x_k)^2 a_i a_k + \sum_{jl} c_\Y(y_j,y_l)^2 b_j b_l - 4 \sum_{i} k_\X(x_i,x_i)a_i \sum_{j} k_\Y(y_j,y_j)b_j,
\end{equation*}

To differentiate the EOT loss, we use the method of \cite{feydy2019interpolating} recalled in \cref{alg:eot_grad} (where the apostrophes indicate tensors whose auto-differentiation is enabled) and the fact that $c_{\Gamma^*}(x_i,y_j)$ is explicitly computable using \cref{eq:c_gamma_kerneltrick}.
The whole implementation of \textsc{EGW-Grad} is given in \cref{alg:egw_grad}, where \textsc{KernelGW} outputs both the optimal coupling $\pi$ and the optimal potentials $(f,g)$ of the last Sinkhorn step.

\begin{algorithm}
\caption{\textsc{EOT-DualLoss}.}
\label{alg:eot_dualloss}
\begin{algorithmic}
\STATE {\bfseries Parameters :} $\varepsilon > 0$. 
\STATE {\bfseries Inputs:} $a \in \R^\NN,b \in \R^\MM$, $f \in \R^\NN$, $g \in \R^\MM$, $C \in \R^{\NN \times \MM}$.
\STATE $f \leftarrow - \varepsilon \cdot [\log \sum_j \exp (\log(b_j) + (g_j - C_{ij}) / \varepsilon ) ]_i$
\STATE $g \leftarrow - \varepsilon \cdot [\log \sum_i \exp (\log(a_i)+ (f_i - C_{ij}) / \varepsilon )]_j$
\STATE $L \leftarrow \sum_i a_i f_i + \sum_j b_j g_j$
\STATE {\bfseries Return $L$}
\end{algorithmic}
\end{algorithm}

\begin{algorithm}
\caption{\textsc{EOT-Grad}.}
\label{alg:eot_grad}
\begin{algorithmic}
\STATE {\bfseries Parameters :} $\varepsilon > 0$. 
\STATE {\bfseries Inputs:} $a \in \R^\NN,b \in \R^\MM$, $x_i \in \X$, $y_j \in \Y$.
\STATE $f, g \leftarrow \textsc{Sinkhorn}_{\varepsilon}(a, b, C = [c(x_i, y_j)]_{ij})$
\STATE $x'_i, y'_j \leftarrow x_i, y_j \text{ with auto-diff. activated}$ 
\STATE $L' \leftarrow \textsc{EOT-DualLoss}(a, b, f, g, C' = [c(x'_i, y'_j)]_{ij})$
\STATE $G^X, G^Y \leftarrow \text{ Gradient of } L' \text{ w.r.t. } x'_i, y'_j$
\STATE {\bfseries Return $G^X, G^Y$}
\end{algorithmic}
\end{algorithm}

\begin{algorithm}
\caption{\textsc{EGW-Grad}.}
\label{alg:egw_grad}
\begin{algorithmic}
\STATE {\bfseries Parameters :} $\varepsilon > 0$. 
\STATE {\bfseries Inputs:} $a \in \R^\NN,b \in \R^\MM$, $x_i \in \X$, $y_j \in \Y$.
\STATE $P, f, g \leftarrow \textsc{Kernel-GW}_{\varepsilon}(a, b, C^X = [c_\X(x_i,x_k)]_{ik}, C^Y = [c_\Y(y_j,y_l)]_{jl})$.
\STATE $x'_i, y'_j \leftarrow x_i, y_j \text{ with auto-diff. activated}$ 
\STATE $L' \leftarrow \textsc{EOT-DualLoss}(a, b, f, g, C' = [-4 \cdot k_\X(x'_i, x'_i) k_\Y(y'_j, y'_j) - 16 \cdot \sum_{kl} k_\X(x'_i, x_k) k_\Y(y'_j, y_l) \pi_{kl} ]_{ij})$
\STATE $L'_0 \leftarrow \sum_{ik} c_\X(x'_i,x'_k)^2 a_i a_k + \sum_{jl} c_\Y(y'_j,y'_l)^2 b_j b_l - 4 \sum_{i} k_\X(x'_i,x'_i)a_i \sum_{j} k_\Y(y'_j,y'_j)b_j$ 
\STATE $G^X, G^Y \leftarrow \text{ Gradient of } L' + L'_0 \text{ w.r.t. } (x'_i), (y'_j)$

\STATE {\bfseries Return $G^X, G^Y$}
\end{algorithmic}
\end{algorithm}

\paragraph{Approximation with Kernel PCA.}

The previous method outputs exact gradients, but it scales quadratically in memory. 
To further improve computational complexity, we provide a method relying on Kernel PCA instead: the principle is to differentiate through kernel embeddings, then differentiate through \textsc{CNT-GW} using the same technique as before.

Although we cannot differentiate through Kernel PCA directly, we use the following strategy.
Let us denote by $\Lambda_1, \dots, \Lambda_\DD \in \R$ and $V_1, \dots, V_\DD \in \R^\NN$ the eigendecomposition of the centered kernel matrix $K_X = [k_X(x_i, x_j)]_{ij}$.
The \emph{kernel projection map} $\Pi_\X : \X \mapsto \R^{\DD}$ sends any point $x \in \X$ to its kernel embedding by interpolating the embeddings of $x_1,\dots,x_\NN$:
\begin{equation*}
    \Pi_\X(x) = \left[\frac{1}{\sqrt{\Lambda_d}} \sum_{i=1}^\NN k_\X(x, x_i) V_{di} \right]_{d=1,\dots,\DD}.
\end{equation*}
This map is differentiable with respect to $x$, with:
\begin{equation*}
    \grad_{x}{\Pi_\X} = \frac{1}{\sqrt{\Lambda_d}} \sum_{i=1}^\NN \grad_{x}{k_\X(x, x_i)} V_{di}.
\end{equation*}
We can therefore use $\grad_{x_i}{\Pi_\X}$ as an approximation of the differentiation of the Kernel PCA with respect to $x_i$.
The corresponding code is provided in \cref{alg:egw_grad_kernelpca}.

\begin{algorithm}
\caption{\textsc{EGW-Grad} (with Kernel PCA).}
\label{alg:egw_grad_kernelpca}
\begin{algorithmic}
\STATE {\bfseries Parameters :} $\varepsilon > 0, \DD > 0, \EE > 0$. 
\STATE {\bfseries Inputs:} $a \in \R^\NN,b \in \R^\MM$, $x_i \in \X$, $y_j \in \Y$.
\STATE $X, V^X_{1:\DD}, \Lambda^X_{1:\DD} \leftarrow \textsc{Kernel-PCA}(x_i, k_\X, \DD)$ with eigendecomposition
\STATE $Y, V^Y_{1:\EE}, \Lambda^Y_{1:\EE} \leftarrow \textsc{Kernel-PCA}(x_j, k_\Y, \EE)$ with eigendecomposition
\STATE $\Gamma, f, g \leftarrow \textsc{CNT-GW}(a, b, X, Y)$
\STATE $x'_i, y'_j \leftarrow x_i, y_j \text{ with auto-diff. activated}$ 
\STATE $X'_i \leftarrow \textsc{Kernel-Proj}(x'_i,x_i,V^X_{1:\DD}, \Lambda^X_{1:\DD})$
\STATE $Y'_j \leftarrow \textsc{Kernel-Proj}(y'_j,y_j,V^Y_{1:\EE}, \Lambda^Y_{1:\EE})$
\STATE $L' \leftarrow \textsc{EOT-DualLoss}(a, b, f, g, C' = [-4 \cdot \norm{X'_i}^2\norm{X'_j}^2 - 16 \cdot \sum_{k} X'_{ik} \sum_l (Y'_{jl} \Gamma_{kl})]_{ij})$  \quad with $C'$ as symbolic matrix
\STATE $L'_0 \leftarrow \sum_{ik} \norm{X'_i - X'_k}^4 a_i a_k + \sum_{jl} \norm{Y'_j - Y'_l}^4 b_j b_l - 4 \cdot \sum_{i} \norm{X'_i}^2 a_i \sum_j \norm{Y'_j}^2 b_j$ 
\STATE $G^X, G^Y \leftarrow \text{ Gradient of } L' + L'_0 \text{ w.r.t. } (x'_i), (y'_j)$
\STATE {\bfseries Return $G^X, G^Y$}
\end{algorithmic}
\end{algorithm}

\begin{algorithm}
\caption{\textsc{Kernel-Proj}.}
\label{alg:kernel_proj}
\begin{algorithmic}
\STATE {\bfseries Inputs:} $x \in \X$, $x_i \in \X$,$V_{1:\DD} \in \R^{\NN \times \DD}$,$\Lambda_{1:\DD} \in \R^{\DD}$.
\STATE $X \leftarrow [\tfrac{1}{\sqrt{\Lambda_d}} \sum_{i} k_X(x, x_i) V_{di} ]_{d=1,\dots,\DD}$
\STATE {\bfseries Return $X$}
\end{algorithmic}
\end{algorithm}

\newpage

\section{Proofs}
We first recall several definitions necessary to the proofs of this paper.

\textbf{Pushforward.} Given a measure $\alpha \in \meas{\X}$ and a map $\func{\Psi}{X}{Y}$, the\emphasize{pushforward} of $\alpha$ by $\Psi$ (dentoed  $\push{\Psi}{\alpha}$) is the only measure of $\meas{Y}$ such that for all continuous function $h \in \Cont{Y}$, $\quad \int h(y) \diff(\push{\Psi}{\alpha})(y) = \int h(\Psi(x))\diff\beta(x)$.

\textbf{Weak convergence of measures.} A sequence of measures $(\alpha_n) \in \meas{\X}$ converges weakly to $\alpha \in \meas{\X}$ (denoted $\alpha_n  \rightharpoonup \alpha $) if for any continuous function $f \in \Cont{\X}$, $\int\!f \diff\alpha_n \rightarrow \int\!f \diff\alpha$.

\textbf{Hilbert-Schmidt norm and scalar product.} Let $\Gamma \in \Hspace$ be a HS operator.
Its HS norm is equal to $\norm{\Gamma}_{\HS}^2 := \sum_i \norm{\Gamma e_i}_{\Hilb_Y}^2$, where $(e_i)$ is an orthonormal basis of $\Hilb_\X$ (this definition is independent of the choice of $(e_i)$). If $x, x' \in \Hilb_X$ and $y,y' \in \Hilb_Y$, the HS scalar product satisfies $\scal{x y^{\top}}{x' y'^{\top}}_{\HS} = \scal{x}{x'}_{\Hilb_X} \scal{y}{y'}_{\Hilb_Y}$.

\subsection{Proofs of \cref{section:background} and \ref{section:gw_cntcosts}}

\gwisometry*
\begin{proof}
This is result have been proven in \citep[Theorem 5.1]{memoli2011gromov} when $c_\X, c_\Y$ are metric distances, and their proof does not use the triangular inequality.
We will recall their demonstration here, showing that it generalizes to our setting.

If $\alpha$ and $\beta$ are isometric, then $\GW(\alpha,\beta)$ is trivially equal to $0$.
Reciprocally, let $\alpha \in \prob{\X}$, $\alpha \in \prob{\Y}$ such that $\GW(\alpha,\beta)=0$.
There exists a $\pi \in \coupling{\alpha}{\beta}$ such that:
\begin{equation}
\label{eq:loss_zero}
    \int (c_\X(x,x')\!-\!c_\Y(y,y'))^2 \diff\pi(x,y)\diff\pi(x',y') = 0, \quad { i.e. } \quad  \forall~(x,y), (x',y') \in \supp{\pi}, \quad c_\X(x,x') = c_\Y(y,y').
\end{equation}

Let $x \in \X$ and $y,y' \in \Y$.
If $(x,y),(x,y') \in \supp{\pi}$, then \cref{eq:loss_zero} implies that $c_\Y(y,y') = c_\X(x,x) = 0$, and $y=y'$ by hypothesis on $c_\Y$: therefore, for each $x \in \supp{\alpha}$, there is a unique $y \in \supp{\beta}$ such that $(x,y) \in \supp{\pi}$.
Reciprocally, for each $y \in \supp{\beta}$ there is a unique $x \in \supp{\alpha}$ such that $(x,y) \in \supp{\pi}$.
This proves the existence of a bijection $\func{I}{\supp{\alpha}}{\supp{\beta}}$ such that $\supp{\pi} = \{(x,I(x)), x \in \X \}$.

From \cref{eq:loss_zero}, we immediately have that all $(x,I(x)), (x',I(x')) \in \supp{\pi}$ satisfy $c_\X(x,x') = c_\Y(I(x),I(x'))$. Moreover, using the fact that $y = I(x)$ for almost every $(x,y) \in \supp{\pi}$, we write:
\begin{equation*}
    \text{For all}\; h \in \Cont{\Y}, \quad \int h(y) \diff\beta(y) = \int h(y) \diff\pi(x,y) = \int h(I(x)) \diff\pi(x,y) = \int h(\Psi(x)) \diff\beta(y).
\end{equation*}
Therefore, $\beta = \push{I}{\alpha}$, i.e. $I$ pushes $\alpha$ onto $\beta$, which proves the result.
\end{proof}

To prove the main EGW decomposition theorem, we need the following lemma:

\begin{lemma}
\label{prop:cnt_gw_developed}
Let $\X$ and $\Y$ be Hilbert spaces and $c_\X, c_\Y$ be the corresponding squared Hilbert norms.
Let $\alpha$, $\beta$ be compactly supported probability measures over $\X$ and $\Y$.
The GW loss decomposes as:
\begin{equation*}
    \text{For all }~\pi \in \coupling{\alpha}{\beta}, \quad \Loss(\pi) = C(\alpha, \beta) - 8 \cdot \norm{ \int x y^{\top} \diff\pi(x,y)}_{\HS}^2 - 4 \cdot \int \norm{x}^2\norm{y}^2 \cdot \diff\pi(x,y),
\end{equation*}
where $C(\alpha, \beta)$ is a constant that only depends on the marginals $\alpha$ and $\beta$:
\begin{equation*}
C(\alpha, \beta) =  \int \norm{x - x'}^4 \diff\alpha(x)\diff\alpha(x') + \int \norm{y - y'}^4 \diff\beta(y)\diff\beta(y') - 4 \cdot \int \norm{x}^2 \diff\alpha(x) \int \norm{y}^2 \diff\beta(y).
\end{equation*}
\end{lemma}

\begin{proof}
We follow the same computations as \cite{zhang2024gromov}.
Let $\pi \in \coupling{\alpha}{\beta}$. We first develop:
\begin{equation*}
 \Loss(\pi) = \int \norm{x - x'}^4\diff \pi(x,y) \diff\pi(x',y') + \int \norm{y - y'}^4 \diff \pi(x,y) \diff\pi(x',y') - 2 \int \norm{x - x'}^2 \norm{y - y'}^2 \diff \pi(x,y) \diff\pi(x',y').
\end{equation*}
Since $\pi \in \coupling{\alpha}{\beta}$, we have $\int f(x) \diff\pi(x, y) = \int f \cdot \alpha$ and $\int g(y) \diff\pi(x, y) = \int g \cdot \beta$ for any measurable functions $f$ and $g$ on $\Hilb_X$ and $\Hilb_Y$.
Therefore:
\begin{gather*}
\int \norm{x - x'}^4\diff \pi(x,y) \diff\pi(x',y') = \int \norm{x - x'}^4 \diff\alpha(x)\diff\alpha(x'), \\
\int \norm{y - y'}^4\diff \pi(x,y) \diff\pi(x',y') = \int \norm{y - y'}^4 \diff\beta(y)\diff\beta(y').
\end{gather*}
Moreover, for any $x,x' \in \Hilb_X$ and $y,y' \in \Hilb_Y$, and using the equality $\scal{x}{x'}\scal{y}{y'} = \scal{x y^{\top}}{x' y'^{\top}}_{\HS}$:
\begin{multline}
\label{eq:egwloss_bigdecomp}
    \norm{x-x'}^2 \norm{y-y'}^2 = \norm{x}^2 \norm{y}^2 +  \norm{x'}^2 \norm{y'}^2 + \norm{x'}^2 \norm{y}^2 + \norm{x}^2 \norm{y'}^2 + 4  \scal{x y^{\top}}{x' y'^{\top}}_{\HS}\\ - 2 \left( \scal{x}{x'}(\norm{y}^2 + \norm{y'}^2) + \scal{y}{y'}(\norm{x}^2 + \norm{x'}^2) \right).
\end{multline}
The four first terms can be simplified as follow:
\begin{gather*}
\int \norm{x}^2 \norm{y'}^2  \diff\pi(x,y) \diff\pi(x',y') = \int \norm{x'}^2 \norm{y}^2  \diff\pi(x,y) \diff\pi(x',y') = \int \norm{x}^2 \diff\alpha(x) \int \norm{y}^2 \diff\beta(y), \\
\int \norm{x}^2 \norm{y}^2  \diff\pi(x,y) \diff\pi(x',y') = \int \norm{x'}^2 \norm{y'}^2  \diff\pi(x,y) \diff\pi(x',y') = \int \norm{x}^2 \norm{y}^2 \diff\pi(x,y).
\end{gather*}
Moreover, by bilinearity of $\scal{\cdot}{\cdot}_{\HS}$:
\begin{equation*}
    \int \scal{x y^{\top}}{x' y'^{\top}}_{\HS} \diff\pi(x,y) \diff\pi(x',y') = \scal{\int x  y^{\top} \diff\pi(x,y)}{\int x' y'^{\top} \diff\pi(x',y')}_{\HS} = \norm{\int x y^{\top} \diff\pi(x,y)}_{\HS}^2.
\end{equation*}
Note that by the Cauchy-Scharwz inequality, this (Bochner) integral is well-defined.
Finally, since $\alpha$ is centered:
\begin{equation*}
\int \scal{x}{x'}\norm{y}^2  \diff\pi(x,y)\diff\pi(x',y') = \int \scal{x}{\int x' \diff\alpha(x')}\norm{y}^2 \diff\pi(x,y) = 0.
\end{equation*}
Using the same computations on the other terms, we obtain:
\begin{equation*}
    \int \left( \scal{x}{x'}(\norm{y}^2 + \norm{y'}^2) + \scal{y}{y'}(\norm{x}^2 + \norm{x'}^2) \right)   \diff\pi(x,y)\diff\pi(x',y') = 0.
\end{equation*}
Putting everything together in \cref{eq:egwloss_bigdecomp}, we finally obtain the desired formula.
\end{proof}

% \entropicgwcntlinearized*
\renewcommand{\thetheorem}{\ref*{prop:entropic_gw_cnt_linearized}}
\begin{theorem}
Let $\alpha \in \prob{\X}$
and
$\beta \in \prob{\Y}$ 
be two probability distributions with compact supports on topological spaces $\X$ and $\Y$,
endowed with definite CNT costs $c_\X$ and $c_\Y$.
Let 
$\Phi(x)=(\varphi(x), \tfrac{1}{2}\|\varphi(x)\|^2)$
and
$\Psi(y)=(\psi(y), \tfrac{1}{2}\|\psi(y)\|^2)$
denote their respective GW-embeddings, as in \cref{def:gw_embeddings}.
Then, for any temperature $\varepsilon \geqslant 0$,
the EGW problem of Eq.~\eqref{eq:entropic_gromov_wasserstein} is equivalent to:
\begin{equation*}
     GW_{\varepsilon}(\alpha, \beta) =  C(\alpha, \beta) +  8 
     \min_{\Gamma \in \Hspace} 
     \min_{\pi \in \coupling{\alpha}{\beta}}
     \mathcal{F}(\Gamma, \pi)~,
\end{equation*}
where $C(\alpha, \beta)$ is an additive constant and:
\begin{equation*}
\mathcal{F}(\Gamma, \pi)
~:=~
\norm{\Gamma}_{\HS}^2
~+~
\varepsilon \KL{\pi}{\alpha \otimes \beta}
~-~2
\int \scal{\overline{\Gamma}}{\Phi(x) \Psi(y)^{\top}}_{\HS} \diff \pi(x,y)~.
\end{equation*}
\end{theorem}

\begin{proof}[Proof of \cref{prop:entropic_gw_cnt_linearized,corrolary:egw_icp_gamma,corrolary:egw_icp_pi}]
This is a consequence of \cref{prop:cnt_gw_developed} combined with the following equality:
\begin{equation*}
    - \norm{ \int x  y^{\top} \diff\pi(x,y)}_{\HS}^2 =  \min_{\Gamma \in \Hspace} \left( \norm{\Gamma}_{\HS}^2 - 2 \scal{\Gamma}{\int x  y^{\top} \diff\pi(x,y)}_{\HS} \right),
\end{equation*}
since $f: \Gamma \mapsto \norm{\Gamma}^2 - 2\scal{\Gamma}{\int x y^{\top} \diff\pi(x,y)}_{\HS}$ is a quadratic function with minimum attained at $\Gamma = \int x y^{\top} \diff\pi(x,y)$.
We only need to apply this formula to the GW-embeddings $\Phi(x)$ and $\Psi(y)$ to obtain the results (\cref{corrolary:egw_icp_gamma} being a consequence of the convexity of $f$).
Finally, \cref{corrolary:egw_icp_pi} is a direct reformulation of \cref{prop:entropic_gw_cnt_linearized} using the fact that:
\begin{equation*}
    \scal{\overline{\Gamma}}{\Phi(x) \Psi(y)^{\top}}_{\HS} = \scal{\Phi(x)}{\overline{\Gamma}(\Psi(y))}_{\Hilb_\X} = \scal{\overline{\Gamma}^{\top}(\Phi(x))}{\Psi(y)}_{\Hilb_\Y}.
\end{equation*}
\end{proof}

\begin{remark}
\label{remark:egw_otherform}
In the proofs below, we will also use the following, equivalent form of \cref{prop:entropic_gw_cnt_linearized}:
    \begin{equation}
    \label{eq:egw_otherform}
        \GW_{\varepsilon}(\alpha, \beta) = C(\push{\varphi}{\alpha}, \push{\psi}{\beta}) + \min_{\Gamma \in \Hspace} \left( 8 \cdot \norm{\Gamma}^2_{\HS} + \OT^{\Gamma}_{\varepsilon}(\push{\varphi}{\alpha}, \push{\psi}{\beta}) \right), 
    \end{equation}
    where $\OT^\Gamma_{\varepsilon}$ is the entropic optimal transport for the cost $c_\Gamma(x,y) = - 16 \scal{\Gamma}{x y^T}_{\HS} - 4 \norm{x}^2 \norm{y}^2$.
    This corresponds to the form in which \cite{rioux2024entropic,zhang2024gromov} stated the dual formula for squared Euclidean costs.
\end{remark}
%\egwicppi*

\subsection{Proofs of \cref{subsection:entropic_bias}}

\cntotproperties*
\begin{proof}
    The main result in \cite{feydy2019interpolating} is that if the cost is such that $e^{-c(x,y)/\varepsilon}$ is a positive definite kernel for all $\varepsilon$, then $\SOT_\varepsilon \geq 0$ for all $\varepsilon$. When the cost $c$ is CNT, this condition is satisfied. 
The converse implication follows by taking the limit when $\varepsilon \to +\infty$. Indeed, we show below that 
$\lim_{\varepsilon \to +\infty} \OT_\varepsilon(\alpha,\beta) = \int_{\mathcal X \times \mathcal Y} c(x,y) \alpha \otimes \beta$ which implies that 
\begin{equation}\SOT_\varepsilon(\alpha,\beta) = \int_{\mathcal X \times \mathcal Y} c(x,y) \diff  \alpha(x) \diff \beta(y) - \frac{1}{2}\left(\int_{\mathcal X \times \mathcal X} c(x,x') \diff  \alpha(x) \diff \alpha(x') + \int_{\mathcal Y \times \mathcal Y} c(y,y) \diff \beta(y) \otimes \diff \beta(y')\right)\,.
\end{equation}
The right-hand side is the maximum mean discrepancy between the two measures $\alpha,\beta$ with respect to the cost $-\frac{1}{2}c(x,y)$. Since this quantity is nonnegative for all differences of probability measures, it implies that the cost $-\frac{1}{2} c(x,y)$ is a conditionally positive kernel.

To show that the entropic  transport cost limit when $\varepsilon \to +\infty$, we remark that, by testing the corresponding variational problem with $\pi(x,y) = \alpha(x) \otimes \beta(y)$, the KL term vanishes and therefore
\begin{equation}    \OT_\varepsilon(\alpha,\beta )\leq \int_{\mathcal X \times \mathcal Y} c(x,y) \diff  \alpha(x) \diff \beta(y)\,.
\end{equation}
Since the cost is nonnegative (and thus the transport cost), we have
\begin{equation}
    \KL{\pi}{\alpha \otimes \beta} \leq \frac{1}{\varepsilon} \int_{\mathcal X \times \mathcal Y} c(x,y) \diff  \alpha(x) \diff \beta(y)\,.
\end{equation}
By Pinsker's inequality, this implies the convergence of $\pi_\varepsilon(x,y)$ to $\alpha \otimes \beta$ w.r.t. the TV norm. Consequently, when $\varepsilon \to +\infty$, we have $ \int_{\mathcal X \times \mathcal Y} c(x,y) \diff  \alpha(x) \diff \beta(y) + o(1)  \leq \OT_\varepsilon(\alpha,\beta)$. 
\end{proof}

\begin{restatable}{lemma}{continuityentropicgwcnt}
\label{prop:continuity:entropic_gw_cnt2}
When the base costs are CNT, the operator $\func{\GW_{\varepsilon}}{\prob{\X}\times\prob{\Y}}{\R}$ is continuous for the weak convergence of measures.
\end{restatable}

\begin{proof}
Since the embeddings provided by \cref{theorem:cnt_embedding} are continuous, it is sufficient to prove the result when $\X$ and $\Y$ are compact subsets of two Hilbert spaces and when $c_\X, c_\Y$ are the corresponding squared Hilbert norms. 
We will proceed it by showing both the upper and lower semi-continuity of $\GW_{\varepsilon}$ for the weak convergence of measures.

For every $\Gamma \in \Hspace$, the function $x, y \longmapsto c_\Gamma(x,y)$ is Lipschitz on $\X\times \Y$.
Therefore, the function $\alpha, \beta \mapsto \OT^{\Gamma}_{\varepsilon}(\alpha, \beta)$ is continuous for the weak convergence \citep[Proposition 2]{feydy2019interpolating}.
Since $GW_{\varepsilon}(\alpha, \beta)$ is the infimum of continuous functions, it is upper semi-continuous for the weak convergence.

Conversely, let $\alpha_n \rightharpoonup \alpha \in \prob{\X}$ and $\beta_n \rightharpoonup \beta \in \prob{\Y}$. 
Let $\pi_n \in \coupling{\alpha_n}{\beta_n}$ be a sequence of optimal plans for $\GW_{\varepsilon}(\alpha_n, \beta_n)$. 
By compacity of $\coupling{\alpha_n}{\beta_n}$ for the weak convergence, there is an extraction $(\alpha_m, \beta_m)$ of $(\alpha_n, \beta_n)$ such that $\pi_m \rightharpoonup \pi \in \coupling{\alpha}{\beta}$,
and $\pi_m \otimes \pi_m \rightharpoonup \pi \otimes \pi$.
 By continuity of $x,x',y,y' \mapsto (\norm{x-x'}^2 - \norm{y-y'}^2)^2$: 
 \begin{equation*}
     \lim_{m \rightarrow + \infty} \int \left(\norm{x-x'}^2 - \norm{y-y'}^2 \right)^2 \diff \pi_m(x,y) \diff \pi_m(x',y') = \int \left(\norm{x-x'}^2 - \norm{y-y'}^2 \right)^{2} \diff \pi(x,y) \diff \pi(x',y').
 \end{equation*}
 By lower semi-continuity of $\mu, \nu \rightarrow \KL{\mu}{\nu}$, we also have $\liminf{\KL{\pi_m}{\alpha_m \otimes \beta_m}} \geq \KL{\pi}{\alpha \otimes \beta}$, so: 
 \begin{equation*}
     \liminf{\GW_{\varepsilon}(\alpha_m, \beta_m)}\geq  \int \left(\norm{x-x'}^2 - \norm{y-y'}^2 \right)^2 \diff \pi(x,y) \diff \pi(x',y') + \varepsilon \KL{\pi}{\alpha \otimes \beta}.
 \end{equation*}
By definition of EGW, the right term is always upper-bounded by $\GW_{\varepsilon}(\alpha, \beta)$.
Hence, any sequence of $\GW_{\varepsilon}(\alpha_n, \beta_n)$ contain a subsequence whose limit inferior is larger than $GW_{\varepsilon}(\alpha, \beta)$, which proves $\liminf{\GW_{\varepsilon}(\alpha_n, \beta_n)} \geq \GW_{\varepsilon}(\alpha, \beta)$ and the lower-semi continuity of $\GW_{\varepsilon}$.
\end{proof}

\cntsgwproperties*
\begin{proof}
From \citep[Proposition 4.4]{van2003probability}, we know that the set of finitely supported probability measures over $\X$ (resp. $\Y$) is dense in $\prob{\X}$ (resp. $\prob{\Y}$).
Therefore, we can restrict to the case where $\alpha$ and $\beta$ have finite support, and conclude on the general case using the continuity of EGW stated in \cref{prop:continuity:entropic_gw_cnt2}.
Since finite subsets of $\Hilb_\X$ and $\Hilb_\Y$ belong to finite-dimensional subspaces of $\Hilb_\X$ and $\Hilb_\Y$, we can further restrict ourselves to measures lying on finite-dimensional euclidean spaces: from now, we set $\alpha \in \prob{\R^\DD}$ and $\beta \in \prob{\R^{\EE}}$ as two centered measures for the squared euclidean norm, with $c_\X, c_\Y$ the squared Euclidean norms of $\R^\DD$ and $\R^{\EE}$.

In finite dimension, Hilbert-Schmidt operators correspond to matrices of $\R^{\DD \times \EE}$: therefore, each $\Gamma \in \Hspace$ admits a singular value decomposition (SVD) of the form $\Gamma = U^T \Sigma V$, where $\Sigma \in \Diag{\RR}$ is diagonal with positive coefficients and $U \in \Orth{\DD,\RR}, V \in \Orth{\EE,\RR}$ are orthogonal matrices (with $\RR \geq \max(\DD, \EE)$). 
\cref{eq:egw_otherform} rewrites as:
    \begin{equation}
    \label{eq:svd_entropicgw}
        GW_{\varepsilon}(\alpha, \beta) = C(\alpha, \beta) +  \min_{\Sigma \in \Diag{\RR}} ~\left[~~~8 \cdot \norm{\Sigma}_F^2 + \min_{U \in \Orth{\DD,\RR}, V \in \Orth{\EE,\RR}}\text{OT}_{\varepsilon}^\Sigma(U_{\sharp} \alpha, V_{\sharp}\beta)~~~\right],
    \end{equation}
where $\text{OT}_{\varepsilon}^\Sigma(\alpha, \beta)$ is the EOT for the cost $c_\Sigma(x,y) = - 4 \norm{x}^2\norm{y}^2 - 16 x^T \Sigma y$.
Denoting by $(\Sigma^*,U^*,V^*)$ the optimal matrices in \cref{eq:svd_entropicgw} (and using the fact that $(\Sigma^*,U^*,U^*)$ and $(\Sigma^*,V^*,V^*)$ are in the feasible set of $\GW_{\varepsilon}(\alpha,\alpha)$ and $\GW_{\varepsilon}(\beta,\beta)$ for the new formulation of \cref{eq:svd_entropicgw}) we obtain:
\begin{equation}
\label{prop:inequalities_sgw_proof}
    \begin{split}
    \GW_{\varepsilon}(\alpha, \beta) &= C(\alpha, \beta) + 8 \cdot \norm{\Sigma^*}_F^2 + \OT_{\varepsilon}^{\Sigma^*}(U^*_{\sharp} \alpha, V^*_{\sharp}\beta), \\
    \GW_{\varepsilon}(\alpha, \alpha) &\leq C(\alpha, \alpha) + 8 \cdot \norm{\Sigma^*}_F^2 + \OT_{\varepsilon}^{\Sigma^*}(U^*_{\sharp} \alpha, U^*_{\sharp} \alpha), \\
    \GW_{\varepsilon}(\beta, \beta) &\leq C(\beta, \beta) + 8 \cdot \norm{\Sigma^*}_F^2 + \OT_{\varepsilon}^{\Sigma^*}(V^*_{\sharp} \beta, V^*_{\sharp} \beta).
    \end{split}
\end{equation}
Writing $\SOT_{\varepsilon}^{\Sigma^*}(U^*_\sharp\alpha, V^*_\sharp\beta) = \OT_{\varepsilon}^{\Sigma^*}(U^*_{\sharp} \alpha, V^*_{\sharp}\beta) - \frac{1}{2}( \OT_{\varepsilon}^{\Sigma^*}(U^*_{\sharp} \alpha, U^*_{\sharp} \alpha) + \OT_{\varepsilon}^{\Sigma^*}(V^*_{\sharp} \beta, V^*_{\sharp} \beta))$, this implies:
\begin{equation*}
    \SGW_{\varepsilon}(\alpha, \beta) \geq C(\alpha, \beta) - \frac{1}{2}(C(\alpha, \alpha) + C(\beta, \beta)) + S_{\varepsilon}^{\Sigma^*}(U^*_\sharp\alpha, V^*_\sharp\beta).
\end{equation*}
The positivity of the constant $C(\alpha, \beta) - \frac{1}{2} (C(\alpha, \alpha) + C(\beta, \beta))$ comes from the following identity:
\begin{equation*}
\begin{split}
    C(\alpha, \beta) - \frac{1}{2}(C(\alpha, \alpha) + C(\beta, \beta)) &= 
     2 \left( \int \norm{x}^2 \diff\alpha(x) \right)^2 + 2 \left( \int \norm{y}^2 \diff\beta(y) \right)^2 - 4 \int \norm{x}^2 \diff\alpha(x) \int \norm{y}^2 \diff\beta(y)  \\
     &= 2 \left( \norm{x}^2 \diff\alpha(x) - \int \norm{y}^2 \diff\beta(y) \right)^2 \geq 0.
\end{split}
\end{equation*}
Finally, that $\SOT_{\varepsilon}^{\Sigma^*}(U^*_\sharp\alpha, V^*_\sharp\beta)$ is positive by recalling that $- c_{\Sigma^*}$ is a scalar product of GW-embeddings:
\begin{equation*}
    - c_{\Sigma^*}(x,y) = 4 \norm{x}^2\norm{y}^2 + 16 x^T \Sigma^* y = \scal{(2 \norm{x}^2, 4 \sqrt{\Sigma^*} x)}{(2 \norm{y}^2, 4 \sqrt{\Sigma^*} y)}.
\end{equation*}
which implies the positivity of the kernel $k_{\varepsilon}(x,y) = \exp(-c_{\Sigma^*}(x,y)/\varepsilon)$. 
We are in the setting of \citep[Theorem 1]{feydy2019interpolating}, which proves the positivity of $\SOT_{\varepsilon}^{\Sigma^*}(U^*_\sharp\alpha, V^*_\sharp\beta)$. 
The positivity of $\SGW_{\varepsilon}(\alpha, \beta)$ follows.
\end{proof}

\begin{remark}
\label{remark:universality_kernel}
The hypotheses of \cite{feydy2019interpolating} are not completely satisfied, since the kernel $k_{\varepsilon}$ is not necessarily universal.
However, the universality of $k_{\varepsilon}$ is only used to prove the definiteness of the Sinkhorn divergence, and the positivity of $k_{\varepsilon}$ is sufficient to show the positivity of $\SOT_{\varepsilon}^{\Sigma^*}(U^*_\sharp\alpha, V^*_\sharp\beta)$. 
It also shows that the lack of definiteness of EGW (evidenced in \cref{prop:sgw_counterex}) is a consequence of the non-universality of $k_{\varepsilon}$, which happens when the embedding $x \mapsto (2 \norm{x}^2, 4 \sqrt{\Sigma^*} x)$ is non-injective. 
In particular, if $\Sigma^*$ is full-rank, then the kernel $e^{((4\|x\|^2\| y \|^2 +16 x^{\top}\Sigma^* y)/\varepsilon}$ is universal on compact domains (since polynomial functions are contained in the corresponding RKHS), and thus $\SOT_{\varepsilon}^{\Sigma^*}$ is definite: $\SOT_{\varepsilon}^{\Sigma^*}(U^*_\sharp\alpha, V^*_\sharp\beta) = 0$ if and only if $U^*_\sharp\alpha = V^*_\sharp\beta$, and the latter implies that $\alpha$ and $\beta$ are isometric since $U^*$ and $V^*$ are orthogonal matrices. 
\end{remark}

\bigskip
\begin{restatable}{proposition}{sgwcounterex}
\label{prop:sgw_counterex}
Let $\alpha$ and $\beta$ be uniform measures on the spheres $\Sphere{\DD}$ and $\Sphere{\EE}$ of dimensions $\DD \neq \EE$, while $c_\X$ and $c_\Y$ are the squared norms of $\R^\DD$ and $\R^{\EE}$, respectively. Then,
$\alpha$ and $\beta$ are not isometric but there exists an $\varepsilon_0 > 0$ such that $\SGW_{\varepsilon}(\alpha, \beta) = 0$ for every $\varepsilon > \varepsilon_0$.    
\end{restatable}

\begin{proof}
The function $\pi \mapsto \KL{\pi}{\alpha \otimes \beta}$ is strictly convex: therefore, when $\varepsilon$ is sufficiently large, the function $\Loss_{\varepsilon}:~\pi~\mapsto~\Loss(\pi)~+~\varepsilon \KL{\pi}{\alpha \otimes \beta}$ is strictly convex as well, and $\Loss_{\varepsilon}$ admits a unique minimum on $\coupling{\alpha}{\beta}$.
By choice of $c_\X$ and $c_\Y$, $\Loss_{\varepsilon}(\pi)$ is invariant by rotation of the marginals of $\pi$: its unique minimizer must be radially symmetric, and the only coupling satisfying this constraint is $\alpha \otimes \beta$.

Similarly, when $\varepsilon$ is sufficiently large, the optimal plans of $GW_{\varepsilon}(\alpha,\alpha)$ and  $GW_{\varepsilon}(\beta,\beta)$  are $\alpha \otimes \alpha$ and $\beta \otimes \beta$. 
Using these optimal plans, the value of  $GW_{\varepsilon}(\alpha,\beta) - \frac{1}{2} ( GW_{\varepsilon}(\alpha,\alpha) + GW_{\varepsilon}(\beta,\beta))$ can be explicitly determined, and a direct computation shows that $\SGW_{\varepsilon}(\alpha, \beta) = 0$.
\end{proof}

\bigskip

\sgwdefinitenessquad*
\begin{proof}
To prove the result, we proceed by contraposition.
Let $\alpha \in \prob{\R^\DD}$, and let assume the existence of $\beta \in \prob{\R^{\EE}}$ non isometric to  $\alpha$ such that $\SGW_{\varepsilon}(\alpha, \beta) = 0$.
We will show that it implies an inequality involving $\varepsilon$ and the smallest eigenvalue of $\Sigma_{\alpha} = \int xx^T \diff\alpha$.

For the nullity of SGW to hold, the inequalities of \cref{prop:inequalities_sgw_proof} must be equalities: if $\Gamma^*$ is the optimal dual matrix for $\GW_{\varepsilon}(\alpha, \beta)$, denoting by $\Gamma^* = U^{*T} \Sigma^* V^*$ its singular value decomposition, we must have:
\begin{equation*}
\GW_{\varepsilon}(\alpha, \alpha) = C(\alpha, \alpha) + 8 \cdot \norm{\Sigma^*}_F^2 + \OT_{\varepsilon}^{\Sigma^*}(U^*_{\sharp} \alpha, U^*_{\sharp} \alpha),
\end{equation*}
and $\Gamma_{\alpha}^* = U^{*T} \Sigma^* U^*$ must be an optimal dual matrix for $\GW_{\varepsilon}(\alpha, \alpha)$.

From \cref{prop:cnt_sgw_properties}, $\SGW_{\varepsilon}(\alpha, \beta) = 0$ also implies $\SOT_{\varepsilon}^{\Sigma^*}(U^*_\sharp \alpha, V^*_\sharp \beta) = 0$.
As explained in \cref{remark:universality_kernel}, if $\Sigma^*$ is a positive definite matrix, then the nullity of $\SGW_{\varepsilon}$ would mean the isometry of $\alpha$ and $\beta$.
Therefore, $\Sigma^*$ cannot be definite: the matrix $\Sigma^*$ must have a null eigenvalue, and the matrix $\Gamma_{\alpha}^*$ has at least one null eigenvalue as well.

Meanwhile, \citep[Proposition 1]{zhang2024gromov} provides an estimate of the difference between the unregularized and entropic GW values, giving a bound of the form:
\begin{equation*}
\GW_{\varepsilon}(\alpha, \alpha) = \abs{ \GW_{\varepsilon}(\alpha, \alpha) - \GW(\alpha, \alpha)} \leq \widetilde{C}(\DD, R) \cdot \varepsilon \max(1, \log(1/\varepsilon)).
\end{equation*}
In the original result of \cite{zhang2024gromov}, the constant depends on the moments of $\alpha$ rather than its radius; however, since $\alpha$ is bounded, we can bound its moments by functions of $R$ and simply make the constant depend on $\DD$ and $R$.

Let $\pi^*$ denote the optimal plan of $\GW_{\varepsilon}(\alpha, \alpha)$ associated to the auto-correlation matrix $\Gamma_{\alpha}^*$. Using \cref{prop:cnt_gw_developed}, we get:
\begin{equation*}
\GW_{\varepsilon}(\alpha, \alpha) = C(\alpha, \alpha) - 8 \cdot \norm{ \int xx'^T \diff\pi^*}_F^2 - 4 \cdot \int \norm{x}^2\norm{x'}^2 \diff\pi^* \leq \widetilde{C}(\DD, R) \cdot \varepsilon \max(1, \log(1/\varepsilon)).
\end{equation*}
A direct computation shows that $C(\alpha, \alpha) = 8 \norm{\int x x^T \diff\alpha}_F^2 + 4 \int \norm{x}^4 \diff\alpha$, and:
\begin{equation*}
    8 \cdot \left(\norm{\int x x^T \diff\alpha}_F^2 - \norm{\int x x'^T \diff\pi^*}_F^2 \right) + 4 \cdot \left(\int \norm{x}^4 \diff\alpha - \int \norm{x}^2\norm{x'}^2 \diff\pi^* \right) \leq \widetilde{C}(\DD, R) \cdot \varepsilon \max(1, \log(1/\varepsilon)).
\end{equation*}
By the Cauchy-Schwartz inequality, $\int \norm{x}^2\norm{x'}^2 \diff\pi^* \leq \int \norm{x}^4 \diff\alpha$, and we finally obtain:
\begin{equation}
\label{eq:gwaa_eps_variation}
    8 \cdot \left( \norm{\Sigma_{\alpha}}_F^2 - \norm{\Gamma_{\alpha}^*}_F^2 \right) \leq \widetilde{C}(\DD, R) \cdot \varepsilon \max(1, \log(1/\varepsilon)).
\end{equation}
Moreover, $\Sigma_{\alpha}$ and $\Gamma_{\alpha}^*$ are symmetric semi-definite matrices.
Therefore, given $(e_i)$ an orthonormal eigenvector basis of $\Gamma_{\alpha}^*$ (by increasing eigenvalue), we have:
    \begin{equation*}
        \norm{\Gamma_{\alpha}^*}_F^2 = \sum_{i=1}^\DD (e_i^T \Gamma_{\alpha}^* e_i)^2 \quad \text{ and } \quad \norm{\Sigma_{\alpha}}_F^2 \geq \sum_{i=1}^\DD (e_i^T \Sigma_{\alpha} e_i)^2.
    \end{equation*}
 From the rearranging inequality, we also have:
\begin{equation*}
\text{For any } i \geq 1, \quad e_i^T \Gamma_{\alpha}^* e_i = \int (e_i^T x) (e_i^T x') \diff\pi \leq \int (e_i^T x) (e_i^T x) \diff\alpha = e_i^T \Sigma_{\alpha} e_i.
\end{equation*}
Finally, since $\Gamma_{\alpha}^*$ admits at least one null eigenvalue, we have $e_1^T \Gamma_{\alpha}^* e_1 = 0$, and $e_1^T \Sigma_{\alpha} e_1 \geq \lambda_\alpha$ where $\lambda_\alpha$ is the smallest eigenvalue of $\Sigma_{\alpha}$.
Putting everything together, we obtain the following inequality:
\begin{equation*}
    \norm{\Gamma_{\alpha}^*}_F^2 = \sum_{i=2}^\DD (e_i^T \Gamma_{\alpha}^* e_i)^2 \leq \sum_{i=2}^\DD (e_i^T \Sigma_{\alpha} e_i)^2 \leq \norm{\Sigma_{\alpha}}_F^2 - (e_1^T \Sigma_{\alpha} e_1)^2 \leq \norm{\Sigma_{\alpha}}_F^2 - \lambda_\alpha^2.
\end{equation*}

Plugging this result in \cref{eq:gwaa_eps_variation}, we finally prove that $ \varepsilon \max(1, \log(1/\varepsilon)) \geq C(\DD,R) \lambda_\alpha^2$ (with $C(\DD, R) =  8 / \widetilde{C}(\DD, R)$).
Applying the same proof on $\beta$ yields $ \varepsilon \max(1, \log(1/\varepsilon)) \geq C(\EE, R')  \lambda_\beta^2$, proving the contraposition of \cref{prop:definiteness_cnt_quad}.
\end{proof}

\subsection{Proofs of \cref{subsection:optim_bias}}

\equivalencepgd* 
\begin{proof}
To simplify the notations, we identify the base points $x \in \X$ and $y \in \Y$ with their GW-embeddings $\phi(x)$ and $\psi(y)$.

The gradient of $\Loss$ at any $\tilde{\pi} \in \meas{\X \times \Y}$ is $\grad_{\tilde{\pi}}{\Loss}: x,y \mapsto 2 \cdot \int \left( c_\X(x,x')-c_\Y(y,y') \right)^2 \diff\tilde{\pi}(x',y')$.
Therefore:
\begin{equation*}
    \text{For all } \pi, \tilde{\pi} \in  \meas{X \times Y}, \quad \int \grad_{\tilde{\pi}}\Loss \cdot \diff\pi = 2 \cdot \int \left( c_X(x,x')-c_Y(y,y') \right)^2 \diff\pi(x,y) \diff\tilde{\pi}(x',y'). 
\end{equation*}
When both $\pi$ and $\tilde{\pi}$ belong to the feasible set $\coupling{\alpha}{\beta}$, the computations of \cref{prop:cnt_gw_developed} can be adapted to show that:
\begin{equation*}
     \int \grad_{\tilde{\pi}}\Loss \cdot \diff\pi = 2 \cdot C(\alpha,\beta) - 16 \cdot \scal{\int x  y^{\top} \diff\pi}{\int x' y'^{\top}  \diff\tilde{\pi}}_{\HS} - 4 \int  \norm{x}^2\norm{y}^2 \cdot \diff\pi(x,y)- 4 \int  \norm{x}^2\norm{y}^2 \cdot \diff\tilde{\pi}(x,y),
\end{equation*}
where the last term is independent of $\pi$.
Therefore, for any $t > 0$, since $\Gamma_{t} = \int x y^{\top} \diff \pi_t$ and $\pi_{t+1} = \argmin_{\pi \in \coupling{\alpha}{\beta}} \OT^{\Gamma_t}_{\varepsilon}(\alpha, \beta)$ (with $\OT^{\Gamma_t}_{\varepsilon}$ defined in \cref{remark:egw_otherform}), we have:
\begin{equation*}
\begin{split}
        \pi_{t+1} &= \argmin_{\pi \in \coupling{\alpha}{\beta}} \left( - 16 \cdot \scal{\int x  y^{\top} \diff\pi}{\Gamma_t}_{\HS} - 4 \int  \norm{x}^2\norm{y}^2 \cdot \diff\pi(x,y) + \varepsilon\KL{\pi}{\alpha \otimes \beta} \right) \\
        &=  \argmin_{\pi \in \coupling{\alpha}{\beta}} \left( 
        \int \grad_{\pi_t}\Loss \cdot \diff\pi  + \varepsilon\KL{\pi}{\alpha \otimes \beta} \right),
\end{split}
\end{equation*}
which proves the desired equivalence.
\end{proof}

\equivalencedualgrad*

\begin{proof}
We first explicit the dual function:
\begin{equation*}
    \text{For all } \Gamma \in \Hspace, \quad \Dualfun_{\varepsilon}(\Gamma) = \norm{\Gamma}^2_{\HS} + (1/8) \cdot \OT_{\varepsilon}^{\Gamma}(\alpha, \beta).
\end{equation*}
We use \citep[Proposition 2]{rioux2024entropic} which proves the smoothness of $\Dualfun_{\varepsilon}$ and a formula for $\grad\Dualfun_{\varepsilon}(\Gamma)$ in the finite dimensional case. 
Their proof remains valid in the Hilbert setting\footnote{Indeed, the first term is quadratic, therefore smooth. The proof in \citep[Proposition 2]{rioux2024entropic} uses the implicit function theorem on $\Gamma$, which applies in this infinite dimensional setting of Hilbert space of Hilbert-Schmidt operators.}; $\grad\Dualfun_{\varepsilon}(\Gamma)$ exists, is continuous, and satisfies:
\begin{equation*}
\text{For all } \Gamma \in \Hspace, \quad 
    \grad\Dualfun_{\varepsilon}(\Gamma) = 2 \cdot \Gamma - 2 \cdot \int x y^{\top} \diff\pi_{\Gamma} \quad \text{ where } \quad \pi_{\Gamma} = \argmin_{\pi}\OT^{\Gamma}_{\varepsilon} (\alpha, \beta).
\end{equation*}
As a consequence, since $\pi_{t+1} = \argmin_{\pi} \OT^{\Gamma_t}_{\varepsilon}(\alpha, \beta)$ and $\Gamma_{t+1} = \int x y^{\top} \diff\pi_{t+1}$ for all $t > 0$, the previous equation becomes:
\begin{equation*}
     \grad\Dualfun_{\varepsilon}(\Gamma_t) = 2 \cdot \Gamma_t - 2 \cdot \Gamma_{t+1},
\end{equation*}
which is equivalent to the desired equation.
\end{proof}

\bigskip

\altminconvergencereg*

\begin{proof}
We first write: 
\begin{equation*}
\Dualfun_{\varepsilon}(\Gamma) = \min_{\pi \in \coupling{\alpha}{\beta}} f_{\varepsilon, \pi}(\Gamma)~, \quad \text{ where } \quad
    f_{\varepsilon, \pi} : \Gamma \longmapsto \norm{\Gamma}_{\HS}^2 - 2 \scal{\Gamma}{\int x y^{\top} \diff\pi}_{\HS} + C(\pi),
\end{equation*}
with $C(\pi)$ a constant independent of $\Gamma$.
Moreover, $\Dualfun_{\varepsilon}(\Gamma_t) = f_{\varepsilon, \pi_{t+1}}(\Gamma_t)$ and $\Gamma_{t+1} = \int x y^{\top} \diff\pi_{t+1}$.
Therefore:
\begin{equation*}
\begin{split}
\Dualfun_{\varepsilon}(\Gamma_t) - \Dualfun_{\varepsilon}(\Gamma_{t+1}) &\geq  f_{\varepsilon,\pi_{t+1}}(\Gamma_t) -  f_{\varepsilon,\pi_{t+1}}(\Gamma_{t+1}) \\
& \geq (\norm{\Gamma_{t}}_{\HS}^2 - \norm{\Gamma_{t+1}}_{\HS}^2) - 2 (\scal{\Gamma_t}{\Gamma_{t+1}}_{\HS} - \scal{\Gamma_{t+1}}{\Gamma_{t+1}}_{\HS}) \\
& \geq \norm{\Gamma_{t+1}}_{\HS}^2 + \norm{\Gamma_{t}}_{\HS}^2 - 2 \scal{\Gamma_t}{\Gamma_{t+1}}_{\HS} \\
& \geq \norm{\Gamma_t - \Gamma_{t+1}}_{\HS}^2.
\end{split}
\end{equation*}

By summation, we obtain:
\begin{equation}
    \sum_{k=0}^{t-1} \norm{\Gamma_k - \Gamma_{k+1}}_{\HS}^2 \leq \Dualfun_{\varepsilon}(\Gamma_0) - \Dualfun_{\varepsilon}(\Gamma_{t}).
\end{equation}
Since the right term converges to a finite value, we must have $\norm{\Gamma_k - \Gamma_{k+1}}_{\HS} \rightarrow 0$.
Since $\Gamma_{t+1} - \Gamma_{t} = \frac{1}{2} \grad\Dualfun_{\varepsilon}(\Gamma_t)$, we also have $\grad\Dualfun_{\varepsilon}(\Gamma_t) \rightarrow 0$: by continuity of $\grad\Dualfun_{\varepsilon}(\Gamma_t)$, all sublimits of $(\Gamma_t)$ are critical points of $\Dualfun_{\varepsilon}$.
\end{proof}

\begin{remark}
The limit $\norm{\Gamma_k - \Gamma_{k+1}}_{\HS} \rightarrow 0$ implies that the set of subsequential limits of $(\Gamma_t)$ is connected: in particular, if the set of critical points of $\Dualfun_{\varepsilon}$ is discrete, then the sequence $(\Gamma_t)$ admits a unique limit.
Moreover, since $\Gamma_{t+1} - \Gamma_{t} = \frac{1}{2} \grad\Dualfun_{\varepsilon}(\Gamma_t)$, we also have a bound on the gradient norms:
\begin{equation*}
   \sum_{k=0}^{t-1} \norm{\grad\Dualfun_{\varepsilon}(\Gamma_k)}_{\HS}^2 \leq \tfrac{1}{2} \cdot (\Dualfun_{\varepsilon}(\Gamma_0) - \Dualfun_{\varepsilon}(\Gamma_t)).
\end{equation*}
\end{remark}

\newpage

\section{Additional Experiments}
\label{appendix:additional_expes}

\subsection{Detailed Analysis of Solver Convergence}
\label{subsection:appendix:solver_cvg}
We now explore convergence properties of EGW solvers.
In particular, we thoroughly investigate the differences between classical and kernel-based implementations by focusing on \textsc{Entropic-GW} \cite{peyre2016gromov} and \textsc{Kernel-GW} (\cref{alg:kernel_gw}), which are both exact EGW solvers.
We mainly show that convergence speed is significantly increased by using symmetrized Sinkhorn (\cref{alg:sinkhorn_symm}) rather than standard Sinkhorn in EGW solvers.
Symmetrized Sinkhorn with \textsc{Entropic-GW} also introduce instabilities that are greatly attenuated with \textsc{Kernel-GW}.
Finally, adaptive Sinkhorn (\cref{alg:kernel_gw_adaptive}) allows symmetrized \textsc{Kernel-GW} to converge without requiring to tune $\ninner$ manually, whereas other variants does not benefit from adaptive scheduling.
Therefore, combining both symmetrization and kernelization is the best option to fasten convergence while keeping stable optimization, explaining the speed-up of \textsc{CNT-GW} over \textsc{Quadratic-LowRank-GW}.

\subsubsection{Visualization of Convergence Properties}

\paragraph{Case of Non-Convergence for Non-CNT Costs.} In \cref{fig:gradientbias}, we illustrate a counter-example to the convergence of \textsc{Entropic-GW}
For the cost function $c_X, c_Y = \norm{\cdot}^6$ (which is not conditionally of negative type), we identify a configuration where the algorithm oscillates indefinitely between two distinct states, which we guarantee to never happen with CNT costs.

\begin{figure*}[h]
\centering
\includegraphics[width=\textwidth]{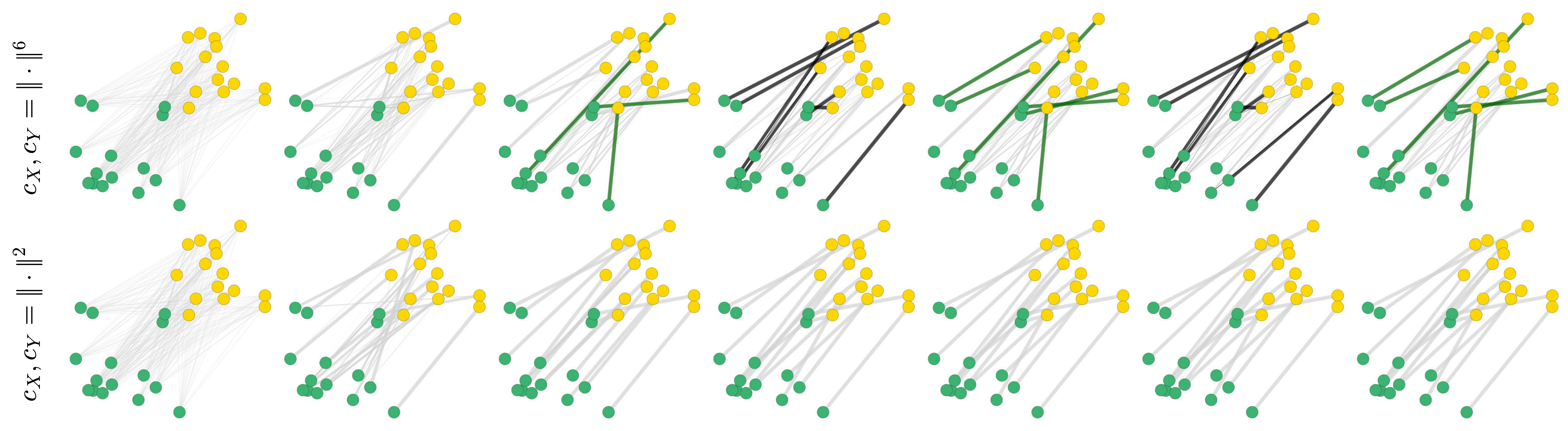} 
    \caption{Successive transport plans obtained by \textsc{Entropic-GW}. For non CNT costs (e.g., $c_X,c_Y = \norm{\cdot}^6$), the algorithm may oscillate indefinitely between distinct transport plans.
    We prove that this cannot occur when the costs are CNT (e.g., $c_X,c_Y = \norm{\cdot}^2$).}
    \label{fig:gradientbias}
\end{figure*}

\paragraph{Convergence Trajectories.} We visualize the impact of Sinkhorn symmetrization and the transition from \textsc{EntropicGW} to \textsc{Kernel-GW} by plotting optimization trajectories in the dual space $\Hspace$.
We take squared norm costs and measures in $\R$ and $\R^2$ so that $\Hspace = \R^{1 \times 2}$ can be displayed in 2D, and we represent the optimization trajectories for different numbers of Sinkhorn iterations $\ninner$.
As shown in \cref{fig:convergencebias}, standard Sinkhorn does not follow the true dual gradient: solver steps are systematically oriented towards the origin, indicating a bias towards the trivial plan $\alpha \otimes \beta$.
Symmetrizing Sinkhorn solves this problem, and the solver follows a straight trajectory towards the optimum.
When further reducing $\ninner$, an additional bias appears: the solver does not converge to the EGW minimum but stabilizes around a value biased towards the origin.
\textsc{Kernel-GW} corrects this effect: when $\ninner$ is small, Sinkhorn steps remain noisy but gravitate around the true optimum.

\begin{figure*}[h]
\begin{subfigure}[b]{0.165\linewidth}
        \centering
        \includegraphics[width=\linewidth]{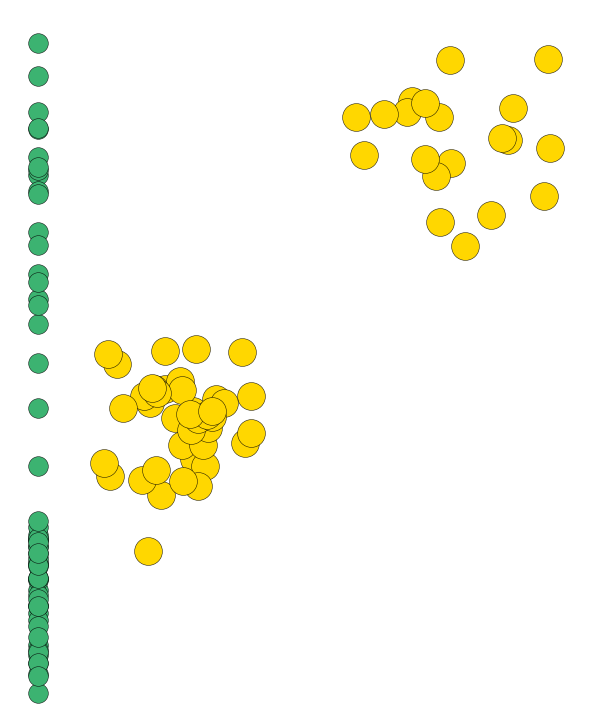}
        \caption{Inputs $\alpha$ and $\beta$}
        \label{subfig:convergencebias_inputs}
    \end{subfigure}
    \centering
    \begin{subfigure}[b]{0.412\linewidth}
        \centering
        \includegraphics[width=0.494\linewidth]{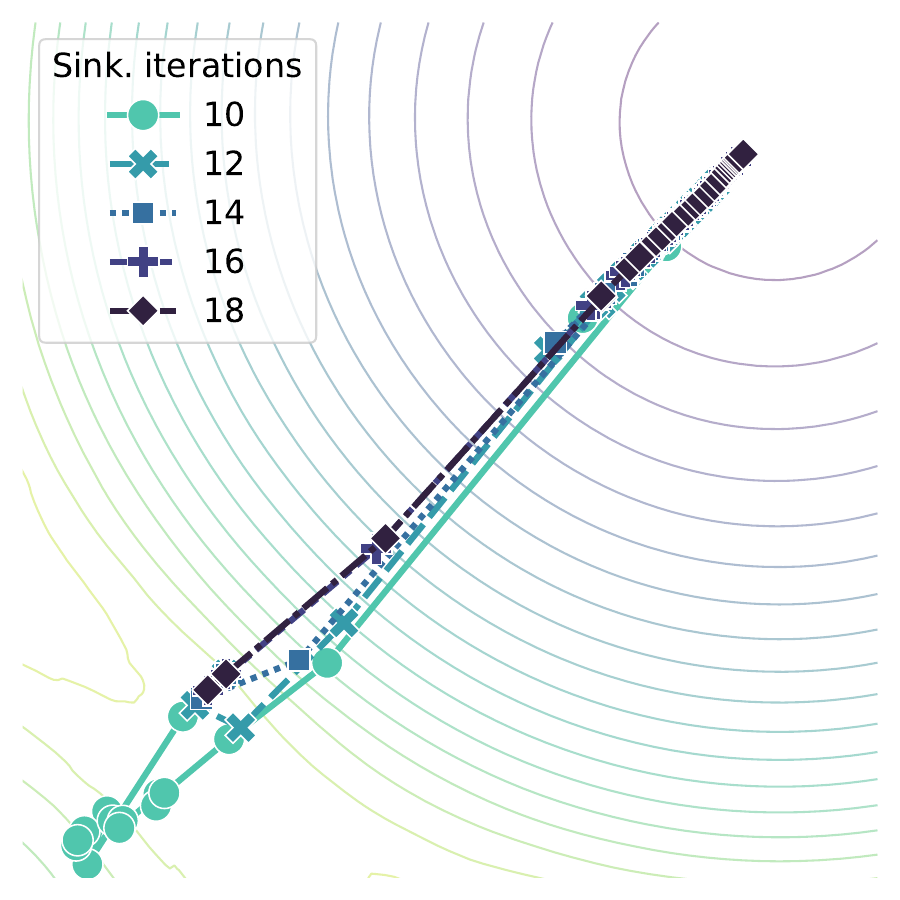}
        \includegraphics[width=0.494\linewidth]{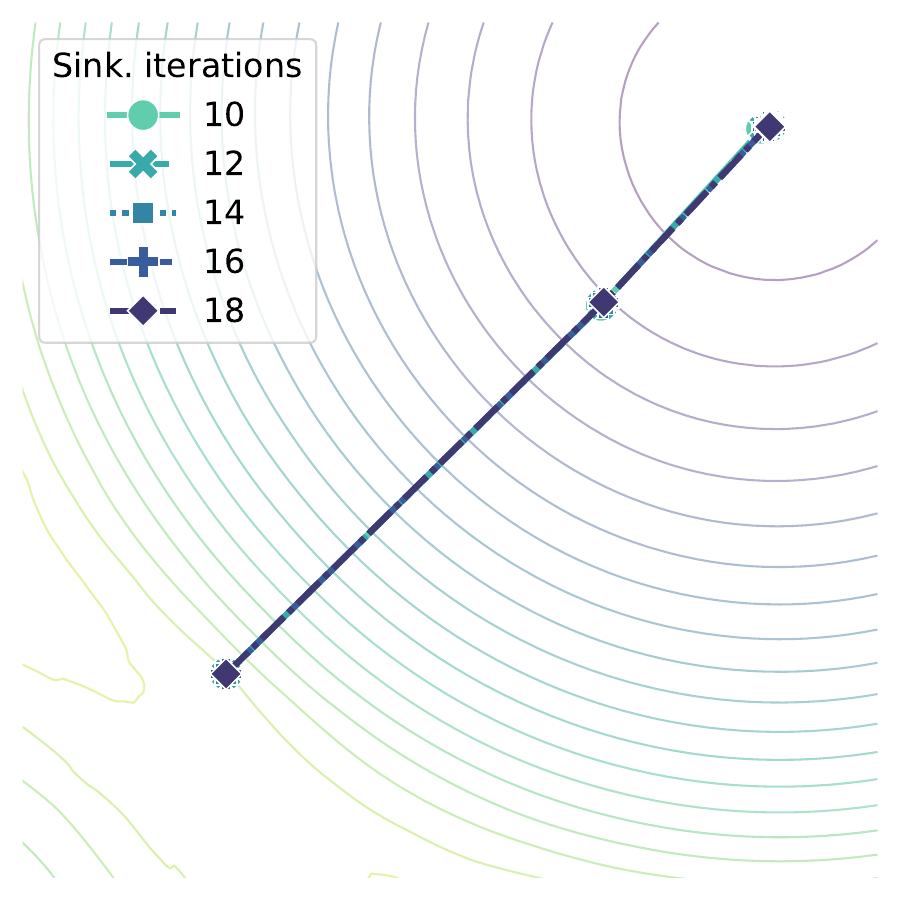}
        \caption{Non-symmetrized \textit{(left)} vs Symmetrized \textit{(right)}}
        \label{subfig:convergencebias_symm}
    \end{subfigure}
    \hfill
    \begin{subfigure}[b]{0.412\linewidth}
        \centering
        \includegraphics[width=0.494\linewidth]{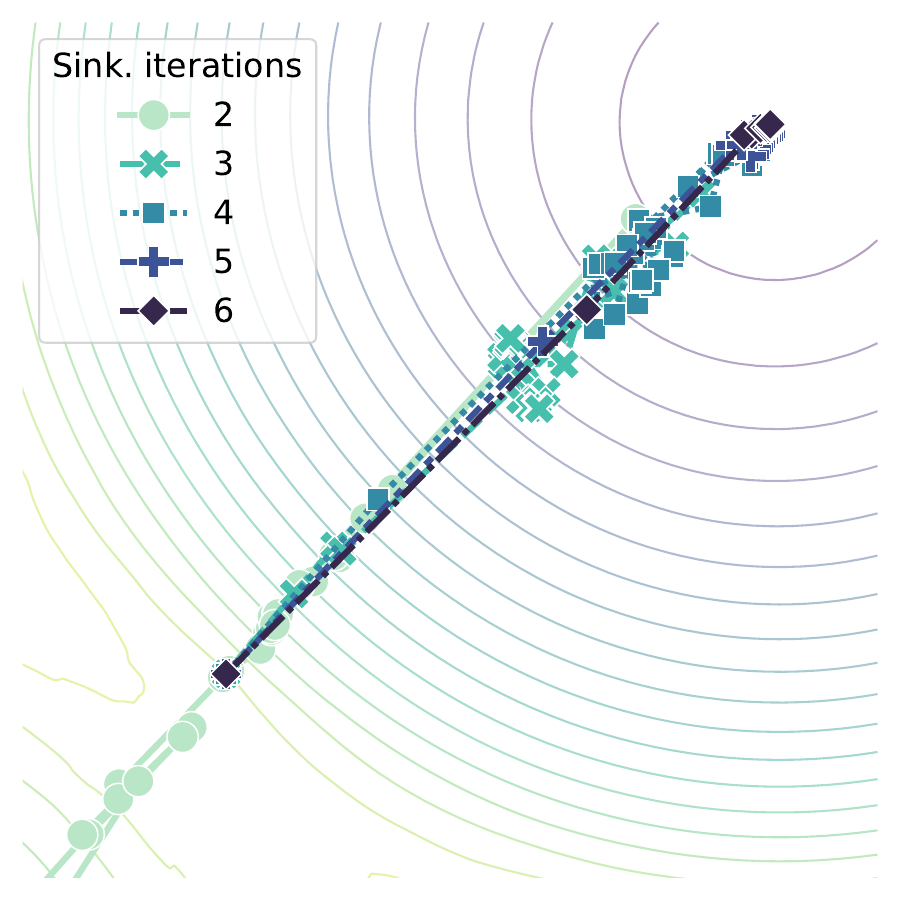}
        \includegraphics[width=0.494\linewidth]{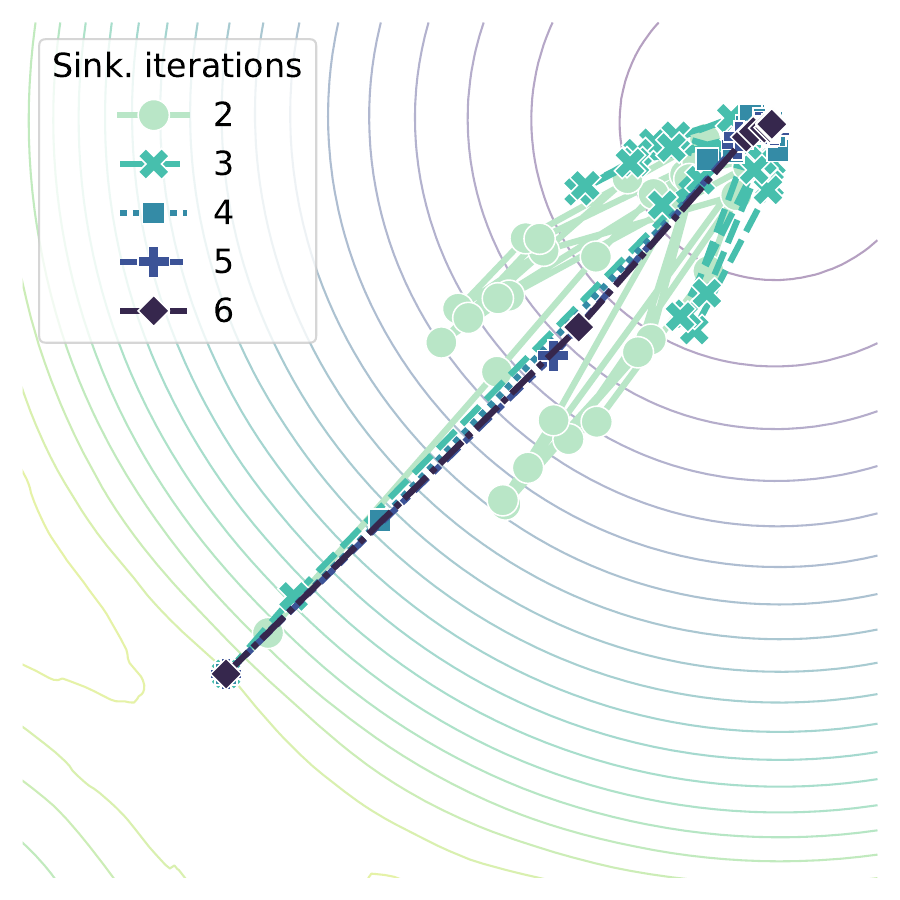}
        \caption{\textsc{Entropic-GW} \textit{(left)} vs \textsc{Kernel-GW} \textit{(right)}}
        \label{subfig:convergencebias_primaldual}
    \end{subfigure}

    \caption{Optimization steps obtained with EGW solvers in different configurations. The dual landscape $\Gamma \mapsto \Dualfun_{\varepsilon}(\Gamma)$ is plot in the background.
    Solvers are initialized at $\pi = \alpha \otimes \beta$, i.e. $\Gamma = 0$. In the right plot, both solvers use symmetrized Sinkhorn.}
    \label{fig:convergencebias}
\end{figure*}

\subsubsection{Convergence Speed Evaluation} 

To evaluate the practical impact of symmetrization and kernelization on convergence, we study the convergence speed of \textsc{Entropic-GW} and \textsc{Kernel-GW} on different shapes, with standard and symmetrized Sinkhorn.
The curves of this section were obtained by sampling $\NN = 2,000$ points uniformly from the horse shape of the SMAL dataset \cite{Zuffi_CVPR_2017}; we record elapsed time and EGW loss at each solving step.
The curves are averaged over $20$ random samplings (except for \cref{fig:ablation} and \cref{fig:smaller_sinkiters} were $50$ samples were computed).
We plot the results for different number of Sinkhorn iterations $\ninner$.
We also evaluate solvers with adaptive scheduling: starting from a small number of iterations, $n_{\text{inner}}^0 = 5$, we double its value each time the estimated EGW loss increases from one step to the other.
We compare \textsc{Entropic-GW} and \textsc{Kernel-GW}, as well as \textsc{Quadratic-LowRank-GW} and \textsc{CNT-GW}. 
Results are given in \cref{fig:ablation}.

Here are our key findings. Methods using standard Sinkhorn updates converge slowly; when the iteration budget $\ninner$ is low, they fail to reach the true global minimum.
Applying symmetrization to \textsc{Entropic-GW} accelerates convergence but introduces an instability: the solver tends to diverge for small $\ninner$ and does not benefit effectively from adaptive scheduling.
With standard Sinkhorn, \textsc{Kernel-GW} is slightly more stable than \textsc{Entropic-GW}, as the spikes observed at the beginning of the curves are attenuated; however, it does not make any meaningful difference on convergence time.
Finally, the combination of \textsc{Kernel-GW} with symmetrized Sinkhorn provides the most robust performance. 
This configuration accelerates and stabilizes convergence across all values of $\ninner$, allowing the adaptive scheduling to function optimally.
We obtain similar conclusions with the approximation methods \textsc{Quadratic-LowRank-GW} and \textsc{CNT-GW}. 

\begin{figure*}[h!]
\centering
\begin{subfigure}{\linewidth}
\centering
\includegraphics[width=\textwidth]{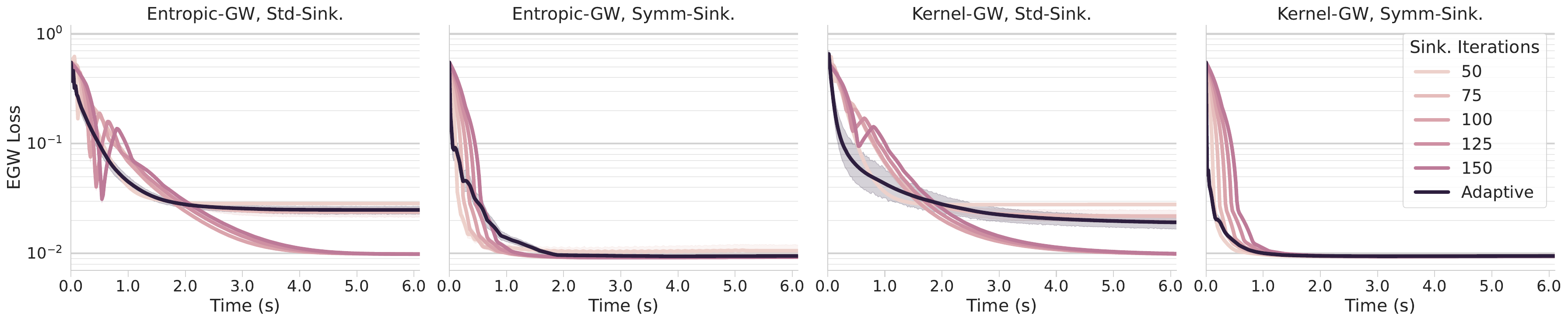} 
\caption{Cubic-time methods: \textsc{Entropic-GW} and \textsc{Kernel-GW}.}
\label{subfig:ablation_distinct_horses}
\end{subfigure}
\begin{subfigure}{\linewidth}
\centering
\includegraphics[width=\textwidth]{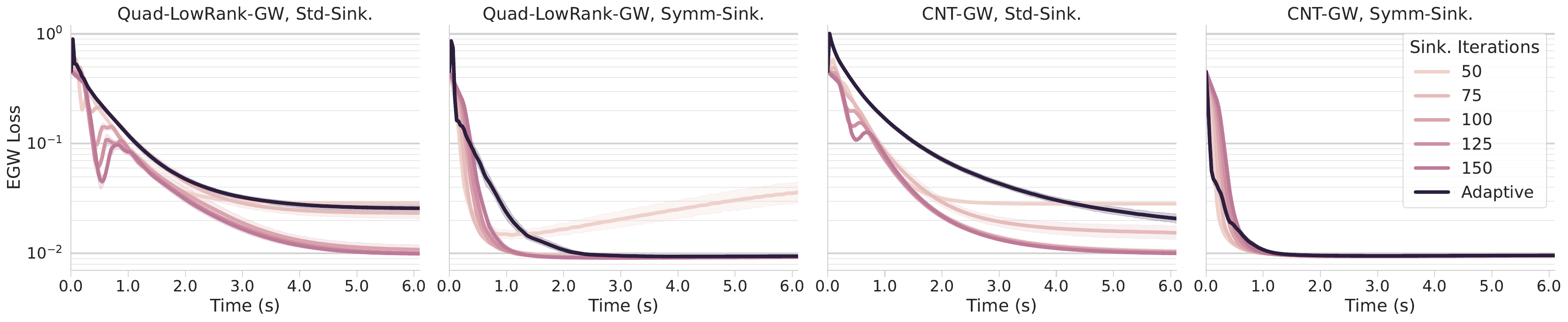} 
\caption{Quadratic-time approximations: \textsc{Quadratic-LowRank-GW} and \textsc{CNT-GW}.}
\end{subfigure}
\caption{Evolution of the EGW loss through time in different configurations, for various numbers of Sinkhorn iterations.}
\label{fig:ablation}
\end{figure*}

\paragraph{Smaller Sinkhorn Iterations Number.} These trends are amplified when we further reduce $\ninner$ (\cref{fig:smaller_sinkiters}).
Notably, the symmetrized versions of \textsc{Entropic-GW} and \textsc{Quadratic-LowRank-GW} exhibit important divergence.
In contrast, \textsc{Kernel-GW} and \textsc{CNT-GW} remain stable for $\ninner \geq 30$.
When $\ninner \leq 20$, instabilities start to appear but final losses remain closer to the true minimum compared to baselines.

\begin{figure*}[ht]
\centering
\begin{subfigure}{\linewidth}
\includegraphics[width=\textwidth]{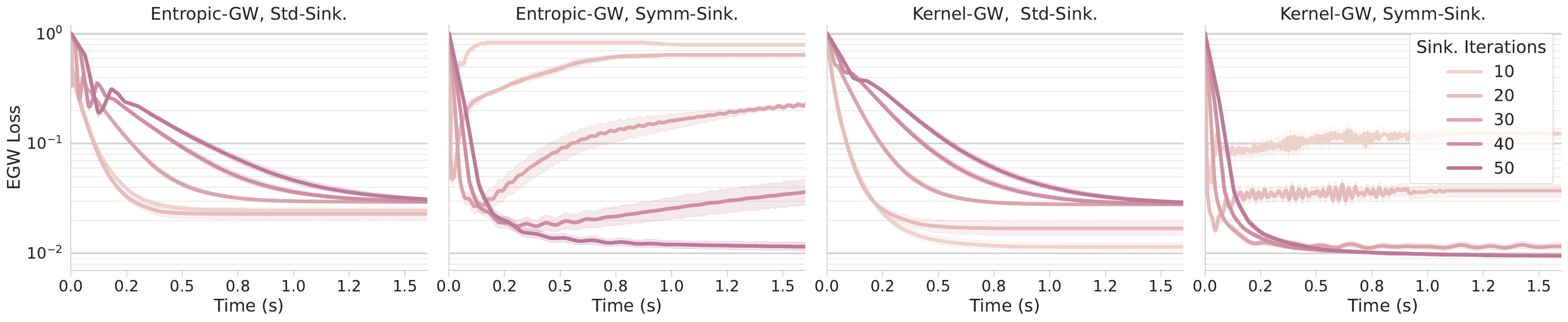} 
\caption{Cubic-time methods: \textsc{Entropic-GW} and \textsc{Kernel-GW}.}
\end{subfigure}
\begin{subfigure}{\linewidth}
\includegraphics[width=\textwidth]{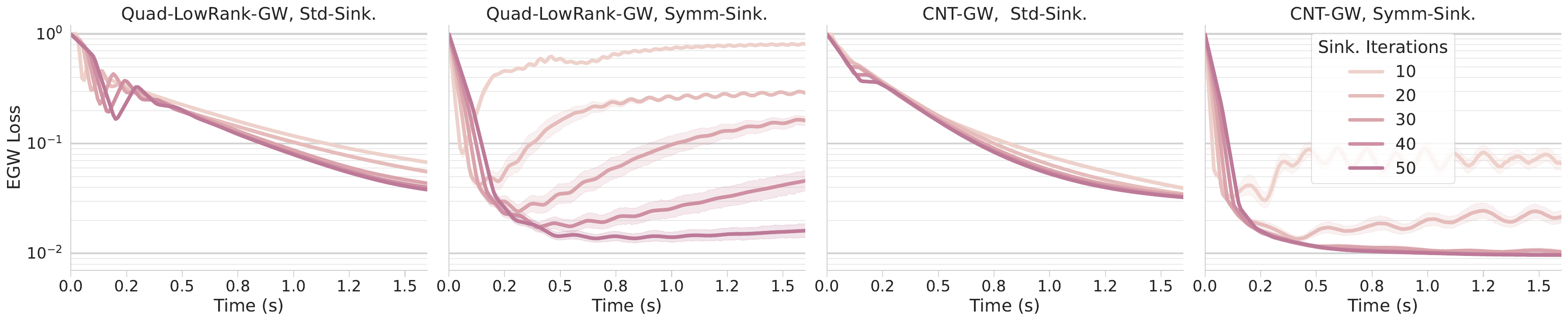} 
\caption{Quadratic-time approximations: \textsc{Quadratic-LowRank-GW} and \textsc{CNT-GW}.}
\end{subfigure}
\caption{Convergence curves for smaller values of $\ninner$.}
\label{fig:smaller_sinkiters}
\end{figure*}

\paragraph{Stopping Criterion for Sinkhorn Iterations.} An alternative to a fixed iteration budget is to terminate the Sinkhorn loop dynamically based on marginal error. Specifically, given a threshold $\tau$, we iterate until the weighted $L_1$ error satisfies:
\begin{equation*}
    \sum_i  \Big\vert \sum_j \pi_{ij} - \alpha_i \Big\vert \cdot \alpha_i < \tau \quad \text{ and } \quad \sum_j  \Big\Vert \sum_i \pi_{ij} - \beta_j \Big\Vert \cdot \beta_j  < \tau.
\end{equation*}
Results in \cref{fig:cvg_stopThr} show that while symmetrizing Sinkhorn still improves convergence, the kernel implementation does not bring noticeable improvements. 
Sinkhorn thresholding also appears to be slower than choosing a fixed number of iterations, and its sensitivity to the choice of $\tau$ makes this criterion more complex to tune.
Therefore, we recommend using a fixed number of iterations, although hybrid strategies (that use both a fixed maximum of iterations and early stopping based on margin error) offer a good alternative.

\begin{figure*}[h!]
\centering
\includegraphics[width=\textwidth]{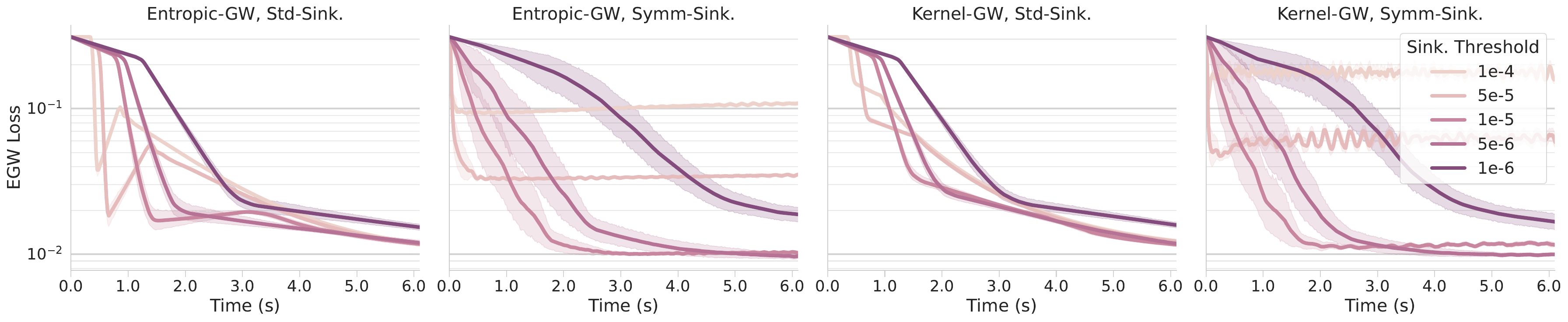} 
\caption{Convergence curves using Sinkhorn convergence thresholding.}
\label{fig:cvg_stopThr}
\end{figure*}

\subsubsection{Comparison with other solvers.}
\label{subsubsection:appendix_other_solvers_comparison}

We provide convergence curves of other EGW baselines, using the same experimental setup as in the previous section.

\paragraph{Sampled-GW.} \cref{subfig:sampledgw} illustrates the behavior of \textsc{Sampled-GW} \cite{kerdoncuff2021sampled} on the horse shapes.
It highlights a fundamental limitation: random sampling induces important fluctuations in the EGW loss, preventing the algorithm from converging to a stationary point.
This variability restricts its utility in applications requiring precise matching.

\paragraph{Dual-GW.} We extend the dual solver of \cite{rioux2024entropic} (originally for squared norms) to CNT costs using our Kernel PCA embeddings.
As shown in \cref{subfig:dualgw}, convergence remains slow, empirically confirming the theoretical convergence rate improvement proven in \cref{prop:altmin_convergence_reg}.

\paragraph{Proximal-GW.} Finally, \cref{subfig:proximalgw} presents results for \textsc{Proximal-GW} \cite{xu2019scalable} -- which outputs exact GW solutions rather than EGW couplings.
We also introduce a kernelized variant, \textsc{ProximalKernel-GW}, based on the same cost matrix as \textsc{Kernel-GW}.
Sinkhorn symmetrization still accelerates convergence in this proximal setting.
However, kernel implementation provides no additional benefit here, and adaptive scheduling fails to converge properly.

\begin{figure*}[ht]
\centering
\begin{subfigure}{0.495\linewidth}
\includegraphics[width=\textwidth]{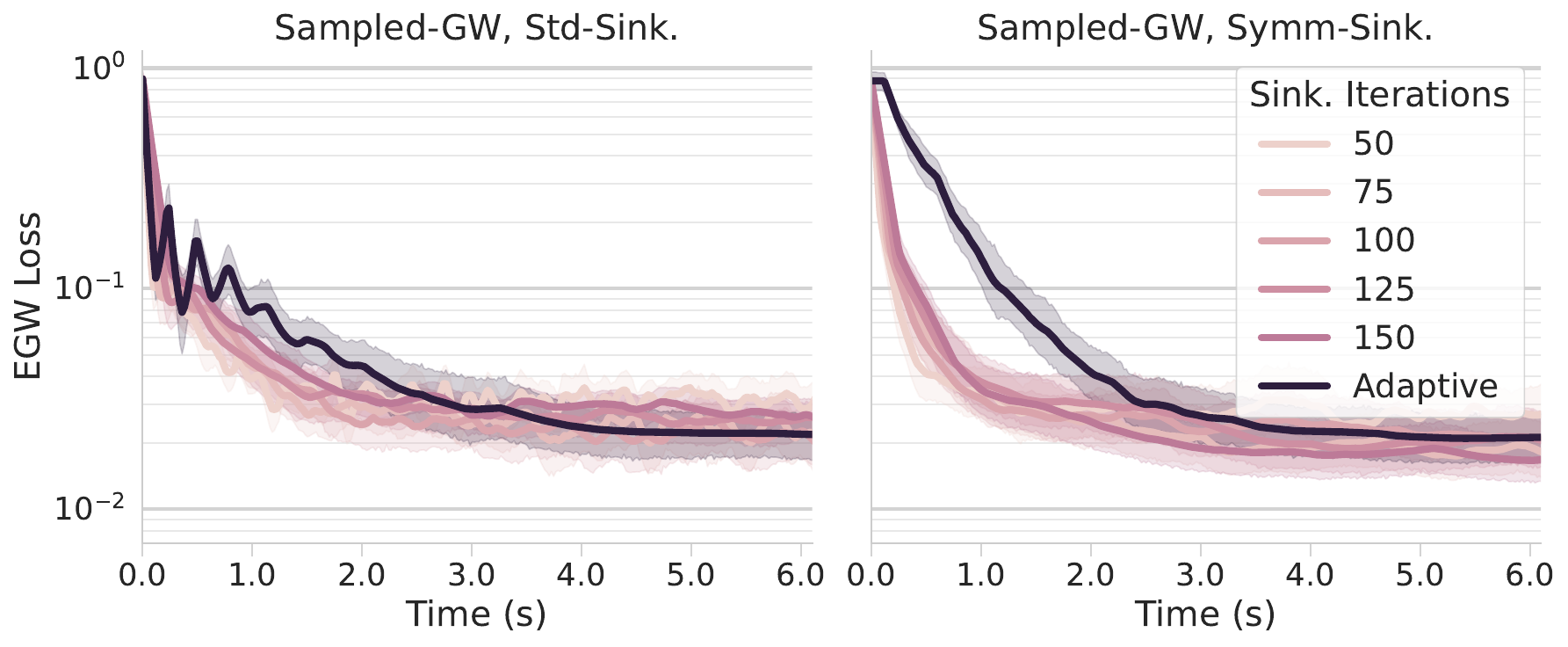} 
\caption{\textsc{Sampled-GW} (using $20$ samples).}
\label{subfig:sampledgw}
\end{subfigure}
\begin{subfigure}{0.495\linewidth}
\includegraphics[width=\textwidth]{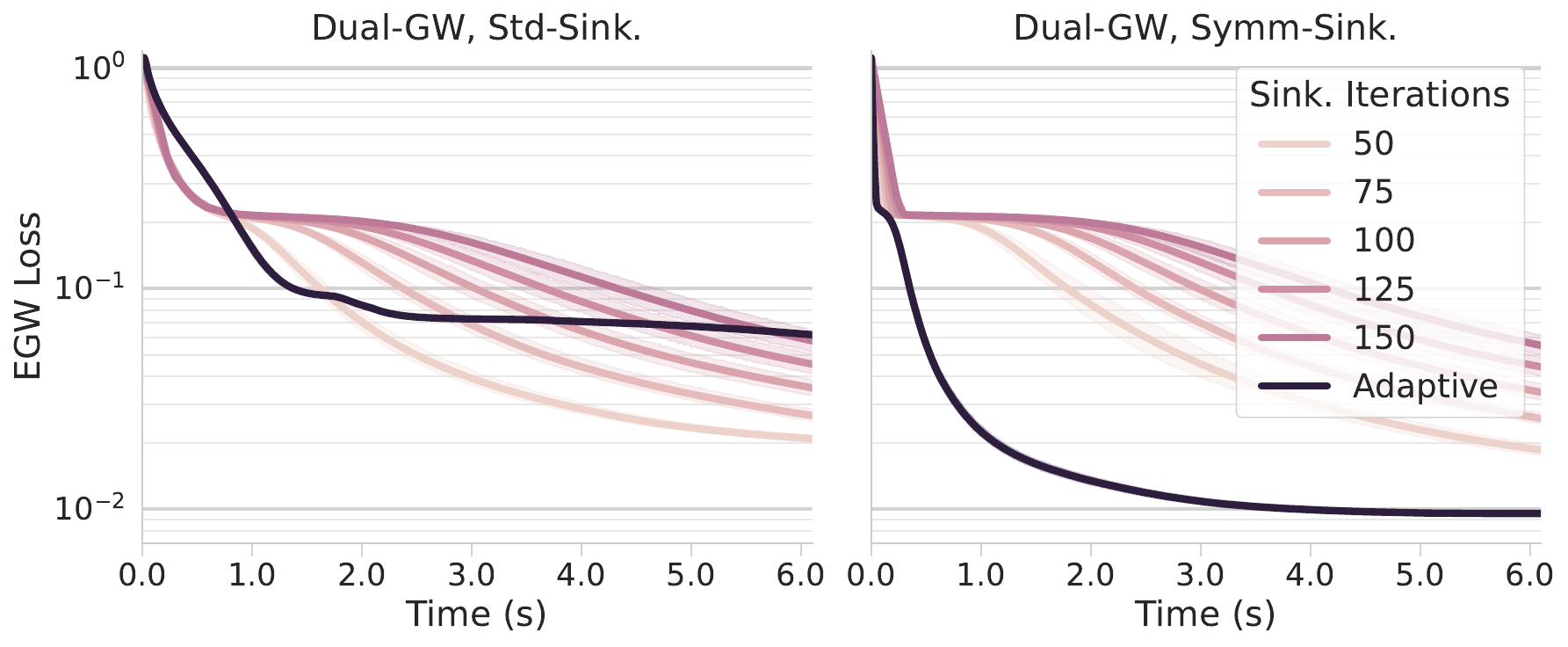} 
\caption{\textsc{Dual-GW} with Kernel PCA (with $D = 20$).}
\label{subfig:dualgw}
\end{subfigure}
\vfill
\begin{subfigure}{\linewidth}
\includegraphics[width=\textwidth]{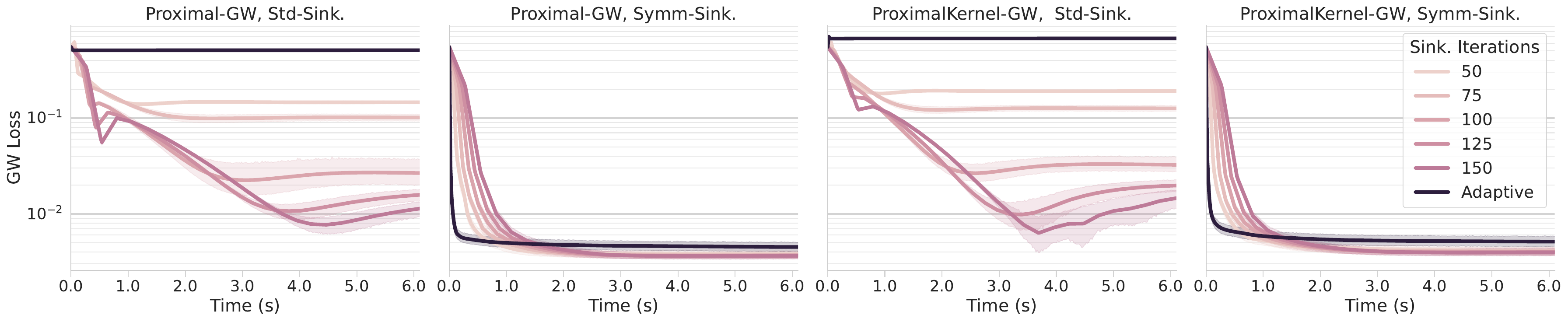} 
\caption{\textsc{Proximal-GW} and \textsc{ProximalKernel-GW}.}
\label{subfig:proximalgw}
\end{subfigure}
\caption{Convergence curves of other solvers with standard or symmetrized Sinkhorn.}
\label{fig:other_solvers}
\end{figure*}

\subsubsection{Convergence analysis on synthetic datasets.} 

\paragraph{Synthetic Shapes.} We evaluate \textsc{Kernel-GW} and \textsc{Entropic-GW} on point clouds sampled from a unit sphere and a regular tetrahedron (\cref{fig:convergence_surfaces}).
The tetrahedra case exhibits the same trend as before.
However, for spheres, symmetrizing Sinkhorn does not improve convergence speed: the large number of symmetries in the input shapes allows standard Sinkhorn methods to converge quickly even without symmetrization.
Yet, run times for symmetrized \textsc{Kernel-GW} are roughly equivalent to non-symmetrized solvers while symmetrized \textsc{Entropic-GW} is highly unstable when the number of iterations is too small.

\begin{figure*}[h!]
\centering
\begin{subfigure}{\linewidth}
\includegraphics[width=\textwidth]{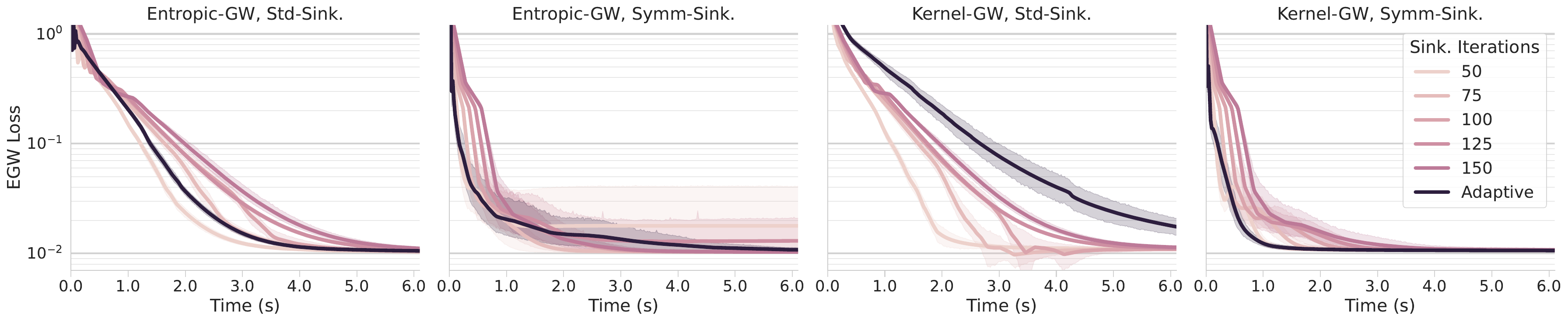} 
\caption{Tetrahedras.}
\label{subfig:cvg_tetrahedras}
\end{subfigure}
\begin{subfigure}{\linewidth}
\includegraphics[width=\textwidth]{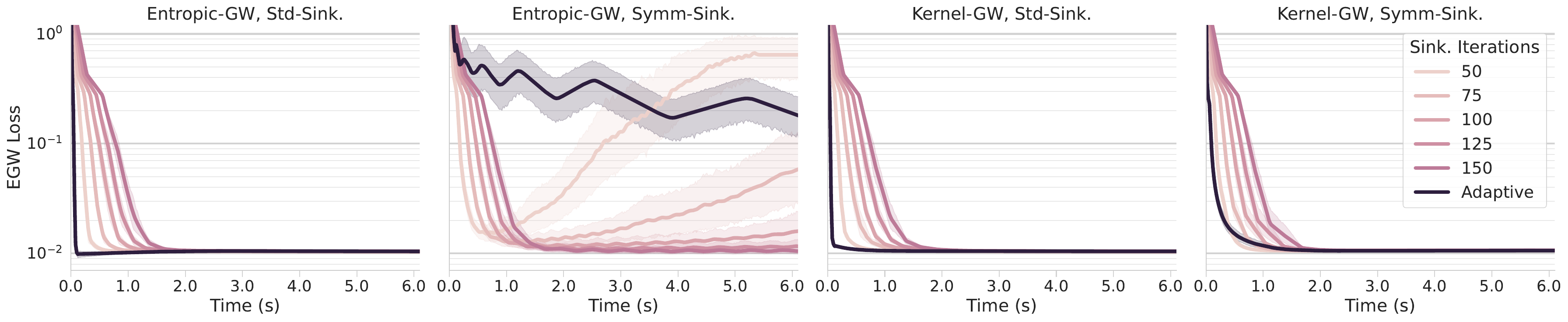} 
\caption{Spheres.}
\label{subfig:cvg_spheres}
\end{subfigure}
\caption{Convergence curves on synthetic surface distributions.}
\label{fig:convergence_surfaces}
\end{figure*}

\paragraph{Gaussian Inputs.} On unstructured Gaussian distributions (\cref{fig:convergence_gaussians}), \textsc{Kernel-GW} with symmetrized Sinkhorn clearly outperforms all other implementations, demonstrating superior stability in the absence of geometric structure.
This is especially true for 3D distributions, where symmetrized \textsc{Entropic-GW} fails to converge even with $\ninner=150$.

\begin{figure*}[h!]
\centering
\begin{subfigure}{\linewidth}
\includegraphics[width=\textwidth]{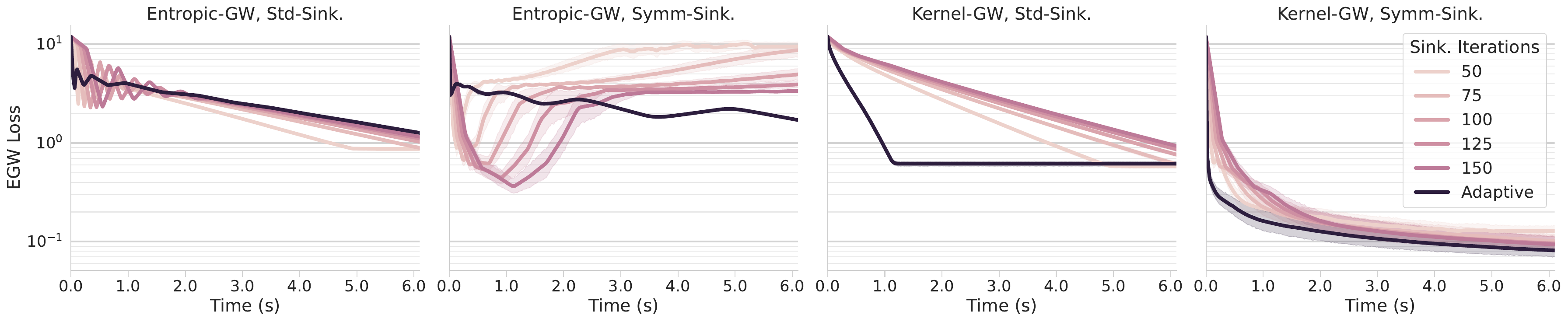} 
\caption{3D Gaussian distributions.}
\label{subfig:cvg_gaussian3d}
\end{subfigure}
\begin{subfigure}{\linewidth}
\includegraphics[width=\textwidth]{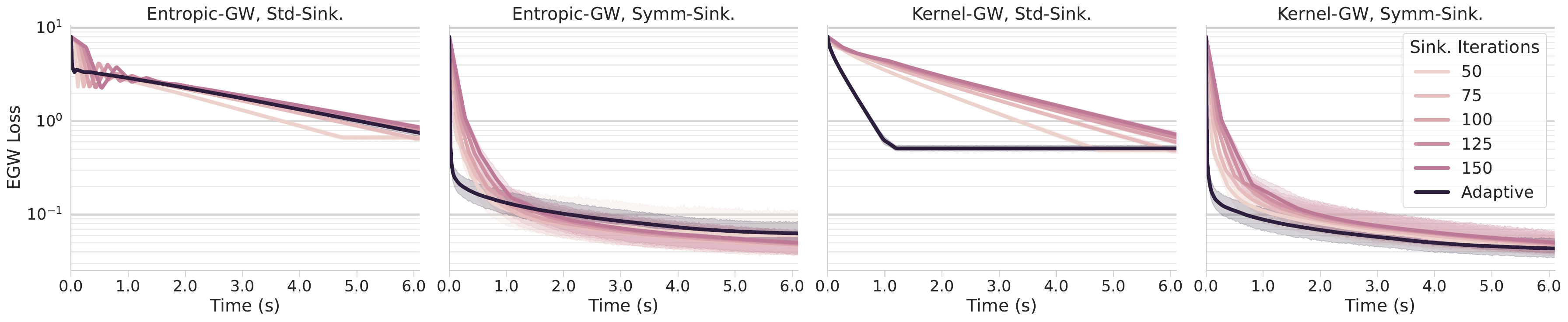} 
\caption{2D Gaussian distributions.}
\label{subfig:cvg_gaussian2d}
\end{subfigure}
\caption{Convergence curves on Gaussian samples.}
\label{fig:convergence_gaussians}
\end{figure*}

\paragraph{Distinct Input Distributions.} We finally identify a scenario where our method is less effective than the baseline.
When the input shapes differ significantly (\cref{fig:convergence_asymmetric}),
particularly when one input is highly symmetric (e.g., \cref{subfig:cvg_tetraspheres}), symmetrized \textsc{Kernel-GW} performs slightly worse than standard \textsc{Entropic-GW} and symmetrization generates oscillations for small Sinkhorn iterations.
However, our method remain convergent for sufficiently large iteration numbers, and adaptive Sinkhorn still succeeds to converge to the true optimum.
Therefore, our algorithm remains competitive even in this scenario.

\begin{figure*}[h!]
\centering
\begin{subfigure}{\linewidth}
\includegraphics[width=\textwidth]{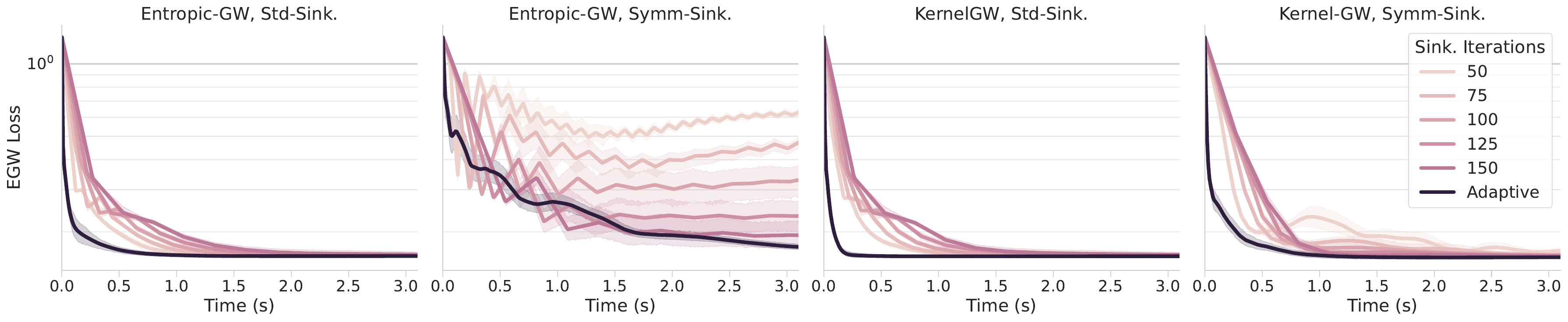} 
\caption{Tetrahedron and Rectangular Cuboid.}
\label{subfig:cvg_tetracubes}
\end{subfigure}
\begin{subfigure}{\linewidth}
\includegraphics[width=\textwidth]{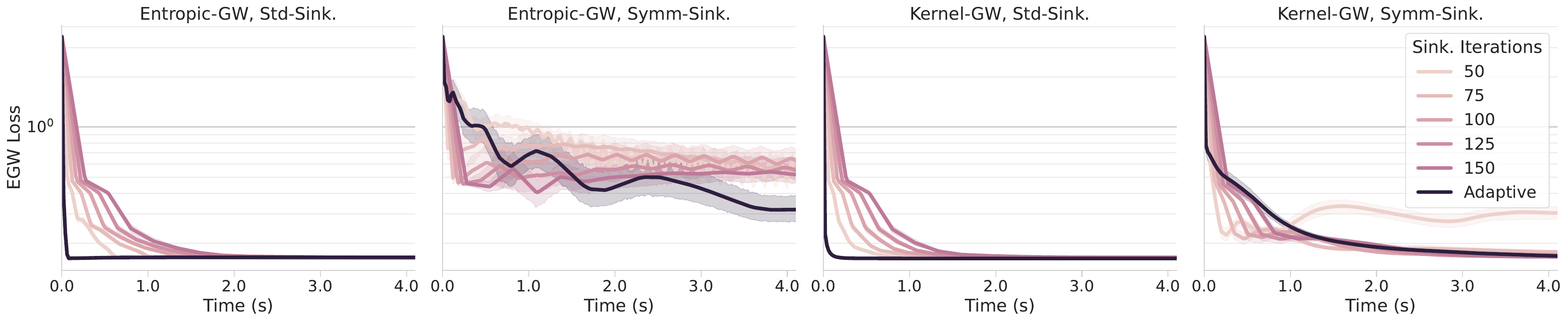} 
\caption{Tetrahedron and Sphere.}
\label{subfig:cvg_tetraspheres}
\end{subfigure}
\caption{Convergence curves on distinct surface distributions.}
\label{fig:convergence_asymmetric}
\end{figure*}

\subsection{Complementary Figures}

\paragraph{Benchmark of EGW Solvers.} \cref{fig:experiments_runtimebyn} displays the convergence times of our solvers and the main baselines as a function of the input size $\NN=\MM$.
Results were obtained by sampling point clouds uniformly from the horse shape of the SMAL dataset, and highlight the significant performance improvement of our kernel methods over \textsc{Entropic-GW} and \textsc{Quadratic-LowRank-GW}: \textsc{Kernel-GW} and \textsc{CNT-GW} achieve a $2-4 \times$ speed-up over their direct competitors, while \textsc{Multiscale-GW} is an order of magnitude faster.
\cref{fig:benchmark_onlyerr} shows that all methods output transport plans with equivalent losses, which confirms that the differences of convergence speeds are solely due to solver performances and not to a difference of solution quality \footnote{For memory reasons, the true EGW losses could not be computed for $\NN > 10^4$, and we relied on solver-specific approximations: this explains the gap between \textsc{CNT-GW}/\textsc{Multiscale-GW} and \textsc{Quadratic-LowRank-GW} in this regime.}.

\paragraph{Impact of the Embedding Dimension on the Quality of the Approximation.} 
\cref{fig:experiments_approxdim} displays the relative loss differences between true optimal plans and approximate ones obtained with \textsc{CNT-GW} and \textsc{Quadratic-LowRank-GW}, with $\NN = 2,000$.
These results show that \textsc{CNT-GW} maintains an approximation quality equivalent to that of \textsc{LowRank-GW}, demonstrating that our computational speed-ups do not compromise the quality of the output.
They also imply that approximation errors drop quickly as the embedding dimension increases, $\DD = 20$ being sufficient to reach a $1 \%$ error on the exact EGW loss.

\begin{figure}[ht]
\centering
    \begin{subfigure}[b]{0.33\linewidth}
        \centering
        \includegraphics[width=\linewidth]{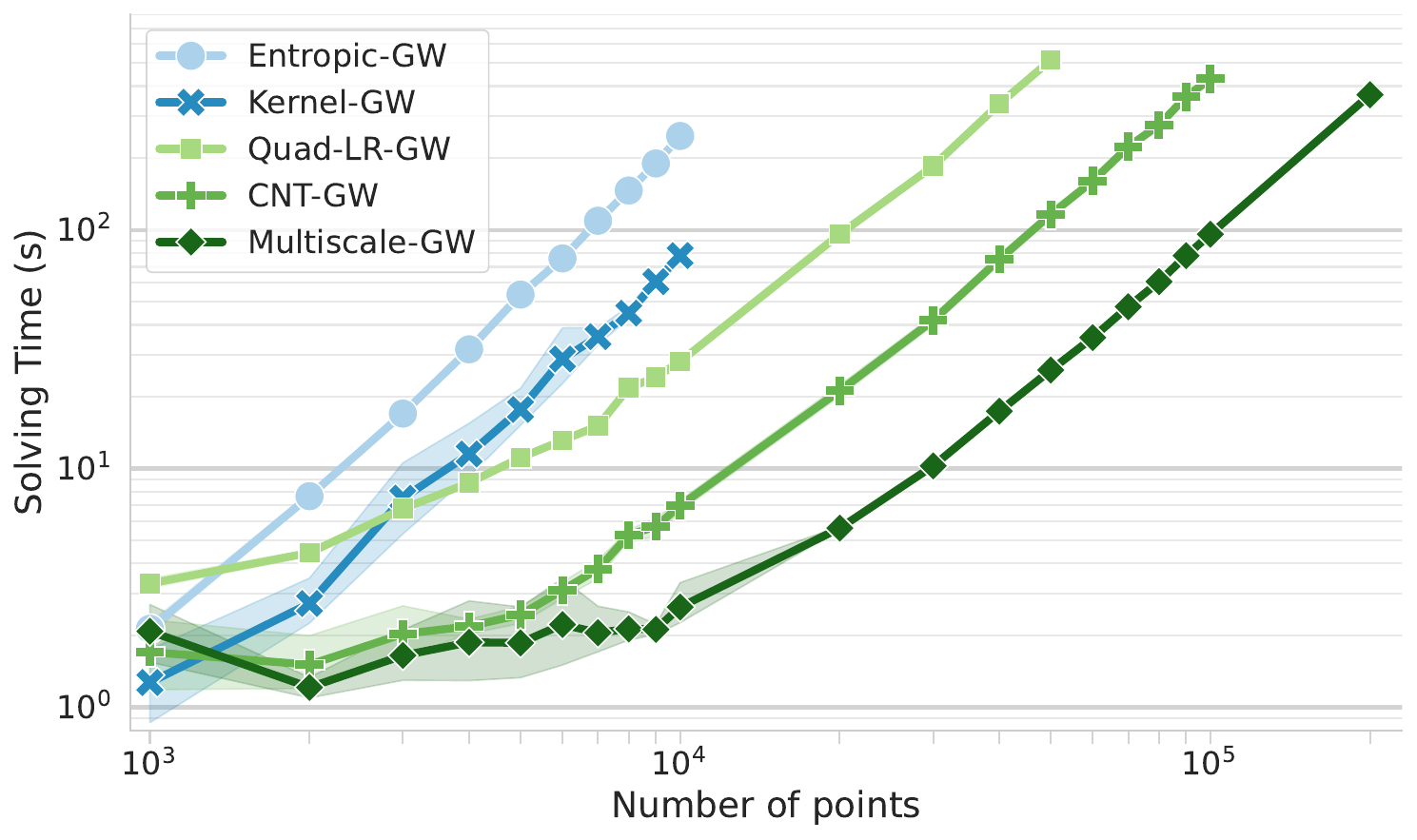}
        \caption{Convergence times.}
        \label{fig:experiments_runtimebyn}
    \end{subfigure}
    \hfill
    \begin{subfigure}[b]{0.33\linewidth}
        \centering
        \includegraphics[width=\linewidth]{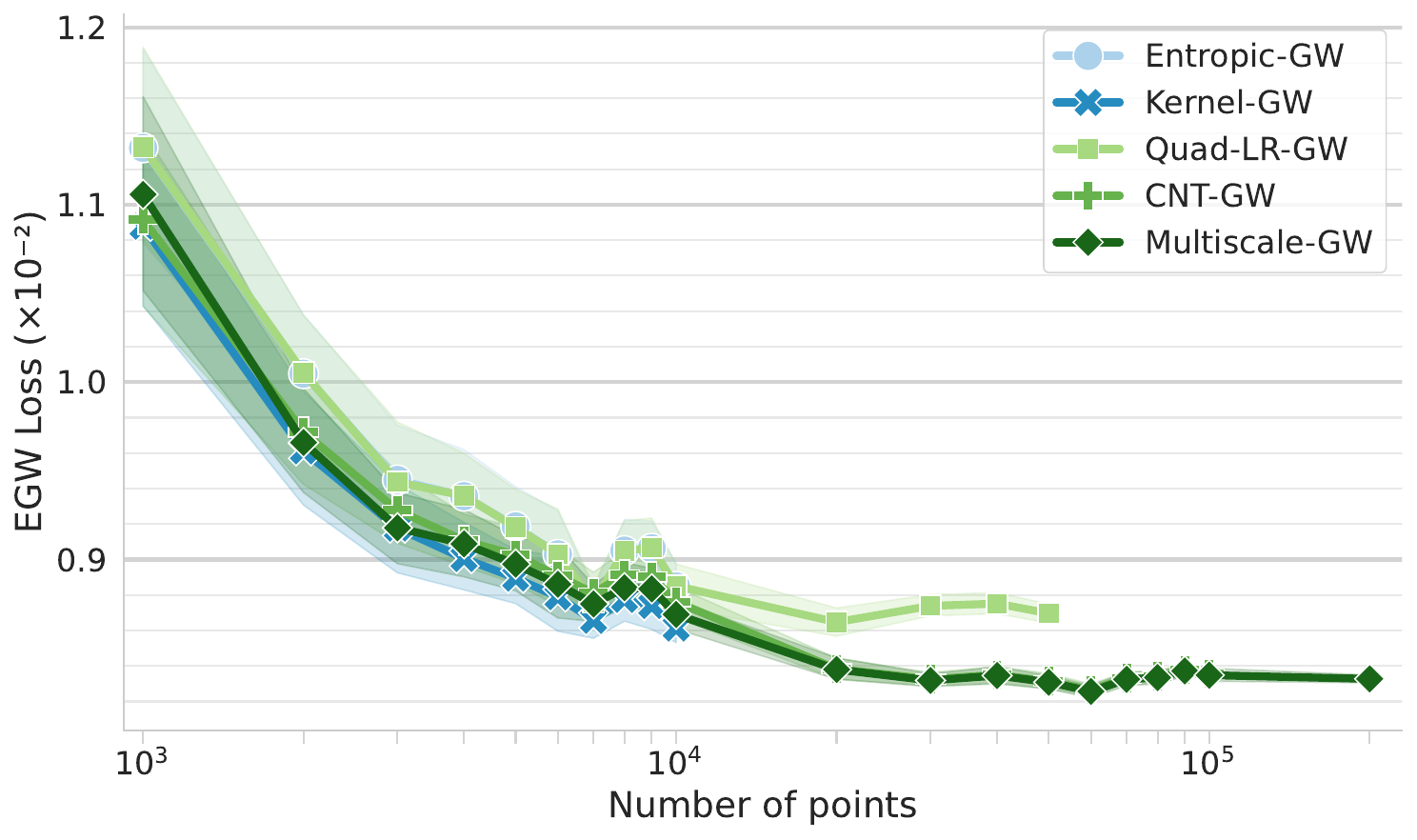}
        \caption{EGW loss of output plans.}
        \label{fig:benchmark_onlyerr}
    \end{subfigure}
    \hfill
    \begin{subfigure}[b]{0.33\linewidth}
            \centering
            \includegraphics[width=\linewidth]{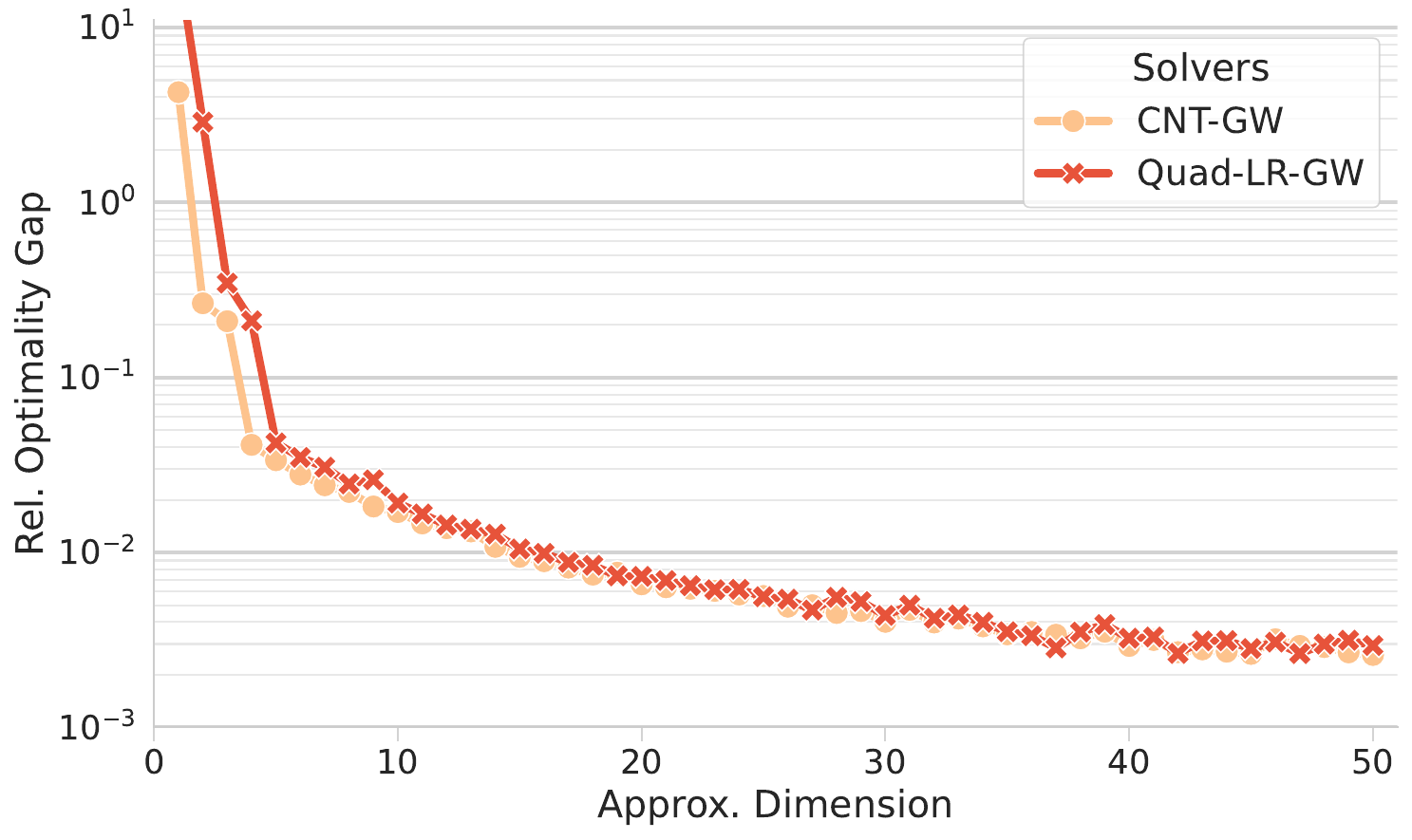}
            \caption{Approximation error in function of $D$.}
            \label{fig:experiments_approxdim}
        \end{subfigure}
    \caption{Benchmark of EGW solvers on inputs uniformly sampled on horse surfaces from the SMAL dataset (with $c_X,c_Y = \norm{\cdot}_{\R^3}$).}
    \label{fig:benchmark}
\end{figure}

%\paragraph{Shapes Used in our Benchmarks.} \cref{fig:shapes_article} represents the pairs of shapes used in the benchmark of \cref{table:benchmark}: \textsc{Horses}, \textsc{Hands}, \textsc{Dogs} and \textsc{Kids}. 
%The pair of \textsc{Faust} is represented in \cref{fig:introduction}.
%Ambient Euclidean costs were used for \textsc{Horses}, \textsc{Hands} and \textsc{Dogs}, whereas we used geodesic costs for \textsc{Kids} and \textsc{Faust}.  
%
%\input{resources/figures/shapes_article}

\paragraph{Visualization of transport plans.} In \cref{fig:hands_texture_comparisons}, we represent the optimal plans produced by the different solvers on the hands shapes (since \textsc{Multiscale-GW} is an acceleration of \textsc{CNT-GW}, we did not represent its output as it would be identical to \textsc{CNT-GW}).
We also visualize in \cref{fig:multiple_texture_comparisons} the results of \textsc{Quadratic-LowRank-GW} and \textsc{CNT-GW} on the other shapes of our benchmark, as they are the only methods able to scale to these inputs. 
We see that all results are qualitatively similar, confirming the quantitative trend observed in \cref{table:benchmark}.
Overall, the texture transfers are smooth and provide meaningful correspondences between the sources and targets.
Some artifacts can be noticed on the Kids and Vessels examples, on the elongate parts of the shapes: these are well-known in shape matching and require more advanced techniques in order to be solved.

\begin{figure*}[h]
\centering
\includegraphics[width=\linewidth]{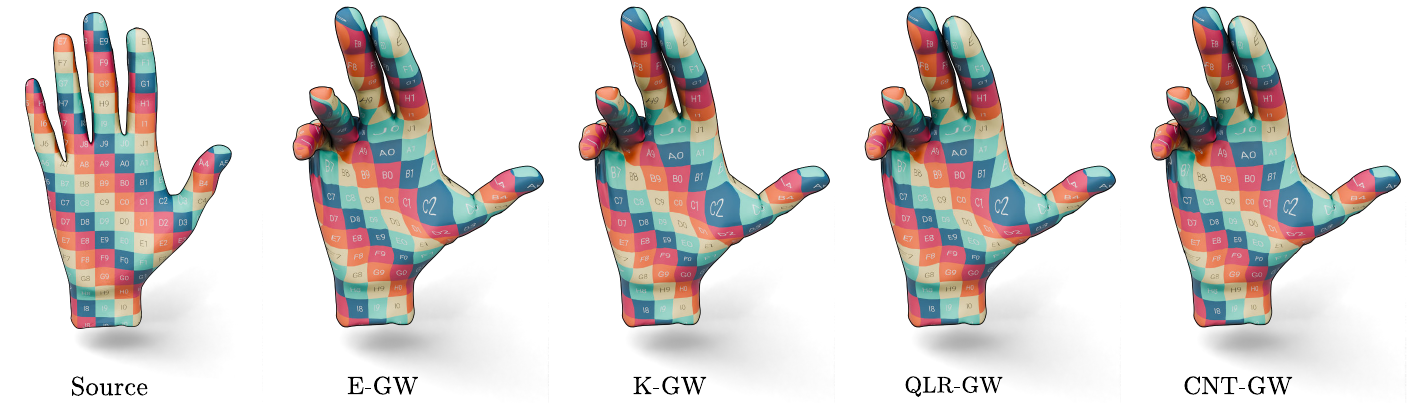} 
    \caption{GW-based texture transfer between hand shapes.}
    \label{fig:hands_texture_comparisons}
\end{figure*}

\begin{figure*}[ht]
\centering
\includegraphics[width=\linewidth]{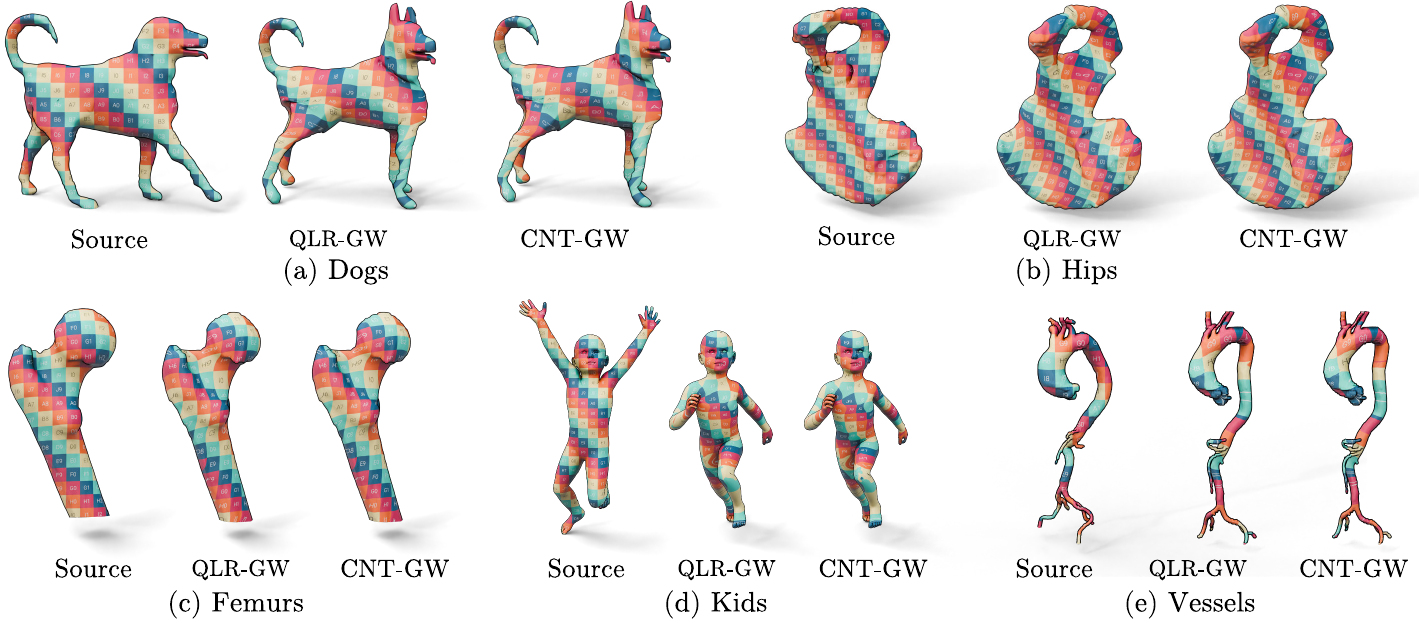} 
    \caption{GW-based texture transfer between the other shapes used in the benchmarks.}
    \label{fig:multiple_texture_comparisons}
\end{figure*}

\paragraph{Impact of Base Costs on EGW Geometry.} We explore the EGW geometries induced by different base costs by displaying a gradient flow between two simple shapes, as we use squared norms, Euclidean norms, root norms and exponential kernels with different radii.
The results are given in \cref{fig:flows_S_C}: more global costs tend to preserve the global ordering of the points during the deformation and create discontinuities to reverse the upper bar of the ``5'' shape.
On the other hand, local costs keep the continuity of the shape, and follow non-trivial paths to push the source shape onto the target one.
The CNT framework is sufficiently expressive to adjust the balance between these local and global-oriented behaviors.

\begin{figure*}[ht]
\centering
\includegraphics[width=\textwidth]{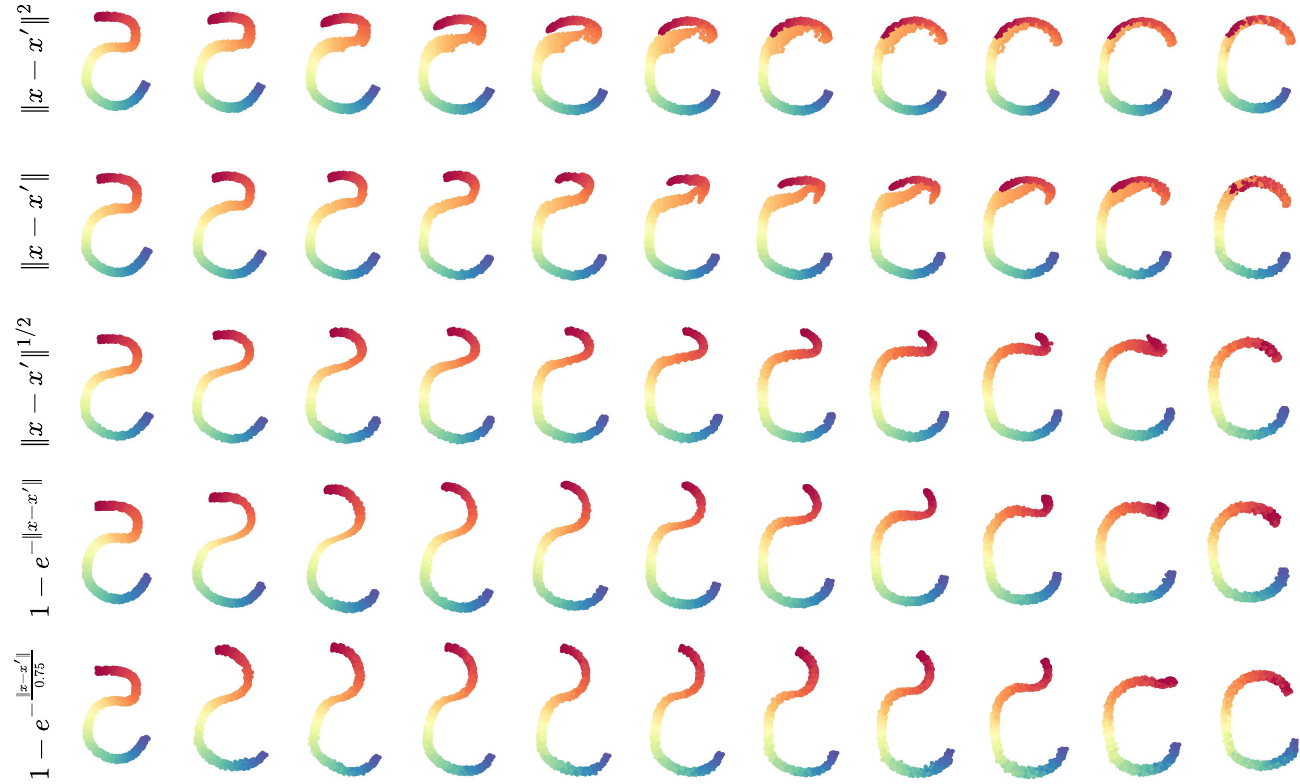} 
    \caption{Gradient flows $\alpha_{t+\delta t} = \alpha_{t} - \delta t \nabla \SGW_\varepsilon(\alpha_t, \beta)$ from a reversed ``5'' shape $\alpha_0$ to a ``C'' shape $\beta$, using different base costs.}
    \label{fig:flows_S_C}
\end{figure*}

In \cref{fig:flows_S_C_approx}, we also plot the flows obtained with approximate gradients (\cref{alg:egw_grad_kernelpca}) using a Kernel PCA dimension of $D = 50$. 
The results are similar to the exact case; however, we also observe several artifacts that highlight the limitations of dimension reduction for precise gradient computations. 
Note that these artifacts do not modify the asymptotic distribution of the point cloud; despite the approximation, all gradient flows eventually converge to the correct shape.

\begin{figure*}[t!]
\centering
\includegraphics[width=\textwidth]{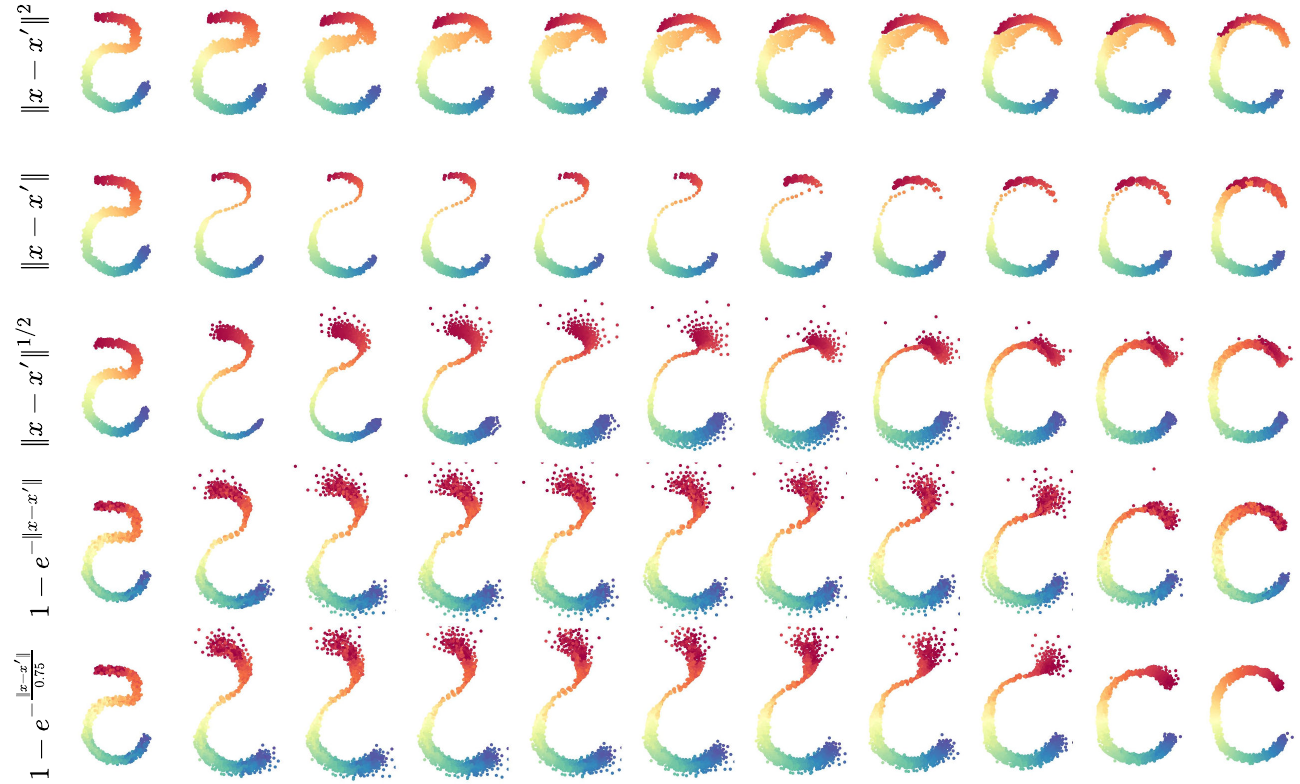} 
    \caption{Gradient flows with approximate gradients obtained via Kernel PCA ($D=50$).}
    \label{fig:flows_S_C_approx}
\end{figure*}

\section{Experiment Details}
\label{appendix:experiment_details}
\paragraph{Numerical Backend.} All solvers were implemented in the \texttt{PyTorch} framework using single precision float numbers (float32) and run on GPU \cite{paszke2017automatic}.
We used the \texttt{KeOps} package to reduce the memory footprint of computations whenever possible.
This includes the implementation of \textsc{CNT-GW} as well as the baselines \textsc{Quadratic-LowRank-GW}, \textsc{Dual-GW} and \textsc{Sampled-GW}.
We used truncated PCA decomposition to compute the low-rank approximation of cost matrices in \textsc{Quadratic-LowRank-GW}.
We systematically use ranks of $\DD = \EE = 20$ for \textsc{Quadratic-LowRank-GW}, matching the choices of dimensions made for \textsc{CNT-GW}. 
We used \texttt{KeOps} to compute this truncated PCA, as well as for Kernel PCA in \textsc{CNT-GW}.
In our main benchmarks, we implemented our algorithms \textsc{Kernel-GW}, \textsc{CNT-GW} and \textsc{Multiscale-GW} using Symmetrized Sinkhorn (\cref{alg:sinkhorn_symm}).
We implemented \textsc{Entropic-GW} and \textsc{Quadratic-LowRank-GW} using standard Sinkhorn (\cref{alg:sinkhorn}), which corresponds to the original implementation as introduced by their authors and avoids all numerical instabilities introduced by symmetrized Sinkhorn to these algorithms  (\cref{subsection:appendix:solver_cvg}).

\paragraph{Libraries.} We also used the \texttt{NumPy} \cite{numpy} and \texttt{SciPy} \cite{2020SciPy-NMeth} libraries for scientific computing.
We relied on \texttt{Matplotlib} \cite{matplotlib}, \texttt{Seaborn} \cite{Waskom2021} and Blender for visualizations.

\paragraph{Computation of Geodesic Embeddings.} To obtain Euclidean approximations of geodesic distances, we implemented the method of  \cite{panozzo2013weighted}.
We decimated input meshes, reducing the number of vertices to $\NN=2,000$.
We computed the exact pairwise geodesic distances on the decimated meshes, and embedded the coarse vertices in $\R^{8}$ using Multi Dimensional Scaling.
We finally interpolated the $8$-dimensional Euclidean coordinates of all vertices by solving a quadratic interpolation problem on the original mesh.
The Euclidean norm $\norm{\cdot}_{\text{geodesic}}$ on this Euclidean embedding approximates true geodesic distances between all pairs of vertices.

\paragraph{Other Implementation Choices.} For \textsc{Multiscale-GW}, we systematically choose a coarsening ratio of $\rho = 0.1$.
Unless specified, we initialize all solvers with the trival transport plan $\pi^0 = \alpha \otimes \beta$ (or $\Gamma^0 = 0$ for dual-based solvers).
We always take $\alpha$ and $\beta$ as uniformly weighted point clouds (i.e. $a_i = 1/\NN$ and $b_j = 1/\MM$ for all $i,j$).

\paragraph{Transport plan visualizations.} To visualize transport plans $\pi$ between 2D or 3D shapes, we either transfer color or texture from the source to the target.
For color transfer, we assign an RGB value $C_\X(x) \in \mathbb{R}^3$ to each point $x$ in the source shape. The color at each target point is then defined as the weighted average:
\begin{equation*}
    \text{For all } y \in \Y, \quad C_\Y(y) = \int C_\X(x) \frac{\diff \pi(x,y)}{\diff \alpha(x) \diff \beta(y)} \diff \alpha(x).
\end{equation*}
Texture transfer follows a similar idea and rely on the notion of \emph{UV mapping} in 3D graphics. Each source point $x$ is assigned UV coordinates $UV_\X(x) \in \mathbb{R}^2$, which are then transported to the target via:
\begin{equation*}
    \text{For all } y \in \Y, \quad UV_\Y(y) = \int UV_\X(x) \frac{\diff \pi(x,y)}{\diff \alpha(x) \diff \beta(y)} \diff \alpha(x).
\end{equation*}
Textures are represented as 2D images, with UV coordinates mapping each vertex of a 3D shape to a location in the image. 
During rendering, these coordinates are interpolated across each triangle of the mesh to produce the final textured surface.

\end{document}